\documentclass[11pt,reqno]{amsart}
\usepackage{geometry}
\usepackage{amssymb, amsmath,amsthm, amsfonts}
\usepackage[colorlinks=true]{hyperref}
 \hypersetup{
     colorlinks=true,
     linkcolor=blue,
     filecolor=blue,
     citecolor = Magenta,      
     urlcolor=cyan,
     }
\usepackage[dvipsnames]{xcolor}
\usepackage{url}
\usepackage{graphicx}
\usepackage{enumerate}
\usepackage{epstopdf}
\usepackage{subfig}
\usepackage{mathtools}

\usepackage{ifthen}
\usepackage{tikz}
\usepackage{cite}
\usepackage{verbatim}
\usepackage{pifont}
\usepackage{bm}  
\usepackage{stackengine}
\usepackage{bbm}

\allowdisplaybreaks
\usepackage{accents}
\patchcmd{\subsubsection}{\itshape}{\bfseries}{}{}

\usepackage{xpatch}
\makeatletter
\AtBeginDocument{\xpatchcmd{\@thm}{\thm@headpunct{.}}{\thm@headpunct{}}{}{}}
\makeatother
\makeatletter
\let\save@mathaccent\mathaccent
\newcommand*\if@single[3]{%
  \setbox0\hbox{${\mathaccent"0362{#1}}^H$}%
  \setbox2\hbox{${\mathaccent"0362{\kern0pt#1}}^H$}%
  \ifdim\ht0=\ht2 #3\else #2\fi
  }
\newcommand*\rel@kern[1]{\kern#1\dimexpr\macc@kerna}
\newcommand*\widebar[1]{\@ifnextchar^{{\wide@bar{#1}{0}}}{\wide@bar{#1}{1}}}
\newcommand*\wide@bar[2]{\if@single{#1}{\wide@bar@{#1}{#2}{1}}{\wide@bar@{#1}{#2}{2}}}
\newcommand*\wide@bar@[3]{%
  \begingroup
  \def\mathaccent##1##2{%
    \let\mathaccent\save@mathaccent
    \if#32 \let\macc@nucleus\first@char \fi
    \setbox\z@\hbox{$\macc@style{\macc@nucleus}_{}$}%
    \setbox\tw@\hbox{$\macc@style{\macc@nucleus}{}_{}$}%
    \dimen@\wd\tw@
    \advance\dimen@-\wd\z@
    \divide\dimen@ 3
    \@tempdima\wd\tw@
    \advance\@tempdima-\scriptspace
    \divide\@tempdima 10
    \advance\dimen@-\@tempdima
    \ifdim\dimen@>\z@ \dimen@0pt\fi
    \rel@kern{0.6}\kern-\dimen@
    \if#31
      \overline{\rel@kern{-0.6}\kern\dimen@\macc@nucleus\rel@kern{0.4}\kern\dimen@}%
      \advance\dimen@0.4\dimexpr\macc@kerna
      \let\final@kern#2%
      \ifdim\dimen@<\z@ \let\final@kern1\fi
      \if\final@kern1 \kern-\dimen@\fi
    \else
      \overline{\rel@kern{-0.6}\kern\dimen@#1}%
    \fi
  }%
  \macc@depth\@ne
  \let\math@bgroup\@empty \let\math@egroup\macc@set@skewchar
  \mathsurround\z@ \frozen@everymath{\mathgroup\macc@group\relax}%
  \macc@set@skewchar\relax
  \let\mathaccentV\macc@nested@a
  \if#31
    \macc@nested@a\relax111{#1}%
  \else
    \def\gobble@till@marker##1\endmarker{}%
    \futurelet\first@char\gobble@till@marker#1\endmarker
    \ifcat\noexpand\first@char A\else
      \def\first@char{}%
    \fi
    \macc@nested@a\relax111{\first@char}%
  \fi
  \endgroup
}
\makeatother
\DeclareSymbolFont{bbold}{U}{bbold}{m}{n}
\DeclareSymbolFontAlphabet{\mathbbold}{bbold}

\newtheorem{theorem}{Theorem}
\newtheorem{cor}{Corollary}
\newtheorem{lemma}{Lemma}
\newtheorem{prop}{Proposition}

\newtheorem{remark}{Remark}

\newtheorem{definition}{Definition}

\newtheorem{assump}{Assumption}

\newcommand{\norm}[1]{\left\Vert#1\right\Vert}
\newcommand{\abs}[1]{\left\vert#1\right\vert}

\newcommand{\cA}{\mathcal{A}}
\newcommand{\cB}{\mathcal{B}}

\newcommand{\cE}{\mathcal{E}}
\newcommand{\cF}{\mathcal{F}}
\newcommand{\cG}{\mathcal{G}}

\newcommand{\cI}{\mathcal{I}}

\newcommand{\cP}{\mathcal{P}}

\newcommand{\cT}{\mathcal{T}}

\newcommand{\cX}{\mathcal{X}}

\newcommand{\BB}{\mathbb{B}}

\newcommand{\EE}{\mathbb{E}}

\newcommand{\NN}{\mathbb{N}}

\newcommand{\PP}{\mathbb{P}}

\newcommand{\RR}{\mathbb{R}}

\newcommand{\ZZ}{\mathbb{Z}}

\newcommand{\supp}{\mathrm{spt}}

\newcommand{\ind}{\mathbbm 1}

\newcommand{\X}{\mathcal{X}}

\newcommand{\kl}[2]{\mathsf{D}_\mathsf{KL}\left(#1\middle\|#2\right)}

\newcommand{\klest}[3]{\hat{\mathsf{D}}_{\mathsf{KL},#3}\left(#1,#2\right)}

\newcommand{\rendiv}[3]{\mathsf{D}_{#3}\left(#1\middle\|#2\right)}
\newcommand{\renest}[4]{\hat{\mathsf{D}}_{#3,#4}\left(#1,#2\right)}

\newcommand{\ubar}[1]{\underaccent{\bar}{#1}}

\definecolor{cblue}{rgb}{0.16, 0.32, 0.75}

\def\h2{\tilde h}

\def\hm1{\hat h_{-1}}

\begin{document}
\title[Deviation Inequalities for R\'{e}nyi Divergence  Estimators]{Deviation Inequalities for R\'{e}nyi 
Divergence  Estimators via Variational Expression}
\thanks{S. Sreekumar is supported in parts by the Centre National de la Recherche Scientifique (CNRS) funding and the CentraleSup\'elec$\mspace{1 mu}$-$\mspace{1 mu}$L2S funding. K. Kato is partially supported by NSF grant DMS-2413405.}

\author[S. Sreekumar]{Sreejith Sreekumar}
\address[S. Sreekumar]{L2S, CNRS, CentraleSupélec,  University of Paris-Saclay, France}
\email{sreejith.sreekumar@centralesupelec.fr}

\author[K. Kato]{Kengo Kato}
\address[K. Kato]{
Department of Statistics and Data Science, Cornell University, USA
}
\email{kk976@cornell.edu}

\begin{abstract}
R\'enyi divergences play a pivotal role in information theory, statistics, and machine learning. While several estimators of these divergences have been proposed in the literature with their consistency properties established and minimax convergence rates quantified, existing accounts of probabilistic bounds governing the estimation error are relatively underdeveloped.  Here, we make progress in this regard by establishing exponential deviation inequalities for smoothed plug-in estimators and neural estimators by relating the error to an appropriate empirical process and leveraging tools from empirical process theory. In particular, our approach does not require the underlying distributions to be compactly supported or have densities bounded away from zero, an assumption prevalent in existing results. The deviation inequality also leads to a one-sided concentration bound from the expectation, which is useful in random-coding arguments over continuous alphabets in information theory with potential applications to physical-layer security.  As another concrete application, we consider a hypothesis testing framework for auditing R\'{e}nyi differential privacy using the  neural estimator as a test statistic and obtain  non-asymptotic performance guarantees for such a test. 
\end{abstract}
\maketitle
\section{Introduction}
R\'enyi divergences \cite{Renyi-60} play a fundamental role in information theory, statistics, machine learning, and related fields. These divergences as well as information measures derived from them, such as mutual information and entropy,   characterize fundamental limits in information-theoretic coding problems. For instance, they characterize the error-exponents in hypothesis testing, the capacity of noisy channels, ultimate limits in data compression and rate-distortion theory, to name a few (see e.g., \cite{CoverThomas}). Further, they appear as discrepancy measures or as optimization metrics in applications. R\'enyi divergence is a functional of the underlying probability distributions, which are often unknown in practice. In such scenarios, the problem of estimating R\'{e}nyi divergences from samples arises, for which various estimators have been proposed 
in the literature, such as those based on plug-in, kernel, and variational methods (see Section \ref{Sec:relwork} below).  While consistency and minimax error rates for several of these estimators have been established, an understanding of the probability law governing the estimation error, which is pivotal for a complete statistical account of performance, is limited.

In this work, we make progress in this regard by establishing exponential deviation inequalities for two classes of R\'{e}nyi divergence estimators, namely Gaussian-smoothed plug-in estimator and  neural estimator. To this end, we utilize known variational expressions for these divergences (see \eqref{eq:varformkl} and \eqref{eq:rendonvarform} below) \cite{Donsker-1975,Nguyen-2010,Birell-2021}. 
Our approach relies on relating the deviation to the suprema of empirical processes indexed by suitable function classes, and using concentration inequalities for the suprema to arrive at the desired conclusion. Our results go beyond the framework of discrete alphabets and compactly supported  distributions with densities bounded away from zero, as required in existing works (see e.g., \cite{Singh-Poczos-2014,pmlr-v32-singh14,sreekumar2021neural}). Specifically, we obtain deviation inequalities for Gaussian-smoothed empirical measures (based on samples) which satisfy neither of these criteria, under the assumption that the unsmoothed distribution is compactly supported or sub-Gaussian. Further, under an additional assumption that the densities have bounded ratios and belong to an appropriate Lipschitz class, we are also able to treat the deviation of the estimate from the unsmoothed value. We also obtain exponential inequalities for neural estimators of R\'{e}nyi divergences and discuss their applications to obtaining the first complete non-asymptotic performance guarantee for auditing R\'{e}nyi differential privacy (DP) \cite{jagielski2020auditing,domingo2022auditing} using the neural estimator as a test statistic.
\subsection{Estimation of R\'{e}nyi  divergences}
For $\alpha \in (0,1) \cup (1,\infty)$, the R\'{e}nyi divergence of order $\alpha$ between Borel probability measures $\mu$ and $\nu$ on $\RR^d$ is
\begin{align}
    \rendiv{\mu}{\nu}{\alpha}:=\begin{cases}
        \frac{1}{\alpha-1} \log \big(\int_{\RR^d} p_{\mu}^{\alpha}p_{\nu}^{1-\alpha}d\rho\big), & \mu \ll \nu \mbox{ or }\mu \not\perp \nu \mbox{ for } \alpha \in (0,1),\\
        \infty, &\mbox{otherwise},
    \end{cases} \notag
\end{align}
where $ \rho$ is a sigma-finite positive measure such that $\mu,\nu \ll \rho$ and $\mu \not\perp \nu$ denotes that $\mu$ and $\nu$ are non-singular. 
The special cases of R\'{e}nyi divergences of order $\alpha \in \{0,1,\infty\}$ are defined as the respective limits of the above expression. The case $\alpha=1$ corresponds to the Kullback-Leibler (KL) divergence and $\alpha=\infty$ is the max-divergence, given by 
\begin{align}
\kl{\mu}{\nu}:=\begin{cases}
       \int_{\RR^d} p_{\mu} \log \frac{p_{\mu}}{p_{\nu}} d\rho, & \mu \ll \nu,\\
        \infty, &\mbox{otherwise},
    \end{cases} \notag
\end{align}
and 
\begin{align}
\rendiv{\mu}{\nu}{\max}:= \begin{cases}
\log \norm{\frac{p_{\mu}}{p_{\nu}}}_{\infty,\rho} & \mu \ll \nu,\\
        \infty, &\mbox{otherwise},
    \end{cases} \notag
\end{align}
respectively ($\norm{\mspace{1 mu}\cdot \mspace{1 mu}}_{\infty,\rho}$ denotes the usual norm in $L^{\infty}(\rho)$, see Section \ref{Sec:notation}). R\'{e}nyi divergences are sensitive to support mismatch between the underlying distributions and are unbounded in general. Commonly proposed estimators for continuous distributions overcome this issue via some form of regularization. For instance, plug-in and kernel methods utilize a regularized version of the density constructed from samples by shifting and/or scaling of an appropriate kernel function such as a Gaussian density or a bump function. In general, R\'{e}nyi divergence is also a non-smooth functional of the underlying densities. Hence, the  classical estimators of \cite{Bickel-Ritov-88,Birge-Massart-95,Laurent-96,Kerkyacharian-96}, which are tailored to smooth functionals, do not suffice to achieve optimal convergence rates over a class of densities with unbounded support.

R\'{e}nyi divergences can be expressed in a variational form as the supremum over a function class of a  functional that involves expectations with respect to (w.r.t.) the underlying distributions.  The following variational (or dual) expressions are well-known\footnote{The expression in \eqref{eq:donvarexpkl} is known as the Donsker-Varadhan variational expression.  The second expression given in \eqref{eq:donvarexpkl} which will be useful for our purposes can be shown via an application of the Legendre-Fenchel duality applied to the convex map $x \mapsto x \log x$ (see \cite{Nguyen-2010}).} for KL divergence \cite{Donsker-1975,Nguyen-2010}: 
\begin{subequations}\label{eq:varformkl}
\begin{align}
 \kl{\mu}{\nu}&=\sup_{f}\EE_{\mu}[f]-\log \EE_{\nu}[e^f] \label{eq:donvarexpkl} \\   
 &=\sup_{f}\EE_{\mu}[f]-\EE_{\nu}[e^f]+1, \label{eq:varexpkl}
\end{align}
\end{subequations}
where the suprema above are over all measurable functions  such that the second expectations in \eqref{eq:donvarexpkl} and \eqref{eq:varexpkl} are finite\footnote{Alternatively, one may consider the supremum to be over all measurable functions with the interpretation, $\infty-\infty=-\infty$ and $-\infty+\infty=-\infty$, or even restrict the supremum to be over all bounded measurable functions without loss of generality.}. 
 The suprema in \eqref{eq:varformkl} are achieved when $\mu \ll \nu$ by $f_{\mu,\nu}^{\star}=\log \frac{d \mu}{d \nu}$.   For R\'{e}nyi divergences of order $\alpha \in (0,1) \cup (1,\infty)$,  the following  variational expression was shown in   \cite{Birell-2021}:
\begin{align}
    \rendiv{\mu}{\nu}{\alpha}=\sup_{f} \frac{\alpha}{\alpha-1} \log \left(\int_{\RR^d} e^{(\alpha-1)f} d\mu\right)-\log \left(\int_{\RR^d} e^{\alpha f} d\nu\right), \label{eq:rendonvarform}
\end{align}
where the supremum is over all measurable functions such that the second integral is finite. The supremum is achieved when $\mu \ll \nu$ and $\left(\frac{d\mu}{d\nu}\right)^{\alpha} \in L_1(\nu)$ by $f_{\mu,\nu}^{\star}$ \cite{Birell-2021}.  It is known that the value of the expressions on the right hand side (RHS) of \eqref{eq:donvarexpkl},  \eqref{eq:varexpkl}, and \eqref{eq:rendonvarform} are unchanged if $f$ is restricted to the set of all bounded continuous functions (see e.g., \cite{Birell-2021}). An alternative expression similar in spirit to \eqref{eq:varexpkl} also holds for  $\rendiv{\mu}{\nu}{\alpha}$ (see Lemma \ref{lem:varexpren} in Appendix). 

The aforementioned variational expressions have been leveraged recently to construct an estimator of R\'{e}nyi divergences and $f$-divergences \cite{belghazi2018,sreekumar2021neural,Birell-2021}.  The  estimate is obtained simply by approximating the expectations in the dual form by that w.r.t. the empirical measures based on samples and maximizing over an appropriate class of functions, e.g., the output of a neural net. Such estimators have excellent adaptability and scalability properties in practice, and have been shown to be consistent under mild assumptions on the underlying distributions. In fact, such estimators are minimax optimal achieving a parametric rate of convergence under sufficient smoothness of densities, provided the neural network is scaled appropriately according to sample size\cite{sreekumar2021neural}. Additionally, the specific form of the variational formula as well as the resulting estimator lends itself naturally to an analysis via empirical process theory. Here, we utilize these aspects to derive exponential deviation inequalities for Gaussian-smoothed plug-in estimator and neural estimator via  powerful concentration inequalities for empirical process suprema. 

Our approach to deriving the deviation inequality for Gaussian-smoothed  plug-in estimator relies on identifying primitive conditions on the distributions which allow the estimation error to be controlled by the suprema of empirical processes with good concentration properties. In this regard, we show that when  $\mu$ and $\nu$ satisfy certain regularity conditions, which are in particular satisfied in the compactly supported or sub-Gaussian case, it is possible to relate the desired deviation to suprema of empirical processes indexed by a smooth  function class with an integrable envelope\footnote{An envelope of a class $\cF$ of functions on $\cX \subseteq \RR^d$ is a function $F:\cX \rightarrow [0,\infty]$ such that $\sup_{f \in \cF} \abs{f(x)} \leq F(x)$ for all $x \in \cX$.}. The integrability of the envelope then translates to a concentration property for the suprema as quantified by a version of Talagrand's concentration inequality given in \cite{BOUSQUET-2002} (or a Fuk-Nagaev type inequality\cite{Fuk-Nagaev,Adamczak-2008,Adamczak-2010} in the general setting). 
In order to apply these inequalities,  we require an upper bound  on the covering entropy (see Definition \ref{cov-pack-num}) of a Gaussian-smoothed function class  with an envelope  $F(x)$ that scales at most as some power of $\norm{x}$.  
As for neural estimators, we again leverage Talagrand's concentration inequality to obtain the desired exponential inequality. 
 \subsection{Contributions}
\noindent Below, we summarize  our main contributions. Let $\hat \mu_n$ and $\hat \nu_n$ denote the standard empirical probability  measures constructed from $n$ independent and identically distributed (i.i.d.) samples generated according to $\mu$ and $\nu$, respectively. Let $*$ denote the convolution operation and $\gamma_{\sigma}$ be the centered Gaussian measure (or distribution) on $\RR^d$ with covariance matrix $\sigma^2 I_d$. Then, we show the following$\mspace{1 mu}$\footnote{Here, $g=O_z(f)$ and $g=\Omega_z(f)$, respectively, signifies that $g \leq c_z f$ and $g \geq c_z f$ for some constant $c_z>0$ that depends on $z$.} deviation inequalities including explicitly bounding the constant pre-factors appearing in these expressions up to universal constants: 
 \begin{enumerate}[(i)]
    \item  For all $\mu,\nu$ whose support is contained in a $d$-dimensional ball of radius $r$,  we prove in Theorem \ref{Thm:KL_compsupp_sg} that 
    \begin{align}
&\PP \left(  \abs{\kl{\hat \mu_n*\gamma_{\sigma}}{ \hat\nu_n*\gamma_{\sigma}}-\kl{\mu*\gamma_{\sigma}}{\nu*\gamma_{\sigma}}} \geq  \Omega_{r,d,\sigma}\mspace{-4 mu}\left(n^{-\frac 12} +z\right)\right) \leq 2e^{-\left(nz^2 \wedge nz\right)}.\label{eq:devcont1} 
\end{align} 
Combining this with an approximation result (Proposition \ref{prop:smoothRen-stability}) for R\'enyi divergences, we obtain a deviation inequality (see Corollary \ref{Cor:Rendivsmoothrate}) for $\abs{\kl{\hat \mu_n*\gamma_{\sigma}}{\nu*\gamma_{\sigma}}-\kl{\mu}{\nu}}$,  under a boundedness and  Lipschitz continuity assumption on the densities of $\mu,\nu$.  An extension to the case of sub-Gaussian $\mu$ and $\nu$ in the one-sample case,  i.e. $\hat \nu_n=\nu$ for all $n$, is also obtained.

\item For R\'{e}nyi divergences of order $\alpha \in (0,1) \cup (1,\infty)$, we establish in  Theorem \ref{Thm:compactlysuppRen} that  for $\mu,\nu$ as above  
    \begin{align}
& \PP \left(  \abs{\rendiv{\hat \mu_n*\gamma_{\sigma}}{\hat \nu_n*\gamma_{\sigma}}{\alpha}-\rendiv{\mu*\gamma_{\sigma}}{\nu*\gamma_{\sigma}}{\alpha}} \geq  \Omega_{r,d,\sigma,\alpha}\big(n^{-\frac 12} +z\big)\right)\leq 2e^{-\left(nz^2 \wedge nz\right)}. \label{eq:devcont2}
\end{align}
\item For neural estimators of R\'{e}nyi divergences  with a suitably chosen bounded neural net class $\cG$, which can approximate  $f^{\star}_{\mu,\nu}$ up to an error $\delta$ uniformly on the support of $\rho$ and whose uniform $L^2-$covering entropy numbers satisfies an integrability condition, we show in Theorem \ref{Thm:devineqren} that 
\begin{align}   &\PP\left(\abs{\renest{\hat \mu_n}{\hat \nu_n}{\alpha}{\cG}-\rendiv{\mu}{\nu}{\alpha}} \geq \big(M^{2\alpha}+M^{2\abs{\alpha-1}}\big)\delta+  \Omega_{d,M,\alpha}(n^{-\frac 12} +z) \right)\leq  2 e^{-(nz^2 \wedge nz)}. \label{eq:devcont3}
    \end{align}  
    Here,  $M$ is a parameter which upper bounds the ratio of densities of $\mu$ and $\nu$.    
\end{enumerate}
As our key technical contributions, we derive a more general lemma (Lemma \ref{Lem:concent-KL}) that provides deviation inequalities for smoothed plug-in estimator of KL divergence under certain regularity conditions involving the smoothing kernel and underlying probability measures. Analogous lemma can also be stated for R\'{e}yi divergences, however, we omit this to avoid repetition. A key tool required to apply Lemma \ref{Lem:concent-KL} to prove Theorem \ref{Thm:KL_compsupp_sg} and \ref{Thm:compactlysuppRen} is a bound on the covering entropy of a Gaussian-smoothed function class whose envelope scales as an affine or quadratic function of $\norm{x}_2$. To this end, we obtain a uniform upper bound (Lemma \ref{lem:bndholdernorm}) on the H\"{o}lder norm of functions in this class, and use the covering entropy estimate for H\"{o}lder class given   in \cite[Theorem 2.7.1]{AVDV-book}. In fact, we apply a slightly refined version of this result post quantifying the dependence of dimension and support (see Theorem \ref{Thm:coventconst}). Leveraging \cite[Theorem 2.1]{Bartlett-2005}, we also provide a statistical version (Theorem \ref{Thm:KLdev-symmradexp}) of \eqref{eq:devcont1} in terms of Rademacher complexity of the Gaussian-smoothed function class, which is in general tighter than \eqref{eq:devcont1}.  
We also prove a generalization of \cite[Theorem 2.1]{Bartlett-2005} to function classes with unbounded envelope (Lemma \ref{Lem:statistver:Fuk-Nagaev}), which is a statistical version of Fuk-Nagaev type inequality given in \cite[Theorem 4]{Adamczak-2008} and  could be of independent interest. 

  Exponential concentration inequalities  combined with random coding- or expurgation-based  arguments are useful in information-theoretic security to show existence of coding schemes satisfying secrecy criteria quantified in terms of a divergence  measure (see e.g., \cite{Jingbo-2017,Goldfeld-2020-wiretap,SBGPS-2021,SG-2021-SC}). The deviation inequalities   \eqref{eq:devcont1}-\eqref{eq:devcont2} are continuous alphabet analogues of such inequalities for R\'{e}nyi divergences, which, to the best of our knowledge, are known currently only for finite alphabets or compactly supported distributions. Hence, we expect that these inequalities will have potential applications to secure and covert communications over a Gaussian wiretap channel\cite{GWC-1978}.    As a concrete application of \eqref{eq:devcont3},  we consider a hypothesis testing based framework for auditing R\'{e}nyi DP and  obtain non-asymptotic performance guarantees for this test. While this application has been considered before in  \cite{domingo2022auditing} and \cite{SGK-IT-2023}, the auditing guarantees were either  partial or only valid asymptotically. Here, we address these shortcomings thus making the  auditing framework more appealing in practice.  
\subsection{Related Work}\label{Sec:relwork}
The literature on estimation of information and divergence measures is extensive, hence we will only focus on the literature relevant to R\'{e}nyi divergence estimators, which  includes KL divergence, mutual information and entropy estimators as a special case. Several estimators of these measures have been proposed and their statistical properties analyzed. In this regard,  consistency, convergence rates, and minimax rates of various estimators based on  plug-in, kernel and k-nearest neighbour methods have been studied in \cite{Tsybakov-1996,Darbellay-1999,Antos-Ioannis-2001,Paninski-2003,Singh-2003,Kraskov-2004,Goria-2005,Wang-2005,Cai-2006,Leonenko-2008,Perez-2008,haje2009entropy,Nguyen-2010,Pal-2010,Valiant-2011,Poczos-Schneider-2011,Sricharan-2012,Krishnamurthy-2014,kandasamy2015nonparametric, Jiantao-2015,Singh-Poczos-2016,Yihong-Yang-2016,Ioannis-Skoularidou-2016,Noshad-2017,Delatre-Fournier-2017,Valiant-2017,Gao-2018,Moon-2018,Bu-2018,Wisler-2018,berrett2019-knn,Jayadev-2019,Goldfeld-2020-smoothemp,Yanjun-2020-renyi,Yanjun-2020,Zhao-2020,Bulinksi-2021,AZYA-2022,Ryu-2022,berrett2023-twosampfunc,Tsur-2023,Tsur-2024,Mina-2024} (see also references therein).  Further statistical properties such as the asymptotic distribution of the estimation error in the flavor of the classical central limit theorem have also been characterized, see e.g. \cite{kandasamy2015nonparametric,Moon-2018,berrett2019-knn,berrett2023-twosampfunc,SGK-IT-2023,Sreekumar-IT-2025}. While many of the above references assume compactly supported distributions with smooth densities bounded away from zero, some of the results hold more generally. Notably,  \cite{berrett2019-knn} and \cite{berrett2023-twosampfunc} show that weighted average
of Kozachenko–Leonenko estimators \cite{Kozachenko-Leonenko} are minimax-rate optimal for entropy and divergence estimation, respectively, over a general class of locally smooth distributions (which includes multi-variate Gaussians) under the squared error criterion. For the same error criterion, \cite{Yanjun-2020} establish minimax rate optimality of a kernel-based entropy estimator over a Lipschitz class of densities  supported on $[0,1]^d$.
Consistency of a neural estimator of R\'{e}nyi divergence has been explored in   \cite{Birell-2021},   while minimax convergence rate for neural estimation of KL divergence over a sufficiently regular class of distributions has been established in \cite{sreekumar2021neural}. 

Probabilistic tail inequalities for kernel-based  estimators of R\'{e}nyi divergences and mutual information as well as neural estimators of KL divergence for compactly supported distributions have been derived in \cite{Singh-Poczos-2014,pmlr-v32-singh14,Liu-2012, sreekumar2021neural}. For plug-in estimators of KL divergence, Chernoff-type concentration bounds have been studied in \cite{Guo-Richardson-2021} for the finite alphabet setting, building upon earlier works \cite{Agrawal-2020,Mardia-2019}. However, all these results are limited to distributions with compact support and smooth densities bounded away from zero, which we aim to generalize upon.
\section{Preliminaries}
\subsection{Notation} \label{Sec:notation}
Let $(\Omega, \cA,\PP)$ be a sufficiently rich probability space on which all random variables are defined. Let  $\rho$ be a sigma-finite positive measure on the Borel sigma algebra $\cB(\RR^d)$ and  $\cP(\RR^d)$ be the set of Borel probability measures. For $\mu \in \cP(\RR^d)$ such that $\mu\ll \rho$, we use $p_{\mu}=\frac{d \mu}{d\rho}$ to denote its density (Radon-Nikodym derivative), w.r.t.  $\rho$. $\mu^{\otimes n}$ stands for the $n$-fold product measure, and $\delta_x$ represents the Dirac measure at $x$. $\ind_{\cE}$ denotes the indicator of set $\cE$.
  $\mu*\nu$ denotes the convolution of $\mu$ and $\nu$; likewise, $f*g$ represents convolution of two measurable functions $f,g$.  $X\sim \mu$ stands for random variable $X$ distributed according to $\mu$. The expectation of $X \sim \mu$, whenever it exists, will be denoted by $\EE_{\mu}[X]$. We write $\gamma_{\sigma} = N(0,\sigma^2 I_d)$ for the centered Gaussian distribution on $\RR^d$ with covariance matrix $\sigma^2 I_d$, and use $\varphi_{\sigma}(x) = (2\pi\sigma^2)^{-d/2}e^{-\norm{x}^2/(2\sigma^2)}$ for the corresponding density. Inequalities between random variables are interpreted to hold almost surely (a.s.) or surely, which will be evident from the context.

For $1 \leq r \leq \infty$, let $L^r(\rho)$ be the space of all real-valued measurable functions $f$ on $\RR^d$ such that $\norm{f}_{r,\rho}:=\big(\int_{\mathfrak{S}} \abs{f}^r d \rho\big)^{1/r} <\infty$, with the usual identification of functions that are equal $\rho$-almost everywhere (a.e.). For $1 \leq r < \infty$, the space $(L^r,\norm{\cdot}_{r,\rho})$ is a separable Banach (and hence, Polish) space.  When $\rho$ is the Lebesgue measure on $\RR^d$, we use  $\| \cdot \|_{r}$ and $L^r(\RR^d)$ to denote the relevant norm and space of functions. The corresponding restrictions to a measurable subset $\cX \subseteq \RR^d$ are denoted by $\| \cdot \|_{r,\cX}$ and $L^r(\cX)$, respectively.  The support of $\rho$ is denoted by $\supp(\rho)$. We use $a \lesssim_{x} b$ to denote that $a \leq c_x b$ for some constant $c_x>0$ which depends only on some parameter $x$, and omit the subscript $x$ for universal constants. For $x,y \in \RR$, $x \wedge y=\min\{x,y\}$ and $x \vee y= \max\{x,y\}$. 
\subsection{Definitions}
The following function classes will play a role in our results. 
\begin{definition}[H\"{o}lder class, see e.g. \cite{AVDV-book}] \label{def:Holderclass} 
For a $d$-dimensional vector $k$ with non-negative integers as elements, consider the differential operator 
\begin{align}
    D^k =\frac{\partial^{|k|}}{\partial x_1^{k_1} \ldots \partial x_d^{k_d}}. \notag
\end{align}
For a function $f:\cX \rightarrow \RR$ defined on a set $\cX \subseteq \RR^d$ with nonempty interior, the H\"{o}lder norm  of  order $\beta>0$ is 
\begin{align}
\norm{f}_{\beta,\cX}:=\max_{0 \leq |k| \leq \ubar \beta} \sup_{x} \abs{D^k f(x)}+\max_{|k|=\ubar \beta} \sup_{x,y} \frac{\abs{D^{k}f(x)-D^{k}f(y)}}{\norm{x-y}^{\beta-\ubar{\beta}}}, \label{eq:defHoldnorm}
\end{align}
where $\ubar{\beta}$ is the largest integer strictly smaller than $\beta$,  and  the supremums are taken on the interior of $\cX$ with $x \neq y$. The H\"{o}lder class $C_M^{\beta}(\cX)$ of smoothness parameter $\beta$ and norm parameter $M$ is the set of all continuous functions $f:\cX \rightarrow \RR$ such that $\norm{f}_{\beta,\cX} \leq M$. 
\end{definition} 
\begin{definition}[Lipschitz class, see \cite{DL-1993}] \label{def:Lipschitzclass} 
For  $r \in (0,\infty]$, $m \in \NN$, 
and $f \in L^r\big(\rho\big)$, the $m$-th modulus of smoothness of $f$ is
\begin{equation}
  \chi_{m,r}(f,t,\rho):=\sup_{y \in \RR^d, \norm{y} \leq t}\norm{\Delta_{y}^m f}_{r,\rho}, \notag
\end{equation}
where $\Delta_{y}^m f(x)=\sum_{j=0}^{m} (-1)^{m-j}f(x+jy)$. For $0<s \leq 1$ and $\cX \subseteq \RR^d$,
 the Lipschitz class with smoothness parameter $s$ and norm parameter $M$ is
$\mathsf{Lip}_{s,r,M}(\rho):=\big\{f \in L^r\big(\rho\big): \norm{f}_{\mathrm{Lip}(s,r,\rho)} \leq M \big\}$
where 
\begin{align}
\norm{f}_{\mathrm{Lip}(s,r,\rho)}:=\inf_{\substack{g: \\g=f, \rho \text{  a.e.}}}\norm{g}_{r}+\sup_{t>0} t^{-s}\chi_{1,r}(g,t) \notag
\end{align}
 is the Lipschitz seminorm. 
\end{definition} 
The above definition of Lipschitz class w.r.t. a general $\rho$ is a slight generalization of the one given in \cite{DL-1993} w.r.t. the Lebesgue measure. This is required  since we also want to treat discrete probability measures on $\RR^d$.  It is known that  $\mathsf{Lip}_{s,r,M}(\rho)$ grows as the smoothness parameter $s$ decreases. Also, the class of functions with bounded variation (for $d=1$)  is contained in $\mathsf{Lip}_{1,r,M}(\rho)$ for some appropriately chosen $M$ and $\rho$ being the Lebesgue measure. Hence, the Lipschitz class of smoothness $0 < s \leq 1$ contains most probability densities of practical interest.

We will also require the notion of covering entropy as stated next.
\begin{definition}[Covering  entropy, see e.g.\cite{AVDV-book}] \label{cov-pack-num}
Let $(\Theta,\mathsf{d})$ be a metric space.  A set $\Theta'\subseteq\Theta$ is an $\epsilon$-covering of $(\Theta,\mathsf{d})$  if for every $\theta \in \Theta$, there exists $\tilde{\theta} \in \Theta' $ such that $\mathsf{d}(\theta,\tilde{\theta})\leq \epsilon$.  The $\epsilon$-covering number is 
\begin{align}
N(\epsilon,\Theta,\mathsf{d}):=\inf \left\{|\Theta'|:\,\Theta' \mbox{ is an } \epsilon \mbox{-covering of }\Theta \right\},  \notag
\end{align}
and the $\epsilon$-covering entropy is $\log N(\epsilon,\Theta,\mathsf{d})$. 
\end{definition}

The Orlicz norm on the space of random variables that specify their tail behaviour will play a role in our proofs.
\begin{definition}[Orlicz space, see e.g., \cite{AVDV-book}]\label{def:Orliczspace}
An  increasing 
convex function $\psi:[0,\infty) \rightarrow [0,\infty)$ with $\psi(0)=0$ and $\lim_{x \rightarrow \infty} \psi(x)=\infty$ is called an Orlicz function. The associated Orlicz norm of a random variable  $X$ taking values in a Banach space $\mathfrak{B}$ is 
\begin{align}
   \norm{X}_{\psi}=\inf\left\{c>0: \EE\left[\psi\left(\frac{\norm{X}_{\mathfrak{B}}}{c}\right) \right]\leq 1\right\}. \label{eq:Orlicznorm}
\end{align}
\end{definition}
Examples of Orlicz functions include $\hat{\psi}_{p}(z)=z^p$ and  $\psi_{p}(z)=e^{z^p}-1$, $z \in \RR$, for   $p \geq 1$; in particular, $\psi_{2}$ corresponds to the sub-Gaussian class defined next.
\begin{definition}[Sub-Gaussian distribution, see e.g., \cite{jin-2019-subgaussnorm}]\label{def:subgaussvec}
A distribution $\mu \in \cP\big(\RR^d\big)$ is $\sigma^2$-sub-Gaussian for $\sigma>0$ if $X \sim \mu$ satisfies
\begin{align}
    \EE \left[e^{u \cdot(X-\EE[X])}\right] \leq e^{\frac{\sigma^2\norm{u}^2}{2}},\quad \forall~ u \in \RR^d.\notag
\end{align}
For $M \geq 0$, let $\mathsf{SG}(M)$ be the set of all $\sigma^2$-sub-Gaussian distributions with $\sigma^2 \vee \norm{\EE[X]} \leq M$.
\end{definition}
We will also consider $\psi_q$ for $0<q<1$. In this case, $\norm{X}_{\psi_q}$  as given in \eqref{eq:Orlicznorm}, only defines a quasi-norm on the space of random variables as opposed to a norm (since $\psi_q$ is not convex). However, we will allow this abuse of notation.

We further require the following technical assumption on certain function classes of interest.
\begin{assump}[Pointwise separability and existence of envelope] \label{Assump:funclass}
 A class of measurable functions $f: \RR^d \rightarrow \RR$ is pointwise separable if it is separable w.r.t. the topology of pointwise convergence on $\RR^d$. Further, a class $\cF$ has a measurable envelope if there exists a measurable function  $F: \RR^d \rightarrow [0,\infty)$ such that $\sup_{f \in \cF} \abs{f(x)} \leq F(x)$ for all $x \in \RR^d$. Here, $\sup_{f \in \cF} \abs{f(x)}$ is the minimal envelope, but other choices can also be considered. Henceforth, whenever we consider an envelope, measurability will be implicitly assumed.  
\end{assump}
Assumption \ref{Assump:funclass}  effectively allows empirical process suprema over $\cF$ to be replaced by those over a countable sub-class of functions and avoids measurability issues that could arise when dealing with such quantities. 
\section{Main Results} 
For $X^n \sim \mu^{\otimes n}$ and $Y^n \sim \nu^{\otimes n}$ where $\mu,\nu \in \cP(\RR^d)$,   let \[
\hat{\mu}_n =\frac{1}{n}\sum_{i=1}^n \delta_{X_i} \quad \text{and} \quad  \hat{\nu}_n = \frac{1}{n}\sum_{i=1}^n \delta_{Y_i}
\]
denote the empirical distributions of  $X^n $ and   $Y^n$, respectively. Let $\BB_d(r):=\{x \in \RR^d: \norm{x}_2 \leq r\}$ denote the Euclidean ball of radius $r$ and  $\cP\big(\BB_d(r)\big):= \{\mu \in \cP(\RR^d): \supp(\mu)\subseteq \BB_d(r)\}$ be the set of Borel probability measures whose support is contained in $\BB_d(r)$.
\subsection{Deviation Inequalities for Smoothed Plug-in Estimators}
We first state the deviation inequality for Gaussian-smoothed plug-in estimator of KL divergence. 
\begin{theorem}[Compactly supported and sub-Gaussian distributions]\label{Thm:KL_compsupp_sg}
  Let $0 <\sigma \leq 1$, $r \geq 0$, $p \geq 2$, $L\geq 1$ and $0 \leq \tau \leq 1$. Then, the following hold: 
   \begin{enumerate}[(i)]
   \item 
\begin{flalign}
&\sup_{\mu,\nu \in \cP\left(\BB_d(r)\right)}\mspace{-3 mu}\PP \mspace{-2 mu}\left(\abs{\kl{\hat \mu_n*\gamma_{\sigma}}{\hat \nu_n*\gamma_{\sigma}}-\kl{\mu*\gamma_{\sigma}}{\nu*\gamma_{\sigma}}} \mspace{-2 mu}\gtrsim\mspace{-2 mu}  \xi_{r,d,\sigma}\mspace{-4 mu}\left(\tilde \xi_{r,d,\sigma} n^{-\frac 12} +z\right)\mspace{-2 mu}\right) \mspace{-2 mu}\leq 2e^{-\big(nz^2 \wedge  nz\big)}, \label{eq:KLdevineq-compsupp-simp} && 
\end{flalign}
where $\xi_{r,d,\sigma}$ and $\tilde \xi_{r,d,\sigma}$ are defined in \eqref{eq:devterm1} and \eqref{eq:devterm2}, respectively.
\item 
\begin{align}
&\sup_{\substack{\mu,\nu \in \mathsf{SG}(L):\\X \sim \mu,~ \norm{X}_{\psi_p} \leq L}} \PP\left(\big|\kl{\hat \mu_n*\gamma_{\sigma}}{ \nu*\gamma_{\sigma}}-\kl{\mu*\gamma_{\sigma}}{\nu*\gamma_{\sigma}}\big| \gtrsim \check{\xi}_{d,L,\sigma}  \left(\hat{\xi}_{d,L,\sigma}n^{-\frac 12(1-\tau)}+z\right)\right) \notag \\
& \qquad \qquad\qquad\qquad\qquad\qquad\qquad\qquad\qquad\qquad\qquad\qquad\qquad\qquad\qquad\leq  3 e^{-\hat{\zeta}_{n,d,p,\tau}(z)}, \label{eq:KLdevineq-subgauss}
\end{align}
where 
\begin{align}
 \hat{\zeta}_{n,d,p,\tau}(z):=n^{\frac{d+12-8\tau}{d+12}}z^2 \wedge n^{\frac{d+12-4\tau}{d+12}}z \wedge  \frac{n^{\frac{p\tau}{d+12}}}{\log(n+1)}, \label{eq:rateorlicsg}
\end{align}
and
$\hat{\xi}_{d,L,\sigma}$,  $\check{\xi}_{d,L,\sigma}$ are given in \eqref{eq:constrateexp}. 
\end{enumerate}
\end{theorem}
The proof of Theorem \ref{Thm:KL_compsupp_sg} is given in Section \ref{Thm:KL_compsupp_sg-proof} and relies on a more general result (Lemma \ref{Lem:concent-KL}) stated below. At a high level, the proof relies on relating the desired deviation to suprema of empirical processes indexed by appropriate function classes and  
 leveraging  concentration inequalities from empirical process theory. More specifically,  Lemma \ref{Lem:concent-KL} relies on a Fuk-Nagaev type  inequality given in  \cite[Theorem 2]{Adamczak-2010} (see Theorem \ref{Thm:Admczak-tailbnd} below) that applies to an unbounded function class with an integrable envelope.   
 We also obtain an upper bound on the covering entropy of Gaussian-smoothed bounded function class (Lemma \ref{lem:bndholdernorm}) by upper bounding the H\"{o}lder norm of such functions and utilizing a refinement of the covering entropy estimate for the H\"{o}lder class, which specifies the dependence on dimension and smoothness parameter in the prefactor.

 Theorem \ref{Thm:KL_compsupp_sg}  provides deviation inequalities under a more general setting than covered by existing results. Note that the usual assumptions of compact support and bounded (away from zero) densities  do not apply here. These regularity assumptions on the densities are now replaced effectively by the condition that the log-likelihood of its ratio belongs to an appropriate Orlicz class.   It can be noted from the proofs that the condition $\sigma \leq 1$ is not a fundamental restriction, but only imposed since it  simplifies the  analysis in addition to being also the regime of interest for smoothed estimators.

Observe that  \eqref{eq:KLdevineq-compsupp-simp} holds for compactly supported $\mu,\nu$ in the two-sample setting (where both $\mu$ and $\nu$ are approximated), while \eqref{eq:KLdevineq-subgauss} holds more generally for sub-Gaussian $\mu,\nu$, although limited to the one-sample setting (where only $\mu$ is approximated).  
 The latter restriction arises from the limitation that  in order to upper bound the desired deviation in terms of empirical process suprema and invoke  Talagrand's inequality or Fuk-Nagaev type inequality, we require $f_{\hat \mu_n*\varphi_{\sigma},\hat \nu_n*\varphi_{\sigma}}^{\star} $ to belong a.s. to a function class with an (Orlicz) integrable envelope. While this is always possible when $\mu$ and $\nu$ are compactly supported, this is no longer true when  $\mu,\nu$ are sub-Gaussian in general. In the one-sample scenario,  this limitation can be overcome  by considering the event $ E_n(t_n):=\left\{\max_{1 \leq i \leq n} \norm{X_i} \leq t_n\right\}$ to reduce to the compactly supported case, and bounding the probability of the complement event using the sub-Gaussianity of $\mu$. Then, optimizing the truncation level $t_n$ results in \eqref{eq:KLdevineq-subgauss}, with the third factor within the minimum in the RHS of \eqref{eq:rateorlicsg} corresponding to the probability of $E_n(t_n)$ complement. Although this approach is also possible in principle for the two-sample case, it does not lead to exponential inequalities since the relevant norm of the envelope  scales exponentially in the truncation level.   
\begin{remark}[Deviation inequality from expectation]
 Since 
 \begin{align}
     \kl{\mu*\gamma_{\sigma}}{\nu*\gamma_{\sigma}} \leq \EE\left[ \kl{\hat \mu_n*\gamma_{\sigma}}{\hat \nu_n*\gamma_{\sigma}}\right], \notag
 \end{align}
 by convexity of KL divergence (see e.g., \cite[Theorem 11]{vanEvren_Reyni_Div2014}), Lemma \ref{Lem:concent-KL} also yields a deviation inequality for the upper tail, i.e., \eqref{eq:KLdevineq-compsupp-simp} holds with $\abs{\kl{\hat \mu_n*\gamma_{\sigma}}{\hat \nu_n*\gamma_{\sigma}}-\kl{\mu*\gamma_{\sigma}}{\nu*\gamma_{\sigma}}}$  replaced by $\kl{\mu*\gamma_{\sigma}}{\nu*\gamma_{\sigma}} - \EE\left[ \kl{\hat \mu_n*\gamma_{\sigma}}{\hat \nu_n*\gamma_{\sigma}}\right]$. Similar claim also holds for \eqref{eq:KLdevineq-subgauss}. Further, an inspection of the proof  shows that \eqref{eq:KLdevineq-compsupp-simp} also holds when $\hat \mu_n*\gamma_{\sigma}$ or $\hat \nu_n*\gamma_{\sigma}$ (or both) are set to $\mu*\gamma_{\sigma}$ and $\nu*\gamma_{\sigma}$, respectively. 
\end{remark}
The proof of Theorem \ref{Thm:KL_compsupp_sg} is based on upper bounding expected empirical process suprema in terms of entropy integrals (see Lemma \ref{Lem:concent-KL}). While this approach results in easy to compute deviation inequalities, it is not always tight. A more precise deviation inequality can be obtained by replacing the expected suprema by its symmetrized version w.r.t. Rademacher variables, albeit by sacrificing computational simplicity. To state this result, let $\varepsilon:=\{\varepsilon_i\}_{i=1}^n$  denote i.i.d. Rademacher variables independent of $(X^n,Y^n)$,  i.e., $\PP(\varepsilon_i=1)=\PP(\varepsilon_i=-1)=0.5$. 
Denote by $\EE_{\varepsilon}[\cdot]$,  the expectation w.r.t. $\varepsilon$. We have the following deviation inequality for compactly supported distributions (see Lemma \ref{Lem:statverdevineq-KL-gen} in Section \ref{Sec:Thm:KLdev-symmradexp-proof} for a more general statement).
\begin{theorem}[Deviation inequality in terms of Rademacher expectation]\label{Thm:KLdev-symmradexp}
Let $r \geq 0$. Set
\begin{align} Z^{(\varepsilon)}_{n,\sigma} &:=\EE_{\varepsilon}\left[\sup_{f \in \cF_{\varphi_{\sigma}}(b_{r,\sigma})}\abs{\sum_{i=1}^n \varepsilon_i f(X_i)} \right] \quad \mbox{ and } \quad 
\bar Z^{(\varepsilon)}_{n,\sigma} :=\EE_{\varepsilon}\left[\sup_{f \in \cE_{\varphi_{\sigma}}(b_{r,\sigma})}\abs{\sum_{i=1}^n \varepsilon_i f(X_i)} \right], \label{eq:symmsumsmthclscomp}
\end{align}
where $b_{r,\sigma}:=\frac{r^2+4r}{\sigma^2}$, 
\begin{align}
    \cF_{\varphi_{\sigma}}(b)&:=\{f*\varphi_{\sigma}:  \abs{f(x)}  \leq 0.5 b(1+\norm{x}),~\forall~ x \in \RR^d\}, \notag \\
       \cE_{\varphi_{\sigma}}(b)&:=\{e^f*\varphi_{\sigma}:  \abs{f(x)}  \leq 0.5 b(1+\norm{x}),~\forall~ x \in \RR^d\}. \notag 
\end{align}
Then
   \begin{flalign}
&\sup_{\substack{\mu,\nu \in \cP\left(\BB_d(r)\right)}}\PP \left(  \abs{\kl{\hat \mu_n*\gamma_{\sigma}}{\hat \nu_n*\gamma_{\sigma}}-\kl{\mu*\gamma_{\sigma}}{\nu*\gamma_{\sigma}}} \geq 7  n^{-1} \big(Z^{(\varepsilon)}_{n,\sigma}+\bar Z^{(\varepsilon)}_{n,\sigma}\big)+44\xi_{r,d,\sigma}z\right)  \notag \\
&\qquad \qquad\qquad\qquad\qquad\qquad\qquad\qquad\qquad\qquad\qquad\qquad\qquad\qquad\qquad\qquad\quad\leq 2e^{-\big(nz^2 \wedge  nz\big)}, \notag  && 
\end{flalign} 
where $\xi_{r,d,\sigma}$ is given in \eqref{eq:devterm1}.
\end{theorem}
The proof of Theorem \ref{Thm:KLdev-symmradexp} is given in Section \ref{Sec:Thm:KLdev-symmradexp-proof} and follows by replacing Talagrand's inequality in the proof of Lemma \ref{Lem:concent-KL} by a statistical version involving symmetrized Rademacher expectations given in \cite[Theorem 2.1]{Bartlett-2005}. In fact, a more general  claim can be proven as stated below, which could be  of  independent interest. 
\begin{lemma}[Statistical Fuk-Nagaev inequality] \label{Lem:statistver:Fuk-Nagaev}
 Let $(\cF_n)_{n \in \NN}$ be such that  $\cF_n$ satisfies Assumption \ref{Assump:funclass} with envelope $F_n $, and suppose  $\norm{F_n}_{\psi_q}<\infty$ for some  $q \in (0,1]$ and all $n \in \NN$.  For $X^n \sim \mu^{\otimes n}$ and i.i.d. Rademacher variables $\{\varepsilon_i\}_{i=1}^n$  independent of $X^n$, set $M_n(X^n):=\max_{1 \leq i \leq n} F_n(X_i)$, $\overline{M_n}:=\EE\big[\max_{1 \leq i \leq n} F_n(X_i)\big]$,  
\begin{align}
 Z_n&:=\sup_{f \in \cF_n}\abs{\sum_{i=1}^n \big(f(X_i)-\EE_{\mu}[f]\big)} \quad \mbox{and} \quad  Z^{(\varepsilon)}_n:=\EE_{\varepsilon}\left[\sup_{f \in \cF_n}\abs{\sum_{i=1}^n \varepsilon_i f(X_i)} \right]. \label{eq:empsumfrmmean}
\end{align}     
 Then 
 \begin{align}
    \PP \left(Z_n \geq  7 Z^{(\varepsilon)}_n+z \right)  \leq 2e^{-\big(\frac{9z^2}{128 n \norm{F_n}_{2,\mu}^2} \wedge \frac{z}{228\mspace{2 mu}\overline{M_n}}\big)}+2e^{\frac{-z^q}{(c_q\norm{M_n(X^n)}_{\psi_q})^q}},\label{eq:tailbndstatver}
\end{align}
where  $c_q$ is as specified in \eqref{eq:Orlicznormconstfin}.   Further, \eqref{eq:tailbndstatver}  holds with $Z^{(\varepsilon)}_n$ replaced by 
\begin{align}
   \EE_{\varepsilon}\left[\sup_{f \in \cF_n(8 \overline{M_n})}\abs{\sum_{i=1}^n \varepsilon_i f(X_i)} \right]. \notag 
\end{align}
\end{lemma}
Lemma \ref{Lem:statistver:Fuk-Nagaev}  can be considered as a generalization of  \cite[Theorem 2.1]{Bartlett-2005} to an unbounded function class whose envelope is integrable in a suitable Orlicz sense. The proof is given in Section \ref{Sec:Lem:statistver:Fuk-Nagaev-proof}. The essential idea is based on truncation of $\cF_n$ to a bounded class of functions $ \cF_n(T_n)=\left\{f \ind_{F_n \leq T_n}: f \in \cF_n\right\}$, applying  \cite[Theorem 2.1]{Bartlett-2005} to the truncated class,   invoking the Hoffmann-J{\o}rgensen inequality (see e.g. \cite[Proposition 6.8]{LT-1991}) to control the remainder by selecting an appropriate $T_n$, and showing that $\EE_{\varepsilon}\left[\sup_{f \in \cF_n(T_n)}\abs{\sum_{i=1}^n \varepsilon_i f(X_i)} \right] \leq Z^{(\varepsilon)}_n$.  
\begin{remark}[Statistical Fuk-Nagaev with Lebesgue integrable envelope]
    Under the weaker condition $F_n \in L^p(\mu)$ for some $p\geq 2$, it is possible to derive the inequality 
   \begin{align}
    \PP \left(Z_n \geq  7 Z^{(\varepsilon)}_n+z \right)  \leq 2e^{-\big(\frac{9z^2}{128 n \norm{F_n}_{2,\mu}^2} \wedge \frac{z}{228\mspace{2 mu}\overline{M_n}}\big)}+2(16)^p\EE\big[\big(M_n(X^n)\big)^p\big] z^{-p}.\label{eq:statFN-int-env}
\end{align}  
This follows by applying Markov's inequality (in place of  \eqref{eq:exporlicznorm}) in the proof of Lemma \ref{Lem:statistver:Fuk-Nagaev}, which yields
\begin{align}
    \PP \left(Z_n'' \geq  \frac{z}{4} \right) \leq 4^p\frac{\EE\big[Z_n''^p\big]}{z^p} \stackrel{(a)}{\leq} \frac{2(16)^p\EE \left[\big(M_n(X^n)\big)^p\right]}{z^p},\notag
\end{align}
 where $(a)$ is via  an application of  Hoffmann-J{\o}rgensen inequality given in \cite[Proposition 6.8]{LT-1991}:
\begin{align}
\EE\big[Z_n''^p\big] \leq 2 (4^p) \EE \left[\max_{1 \leq i \leq n}F_n^p\ind_{\{F_n>T_n\}}(X_i)\right] \leq 2 (4^p) \EE \left[\big(M_n(X^n)\big)^p\right].\notag
\end{align}
\end{remark}
\begin{remark}[Confidence interval with unknown distribution]
Lemma \ref{Lem:statistver:Fuk-Nagaev} is useful for constructing confidence intervals for $Z_n$ when $\mu$ is unknown, provided $Z^{(\varepsilon)}_n$ can be computed and upper bounds on $\norm{F_n(X)}_{\psi_q}$ are known (since $\norm{F_n}_{2,\mu}$, $\overline{M_n}$ and $\norm{M_n(X^n)}_{\psi_q}$ can be upper bounded in terms of this). 
\end{remark}
Next, we state the  analogue of Theorem \ref{Thm:KL_compsupp_sg} and Theorem \ref{Thm:KLdev-symmradexp} for R\'enyi divergences of order $\alpha \in (0,1) \cup (1,\infty)$.
\begin{theorem}[Compactly supported  distributions]\label{Thm:compactlysuppRen}
Let $r\geq 0$ and $0<\sigma \leq 1$. Then 
\begin{align}
&\sup_{\mu,\nu \in \cP\left(\BB_d(r)\right)} \mspace{-4 mu}\PP \left(  \abs{\rendiv{\hat \mu_n*\gamma_{\sigma}}{\hat \nu_n*\gamma_{\sigma}}{\alpha}-\rendiv{\mu*\gamma_{\sigma}}{\nu*\gamma_{\sigma}}{\alpha}} \gtrsim  \lambda_{\alpha,d,r,\sigma}\big(\tilde \xi_{r,d,\sigma} n^{-\frac 12}+z\big)\mspace{-2 mu}\right) \mspace{-2 mu}\leq \mspace{-2 mu} 2e^{-\left(nz^2 \wedge nz\right)}, \label{eq:Renconcineq-Orlicz} 
\end{align}
where $\tilde \xi_{r,d,\sigma}$ and $\lambda_{\alpha,d,r,\sigma}$  are defined in \eqref{eq:devterm2} and \eqref{eq:devtermren}, respectively. 
\end{theorem}
The high-level idea used in the proof of Theorem \ref{Thm:compactlysuppRen} (see Section \ref{Sec:Thm:compactlysuppRen-proof})  is similar to that of Theorem \ref{Thm:KL_compsupp_sg}, and utilizes the variational expression given in \eqref{eq:rendonvarform} to upper bound the desired deviation in terms of empirical process suprema. However, the presence of logarithms in the expression necessitates a linearization step  prior to achieving this goal. Fortunately, this is possible when $\mu,\nu$ are compactly supported by noting that the Lipschitz constant of logarithm is bounded by $c$ on the interval $[1/c,c]$ for $c \in [1,\infty)$. Notice that in contrast to Theorem \ref{Thm:KL_compsupp_sg}, we do not have the counterpart of \eqref{eq:KLdevineq-subgauss} for sub-Gaussian $\mu,\nu$. This is because even in the one-sample setting, the truncation argument, required for similar reasons  as in Part $(ii)$ of Theorem \ref{Thm:KL_compsupp_sg},  results in the  norm of the  envelope of the function class indexing the empirical process suprema, scaling exponentially with respect to the truncation level.
\begin{remark}[Deviation inequality from expectation]
By joint convexity of R\'enyi divergences in both its arguments for  $\alpha \in [0,1]$ \cite[Theorem 11]{vanEvren_Reyni_Div2014}, Theorem \ref{Thm:compactlysuppRen} also holds with $|\rendiv{\hat \mu_n*\gamma_{\sigma}}{\hat \nu_n*\gamma_{\sigma}}{\alpha}-\rendiv{\mu*\gamma_{\sigma}}{\nu*\gamma_{\sigma}}{\alpha}|$ replaced by
   $\rendiv{\hat \mu_n*\gamma_{\sigma}}{\hat \nu_n*\gamma_{\sigma}}{\alpha}-\EE\left[ \rendiv{\hat \mu_n*\gamma_{\sigma}}{\hat \nu_n*\gamma_{\sigma}}{\alpha}\right]$ for  $\alpha \in (0,1]$. Similarly,  for all $\alpha \in (0,1) \cup (1,\infty)$,   $|\rendiv{ \mu*\gamma_{\sigma}}{\hat \nu_n*\gamma_{\sigma}}{\alpha}-\rendiv{\mu*\gamma_{\sigma}}{\nu*\gamma_{\sigma}}{\alpha}|$ can be replaced by
   $\rendiv{ \mu*\gamma_{\sigma}}{\hat \nu_n*\gamma_{\sigma}}{\alpha}-\EE\left[ \rendiv{ \mu*\gamma_{\sigma}}{\hat \nu_n*\gamma_{\sigma}}{\alpha}\right]$  due to convexity of R\'enyi divergences  in the second argument for all orders \cite[Theorem 12]{vanEvren_Reyni_Div2014}. Further,  \eqref{eq:Renconcineq-Orlicz} also applies when $\hat \mu_n*\gamma_{\sigma}$ or $\hat \nu_n*\gamma_{\sigma}$ (or both) are set to $\mu*\gamma_{\sigma}$ and $\nu*\gamma_{\sigma}$, respectively. 
\end{remark}
Thus far, we considered deviation inequalities for smoothed plug-in estimators of R\'{e}nyi divergences. Under additional assumptions uniformly bounding the ratio of densities of $\mu$ and $\nu$, it is possible to obtain probabilistic inequalities bounding the deviation of the estimate from the unsmoothed value. For stating the result, we require the  concept of a Lipschitz class of functions given in Definition \ref{def:Lipschitzclass}. For $\rho$  such that $\mu,\nu \ll \rho $,  set
\begin{align}
   & \cP_{M,\rho}:=\left\{(\mu,\nu):\mu \ll \gg \nu \ll \rho, \norm{\frac{p_{\mu}}{p_{\nu}}}_{\infty,\rho} \vee \norm{\frac{p_{\nu}}{p_{\mu}}}_{\infty,\rho} \leq M \right\}, \notag \\
   & \tilde \cP_{M,r,s,\rho}:=\left\{(\mu,\nu) \in \cP\big(\BB_d(r)\big) \cap \cP_{M,\rho}: p_{\mu}, p_{\nu} \in \mathsf{Lip}_{s,1,M}(\rho) \right\}, \notag 
\end{align}
where $\mu \ll \gg \nu$ means that $\mu \ll \nu$ and $\nu \ll \mu$. 
With this, we have the following corollary to Theorem \ref{Thm:KL_compsupp_sg}. 
\begin{cor}[Deviation from unsmoothed KL divergence]\label{Cor:Rendivsmoothrate}
Let $s \in (0,1]$, $r \geq 0$, and $M \geq 1$.  Then, with $\sigma_n=n^{-\frac{1}{d+2s+4}}$, we have
\begin{flalign}
&\sup_{\substack{(\mu,\nu) \in  \tilde \cP_{M,r,s,\rho}}} \mspace{-4 mu}\PP \left(\mspace{-2 mu}  \abs{\kl{\hat \mu_n*\gamma_{\sigma}}{\nu*\gamma_{\sigma}}\mspace{-2 mu}-\mspace{-2 mu}\kl{\mu}{\nu}}\mspace{-2 mu} \gtrsim \mspace{-2 mu} \left(\bar c_{d,r}\mspace{-2 mu}+\mspace{-2 mu}c_{d,s}M^2\mspace{-2 mu}+\mspace{-2 mu}(r^2+4r)(1+r\sqrt{d}) z\right)n^{\frac{-s}{d+2s+4}} \mspace{-2 mu}\right) \notag \\
& \qquad \qquad\qquad\qquad\qquad\qquad\qquad\qquad\qquad\qquad\qquad\qquad\qquad \leq  e^{-\left(n^{\frac{d}{d+2s+4}}z^2 \wedge n^{\frac{d+s+2}{d+2s+4}}z\right)}, \label{eq:KLconcineq-smth}&&
\end{flalign}
where  $\bar c_{d,r}:=(3ed)^{3(d+2)}\sqrt{d}r(r+1)^{\frac{d}{2}+3}$ and  $c_{d,s}:=\int_{\RR^d} \norm{z}^s \varphi_1(z) d z$.
\end{cor}
The proof of Corollary \ref{Cor:Rendivsmoothrate} (see  Section \ref{Sec:Cor:Rendivsmoothrate-proof}) relies on the fact that  the difference between  smoothed R\'{e}nyi divergence and its unsmoothed value can be upper bounded given the densities $p_{\mu}$ and $p_{\nu}$  belong to an appropriate Lipschitz class with bounded ratios.  The corresponding approximation error is quantified in the next proposition which is a generalization of  \cite[Lemma 3]{SGK-IT-2023}. 
\begin{prop}[Approximation of R\'{e}nyi divergences] \label{prop:smoothRen-stability}
Let  $\sigma \geq 0$,  $s \in (0,1]$ and $M \geq 1$. Suppose  $ (\mu,\nu) \in \cP_{M,\rho}$ are such that  $p_{\mu}, p_{\nu} \in \mathsf{Lip}_{s,1,M}(\rho)$. Then
\begin{subequations}
\begin{align}
  &  \abs{\kl{\mu}{\nu}-\kl{\mu*\gamma_{\sigma}}{\nu*\gamma_{\sigma}}}\leq c_{d,s} M\left(M+1+\log M\right) \sigma^s, \label{eq:KL-stability-ub} \\
    &  \abs{\rendiv{\mu}{\nu}{\alpha}-\rendiv{\mu*\gamma_{\sigma}}{\nu*\gamma_{\sigma}}{\alpha}}\leq c_{d,s}\left(M^{2\alpha+1}+\frac{\alpha}{\abs{\alpha-1}}M^{2(1 \vee \alpha)}\right) \sigma^s, ~\alpha \in (0,1) \cup (1,\infty),\label{eq:ren-stability-ub}
\end{align}
\end{subequations}
where $c_{d,s}:=\int_{\RR^d} \norm{z}^s \varphi_1(z) d z$.
\end{prop}
The proof of Proposition \ref{prop:smoothRen-stability} is given in Section \ref{prop:smoothRen-stability-proof} and relies on a  multi-variate Taylor's expansion of a certain functional of the densities. The Lipschitz and regularity assumptions on the densities enable to control the remainder term in this expansion.  

\begin{figure}[t]
\centering
\includegraphics[trim=1cm 8cm 13cm 5.8cm, clip, width= 0.8577\textwidth]{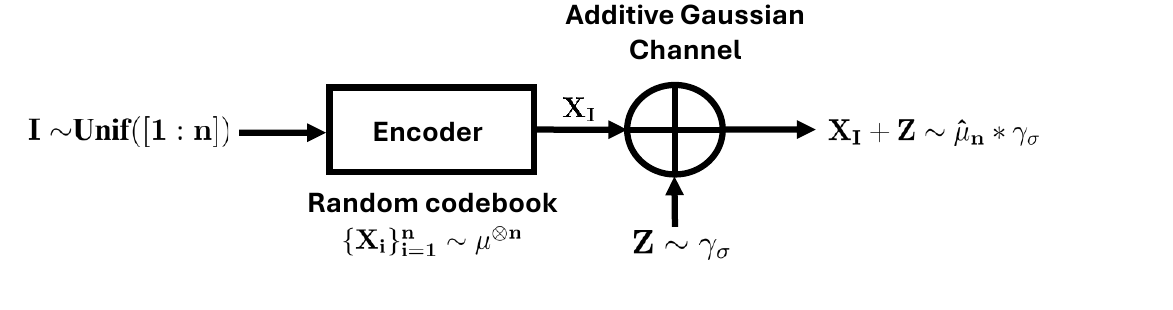}
\caption{Random coding over a Gaussian channel where each codeword $X_i$ of a random codebook $\{X_i\}_{i=1}^n$ of size $n$ is a vector of length $d$ generated i.i.d. according to $\mu$. A power constraint on each input codeword  naturally translates to a bound on its  Euclidean norm. } \label{Fig:DHT-zerorate}
\end{figure}
Before we end this subsection, we briefly mention that the deviation inequalities of the form given in Theorem \ref{Thm:KL_compsupp_sg} and Theorem \ref{Thm:compactlysuppRen} could be useful in other contexts in  information theory such as physical-layer security or covert communications, where the secrecy metric of interest is often expressed in terms of R\'{e}nyi divergences. Exponential deviation inequalities then provide high confidence bounds in random coding- or expurgation-based  arguments to show existence of codebooks satisfying the security criteria approximately.  For instance, consider the case of information-theoretically secure coding over a Gaussian wiretap channel with a per-codeword power constraint. In this context, $\{X_i\}_{i=1}^n$ would denote a random codebook of size $n$ with each codeword of length $d$ sampled independently according to $\mu$, and  $\hat \mu_n*\gamma_{\sigma}$ would correspond to the output distribution of a discrete-time additive white Gaussian noise channel of variance $\sigma^2$ with input  uniformly sampled from the codebook. Then, it is of interest to show that $\rendiv{\hat \mu_n* \gamma_{\sigma}}{\nu* \gamma_{\sigma}}{\alpha} \leq \rendiv{\mu* \gamma_{\sigma}}{\nu* \gamma_{\sigma}}{\alpha} +o(1)$ with overwhelming probability, where $\nu$ is an auxiliary distribution of interest. To the best of our knowledge, existing deviation inequalities relevant to this setting are applicable only to bounded metrics or finite (output-) alphabet channels (see e.g., \cite[Theorem 31]{Jingbo-2017} and \cite[Lemma 4]{Goldfeld-2020-wiretap}). Hence, we expect Theorem \ref{Thm:KL_compsupp_sg} and \ref{Thm:compactlysuppRen}  to have potential new applications in wiretap coding or covert communications over Gaussian channels when sufficient shared randomness is available. For wider applicability, it appears that an improvement of the dependence of dimension in these bounds would be required. 
\subsection{Deviation Inequalities for Neural Estimators of R\'{e}nyi Divergences}
Let  $\cG$ be a class of real-valued functions formed by the output of a neural network, referred to as neural net class. 
The variational forms \eqref{eq:varexpkl} and \eqref{eq:rendonvarform} leads to neural estimators of  KL  and R\'{e}nyi divergence of order $\alpha \in (0,1) \cup (1,\infty)$
given by
\begin{subequations}\label{eq:varexpestunreg}
 \begin{align}
 \klest{\hat \mu_n}{\hat \nu_n}{\cG}&:=\sup_{g \in \cG}\EE_{\hat \mu_n}[g]-\EE_{\hat \nu_n}[e^g]+1,  \label{eq:varexpestklunreg}
\end{align}  
and
\begin{align}
 \renest{\hat \mu_n}{\hat \nu_n}{\alpha}{\cG}&:=\sup_{g \in \cG} \frac{\alpha}{\alpha-1} \log \EE_{\hat \mu_n}\big[ e^{(\alpha-1)g}\big]-\log \EE_{\hat \nu_n}\big[ e^{\alpha g} \big],  \label{eq:varexpestrenunreg}
\end{align}
\end{subequations}
respectively. We will assume that $\cG$ satisfies the following mild assumption.
\begin{assump}[Neural class] \label{Assump:neuralclass}
  $\cG$ is a class of measurable functions on $\RR^d$ which is pointwise separable (i.e., separable w.r.t. the topology of pointwise convergence  on the support of $\rho$)  such that  $\norm{g}_{\infty,\rho} \leq b$, $\forall ~g \in \cG$,  and
\begin{align}
   \sup_{\gamma} \log N(\epsilon,\cG,\norm{\cdot}_{2,\gamma}) &\leq \hat c_{d,\beta}(b\epsilon^{-1})^{\frac{d}{\beta}}, \label{eq:coveringbnd}
   \end{align}
for some $\hat c_{d,\beta} \geq 0$, where the supremum is taken over all finitely supported discrete probability measures $\gamma$ with $\supp(\gamma) \subseteq \supp(\rho)$. 
\end{assump}
It is easy to see that Assumption \ref{Assump:neuralclass} is satisfied by neural nets with finite width, finite depth,  bounded parameters and continuous activation when $\rho$ has compact support for an appropriately chosen $b$. 

We next describe the class of probability measures of interest for neural estimators. 
Let 
\begin{align}
    \bar \cP_{M,\rho}(\cG,\delta) = \left\{(\mu,\nu) \in \cP_{M,\rho} : \exists~ g \in \cG \mbox{ s.t. } \norm{f_{\mu,\nu}^{\star}-g}_{\infty,\rho} \leq \delta\right\}. \notag
\end{align}
In words, $\bar \cP_{M,\rho}(\cG,\delta)$ is the subclass of $\cP_{M,\rho}$  such that the common optimizer $f_{\mu,\nu}^{\star}$ of the variational form in \eqref{eq:varformkl} and \eqref{eq:rendonvarform} can be approximated  $\rho$-a.e. by some function in $\cG$ up to an error $\delta$. Since quantifying the  approximation error $\delta$ is not our focus, we henceforth leave it as an open parameter. 
With this, we are ready to state a uniform deviation inequality for neural estimators of R\'{e}nyi divergences.
\begin{theorem}[Uniform deviation inequality  for neural estimator]\label{Thm:devineqren}   
Let $\beta >d/2$, $M \geq 1$, $a:=d/(2\beta)$,  $b:=\log M$,   $\cG_{b,\beta,\rho}$ be a neural class satisfying Assumption \ref{Assump:neuralclass}, and $c_{d,\beta}$ be such that $1 \leq c_{d,\beta} \lesssim \frac{\sqrt{\hat c_{d,\beta}}}{1-a} $, where $\hat c_{d,\beta}$ is as given in \eqref{eq:coveringbnd}. Then, the following hold:
 \begin{enumerate}[(i)]
\item \textbf{ KL divergence:} 
    \begin{align}   &\sup_{(\mu,\nu) \in \bar \cP_{M,\rho}(\cG_{b,\beta,\rho},\delta)} \mspace{-3 mu}\PP\mspace{-3 mu}\left(\abs{\klest{\hat \mu_n}{\hat \nu_n}{\cG_{b,\beta,\rho}}-\kl{\mu}{\nu}} \mspace{-2 mu}\geq \mspace{-2 mu}(M+1)\delta\mspace{-2 mu}+ \mspace{-2 mu}c_{d,\beta} M\left(1+(\log M)^a\right)\mspace{-2 mu}\big(n^{-\frac 12}+z\big) \mspace{-2 mu}\right) \notag  \\
    &\qquad \qquad \qquad \qquad \qquad \qquad \qquad \qquad \qquad \qquad \qquad   \qquad \qquad  \qquad \qquad     \leq 2 e^{-(nz^2 \wedge nz)}.\label{eq:devineqkl}
    \end{align} 
 \item \textbf{R\'{e}nyi divergence $\big(\alpha\in (0,1) \cup (1,\infty)\big)$:}      
 \begin{align}   &\sup_{(\mu,\nu) \in \bar \cP_{M,\rho}(\cG_{b,\beta,\rho},\delta)} \PP\left(\abs{\renest{\hat \mu_n}{\hat \nu_n}{\alpha}{\cG_{b,\beta,\rho}}-\rendiv{\mu}{\nu}{\alpha}} \geq \big(M^{2\alpha}+M^{2\abs{\alpha-1}}\big)\delta+ c_{d,\beta}  v_{M,\alpha,a}\big(n^{-\frac 12}+z\big) \right)\notag  \\
 &\qquad \qquad \qquad \qquad \qquad \qquad \qquad \qquad \qquad \qquad \qquad  \qquad \qquad  \qquad \qquad  \leq  2 e^{-(nz^2 \wedge nz)},\label{eq:devineqren}
    \end{align}  
where $v_{M,\alpha,a}=\frac{\alpha M^{2|\alpha-1|}}{\abs{\alpha-1}}\left(1+(\abs{\alpha-1}\log M)^{\frac a2}\right)+ M^{2\alpha}\big(1+(\alpha\log M )^a\big)$.   
    \end{enumerate}
\end{theorem}
The proof of Theorem \ref{Thm:devineqren} is given in Section \ref{Sec:Thm:devineqren-proof} and relies on an application of Talagrand's concentration inequality for uniformly bounded function class. 

We note that neural classes $\cG$ such that $\bar \cP_{M,\rho}(\cG,\delta)$ contains non-trivial sets of distributions are known. For instance,  the class of shallow neural nets $\cG_k$ with width $k$,  bounded parameters and ReLU or sigmoid activation can approximate certain smooth class of functions uniformly on compact sets up to an error of $O\big(d^{1/2}k^{-1/2}\big)$\cite{Barron_1993,Klusowski-2018,sreekumar2021neural} (see also  \cite{YAROTSKY-2017,Schmidt-Hieber-2020,Dennis-2021} for some deep neural net approximation results). Moreover, $\cG_k$ also  satisfies \eqref{eq:coveringbnd} for $\beta >d/2$ (see proof of \cite[Theorem 3]{sreekumar2021neural}).  Further, sufficient primitive conditions on the class of distributions $(\mu,\nu)$ such that $f^{\star}_{\mu,\nu}$ belongs to this smooth class was also derived \cite[Proposition 9]{sreekumar2021neural}. These results implies that $\cP_{M,\rho} \big(\cG_k,c \mspace{2 mu} d^{1/2}k^{-1/2}\big)$, for an appropriate constant $c$, contains  finite discrete distributions and  smooth compactly supported continuous distributions  with bounded ratios of densities over appropriate parameter ranges. 
\subsection{Application: Auditing R\'{e}nyi Differential Privacy }\label{Sec:DPaudit}
Here, we consider an application of Theorem \ref{Thm:devineqren} to auditing DP,   which was introduced in \cite{DMKS-2006} as an approach for quantifying privacy leakage of privatization mechanisms. We recall some DP notions that are relevant to our setting. Consider a set~$\mathfrak{U}$ with a relation~$\sim$ such that $u \sim v$, for $u,v \in \mathfrak{U}$, denotes that $u$ and $v$ are adjacent. In the DP context, $\mathfrak{U}$ is a set of databases and $u\sim v$ denotes that $u$ and $v$ are adjacent databases, differing on a single entry. 
\begin{definition}[DP mechanisms\cite{DMKS-2006,DKMMN-2006,Mironov-2017}]
Let $\epsilon, \delta \geq 0$. A randomized (measurable) mechanism  $f:\mathfrak{U} \rightarrow \RR^d$ is 
\begin{enumerate}[(i)]
    \item $\epsilon$-differentially private if  $\PP\left(f(u) \in \cT\right) \leq e^{\epsilon}\,\PP\left(f(v) \in \cT\right)$ for every $u \sim v$ and $\cT \in \cB(\RR^d)$;
   \item $(\epsilon,\delta)$-differentially private if $\PP\left(f(u) \in \cT\right) \leq e^{\epsilon}\,\PP\left(f(v) \in \cT\right)+\delta$  for every $u \sim v$ and  $\cT \in \cB(\RR^d)$;
   \item $\epsilon$- R\'{e}nyi differentially private of order $\alpha$  if $\rendiv{\mu_{u}}{\mu_{v}}{\alpha} \leq \epsilon$ for every $u \sim v$,  where $\mu_{u} \in \cP(\RR^d)$ is the distribution of $f(u)$.
\end{enumerate}
\end{definition}
As $\epsilon$-DP is equivalent to $\sup_{u\sim v}\mathsf{D}_{\max}(\mu_u\|\mu_v)\leq \epsilon$, where $\mathsf{D}_{\max}$ is the $\infty$-order R\'{e}nyi divergence and $\rendiv{\mu}{\nu}{\alpha}$ is non-decreasing in $\alpha$ \cite[Theorem 3]{vanEvren_Reyni_Div2014}, it is clear that R\'{e}nyi DP of any  finite order $\alpha$ is a relaxation of DP. However, the two notions become equivalent at $\alpha = \infty$, and hence $\epsilon$- R\'{e}nyi DP is a good approximation of $\epsilon$-DP for sufficiently large  $\alpha$.

In practice, given output samples from a privacy mechanism, one encounters the problem of ascertaining whether the mechanism is differentially private or not, referred to as auditing DP. This can be framed as a binary hypothesis testing problem  where the null $H_0$ and the alternative $H_1$ corresponds to when privacy holds and does not hold, respectively. Note that the privacy mechanisms to be audited, hence  also the output distributions $\mu_u,\mu_v$, are usually unknown. Hence, auditing DP cannot be reduced to the problem of hypothesis testing  between probability distributions and computing R\'{e}nyi divergence to ascertain whether the mechanism is differentially private or not.

We next formally state the relevant hypothesis testing problem. 
Let $f:\mathfrak{U}\to \RR^d$ be a  
privacy mechanism and denote a pair of adjacent databases\footnote{Note that the current model subsumes the case of deterministic $(U,V)$ by taking $\tilde \pi$ to be a point mass on a pair of adjacent databases. In this case, the randomness in $\big(f(U),f(V))$ only comes from the mechanism.} by $(U,V) \sim \tilde \pi \in \cP(\mathfrak{U} \times \mathfrak{U})$.  
Let $\pi:=(f,f)_\#\tilde \pi$ be the joint distribution of $\big(f(U),f(V)\big)$, where $\#$ is the pushforward operation. Let $(X_1,Y_1),\ldots,(X_n,Y_n) $ $\sim \pi $ be pairwise i.i.d. samples of the privacy mechanism's output when acting on i.i.d.
pairs of adjacent databases, where  $ \pi=\pi_0$ under $H_0$ and $\pi=\pi_1$ under $H_{1}$.   With the marginals of $\pi_0$ and $\pi_1$ denoted by $(\mu_0,\nu_0)$ and $(\mu_1,\nu_1)$, respectively,  the relevant hypothesis test for auditing R\'{e}nyi DP of order $\alpha $ is  
\begin{equation} \label{HT-Ren-DP}
   H_0: \rendiv{\mu_0}{\nu_0}{\alpha} \leq \epsilon \quad \text{against} \quad 
   H_{1}: \rendiv{\mu_1}{\nu_1}{\alpha} > \epsilon.
\end{equation} 
Denote the empirical measures of $(X_1,\ldots,X_n)$ and $(Y_1,\ldots,Y_n)$ by $\hat \mu_{n}$ and $\hat \nu_{n}$, respectively. 
For a test statistic $T_n=T_n(X_1,\ldots,X_n,Y_1,\ldots,$ $Y_n)$, a standard class of tests rejects $H_0$ if $T_n>t_n$ for some $t_n \in \RR$.   
The type I and type II error probabilities achieved by the test $T_n$ are $e_{1,n}(T_n,t_n)=\PP(T_n >t_n|H=H_0)$ and $e_{2,n}(T_n,t_n)=\PP(T_n  \leq t_n|H=H_1)$, respectively, where $H$ denotes the true hypothesis.

In \cite{domingo2022auditing}, a hypothesis test for auditing DP using regularized kernel R\'{e}nyi divergence is proposed, where the null hypothesis is that the mechanism satisfies $(\epsilon,\delta)$-DP. A decision rule achieving any type I error probability is proposed, leaving the characterization of the  type II error probability open. In \cite{SGK-IT-2023}, a hypothesis testing based framework for auditing DP using KL divergence as the metric was proposed along with full asymptotic performance guarantees  that relied on the knowledge of limit distribution for smoothed KL divergence.  Here,  we propose a neural estimator based hypothesis testing framework for auditing R\'{e}nyi DP along with finite sample performance guarantees. We will assume that the neural net class $\cG_n$ used for the audit and the marginal distributions $(\mu_i,\nu_i)$, $i=0,1$, satisfy the following assumption. 
\begin{assump}[Assumptions on the neural class and distributions] \label{Assump:renDPtest}
Let  $\alpha \geq 1$,  $M \geq 1$,  $\beta >d/2$, and $b=\log M$. The neural net class $\cG_n$ and  $(\mu_i,\nu_i)$, $i=0,1$, satisfy $\cG_n \subseteq \cG_{b,\beta,\rho}$ and $ (\mu_i,\nu_i) \in \cP_{M,\beta}(\cG_n,\delta_n)$,  for some  $\delta_n =o(1)$.   
 Further, $\rendiv{\mu_0}{\nu_0}{\alpha} \leq \epsilon$ and $\rendiv{\mu_1}{\nu_1}{\alpha} > \epsilon$.
\end{assump}
Assumption \ref{Assump:renDPtest} is not much restrictive in practice. Indeed, the definition of DP itself necessitates that $\norm{p_{\mu_u}/p_{\mu_{v}}}_{\infty,\rho}$ is bounded\footnote{Here, $\rho$ can be taken to be the Lebesgue measure in the continuous distribution setting or average of probability mass functions  in the discrete setting.} uniformly for all $u,v \in \mathfrak{U}$ with $u \sim v$. Moreover, as alluded to earlier, it is known that even restricting to shallow neural nets,  the class of distributions $\cP_{M,\beta}(\cG_n,\delta_n)$ with $\delta_n=O(n^{-1/2})$ is sufficiently rich. In particular, for  $\beta> d/2$ and $M$ sufficiently large,  such a class  includes discrete distributions that are mutually absolutely continuous or  compactly supported continuous distributions with bounded ratios of densities. Hence,  Assumption \ref{Assump:renDPtest} allows for audit of most privacy mechanisms of practical interest. 

The next result proposes a test statistic for R\'{e}nyi DP audit and characterizes the associated error probabilities (see Section \ref{Sec:prop:ht-audit-proof} for proof).
\begin{prop}[DP audit guarantees] \label{prop:ht-audit}
    Suppose Assumption \ref{Assump:renDPtest} holds, and $c_{d,\beta},\mspace{2 mu}  v_{M,\alpha,a}$ be as specified in Theorem \ref{Thm:devineqren}.  Then, for $0 \leq \tau \leq 1$, the test statistic $T_n=\renest{\hat \mu_n}{\hat \nu_n}{\alpha}{\cG_n}$ with critical value 
    \begin{align}  t_n=\epsilon+\big(M^{2\alpha}+M^{2\abs{\alpha-1}}\big)\delta_n+ c_{d,\beta}  v_{M,\alpha,a}\big(n^{-\frac 12}+n^{\frac{\tau-1}{2}}\big), \notag
    \end{align}
    satisfies
    \[
       e_{1,n}(T_n,t_n) \leq 2e^{-n^{\tau}} \quad \text{and} \quad  
       e_{2,n}(T_n,t_n) \leq 2e^{-\vartheta_{n,\alpha,\delta_n}(\tau)}, 
    \]
 where  $\vartheta_{n,\alpha,\delta_n}(\tau)$ is given in \eqref{eq:expt2err} below.
\end{prop}
We note that the performance guarantees given above, though derived for auditing  DP, apply to any hypothesis test of a similar kind using R\'{e}nyi divergence as the metric to differentiate the hypotheses.
\section{Proofs}
We first state a few definitions that will be used throughout this section. 
 For $\alpha \in (0,\infty)$, let
  \begin{align}
   \cF_{\varphi_{\sigma}}(\cF):=\big\{f*\varphi_{\sigma}: f \in \cF\big\}  \quad \mbox{and } \quad   \cE_{\varphi_{\sigma},\alpha}(\cF):=\big\{e^{ \alpha f}*\varphi_{\sigma}:f \in \cF\big\}, \label{eq:smoothfuncls}
\end{align}    
 be the Gaussian-smoothed function classes obtained by convolving the functions in $\cF$ and $e^{\alpha \cF}:=\{e^{\alpha f}:f \in \cF\}$ by $\varphi_{\sigma}$, respectively. We will denote $\cE_{\varphi_{\sigma},1}(\cF)$ by $\cE_{\varphi_{\sigma}}(\cF)$ henceforth. 
 
 For $\delta \geq 0$ and $\cF$ with envelope $F$, define the  entropy integral (see \eqref{cov-pack-num} for definition of covering entropy): 
\begin{align}
    J(\delta,F,\cF,\mu):= \int_0^{\delta} \sup_\gamma \sqrt{\log N(\epsilon \norm{F}_{2,\gamma},F, \norm{\cdot}_{2,\gamma} )} d\epsilon. \label{eq:coventrdef}
\end{align}
Here, the supremum is over all  probability measures $\gamma$ with discrete finite support contained in that of $\mu$.
\subsection{Proof of Theorem \ref{Thm:KL_compsupp_sg}}
\label{Thm:KL_compsupp_sg-proof}
We will state a more general lemma than required for our purposes, which provides a deviation inequality for smoothed plug-in estimators of KL divergence satisfying certain regularity conditions. Let $\Phi \in \cP(\RR^d)$, and $\phi$ denote its Lebesgue density which will act as a smoothing kernel. Note that $\norm{\phi}_{1,\RR^d}=1$. 
Set $ \cF_{\phi}(\cF):=\big\{f*\phi: f \in \cF\big\}$ and $\cE_{\phi}(\cF):=\big\{e^f*\phi: f \in \cF\big\}$. 
\begin{lemma}[Deviation inequality for KL divergence]\label{Lem:concent-KL}
 Let $\pi \in \cP(\RR^d \times \RR^d)$ with marginals $\mu$ and $\nu$ be such that $\mu \ll \nu$.  Suppose $\cF_n$ with envelope $F_n$ is such that $\cF_{n,\phi}:=\cF_{\phi}(\cF_n)$ and $\cE_{n,\phi}:=\cE_{\phi}(\cF_n)$ with envelopes\footnote{By Young's convolution inequality,  a sufficient condition for existence of a measurable envelope $F_n *\phi$ for $\cF_{n,\phi}$ is $F_n \in L^p(\RR^d)$ for some $p \geq 1$.}  $F_n *\phi$ and $e^{F_n} *\phi$, respectively, satisfy Assumption \ref{Assump:funclass},  $f_{\mu*\Phi,\nu*\Phi}^{\star}\in \cF_n$, $F_n*\phi \in L^p(\mu)$ and $e^{F_n}*\phi \in L^p(\nu)$ for   some $p\geq 2$. Further, assume that   for $(X^n,Y^n) \sim \pi^{\otimes n}$,  $\hat \mu_n * \Phi \ll \hat \nu_n * \Phi$ and 
  $f_{\hat \mu_n*\Phi,\hat \nu_n*\Phi}^{\star} \in \cF_n$ a.s. 
Then
\begin{align}
&\PP \left(  \abs{\kl{\hat \mu_n*\Phi}{\hat \nu_n*\Phi}-\kl{\mu*\Phi}{\nu*\Phi}} \gtrsim_p n^{-\frac 12} t_{n,\phi,\mu,\nu} +z\right) \notag \\
& \qquad \qquad \qquad  \leq e^{-R_{n,p,\phi,\mu}(z)} + e^{-\tilde R_{n,p,\phi,\nu}(z)}+\Big(\norm{F_n*\phi}_{p,\mu}^p +\norm{e^{F_n}*\phi}_{p,\nu}^p\Big) n^{1-p} z^{-p}, \label{eq:KLdevineq-mombnd} 
\end{align}
where  
\begin{subequations} \label{eq:paramsconbndmom}
\begin{flalign}  
&t_{n,\phi,\mu,\nu}:=J(1,F_n*\phi,\cF_{n,\phi},\mu) \norm{F_n*\phi}_{2,\mu}+   J(1,e^{F_n}*\phi,\cE_{n,\phi},\nu) \norm{e^{F_n}*\phi}_{2,\nu},  \\
&R_{n,p,\phi,\mu}(z):=    \frac{nz^2}{ \norm{F_n*\phi}_{2,\mu}^2} \wedge \frac{n^{\frac{p-1}{p}}z}{\norm{F_n*\phi}_{p,\mu}} \quad \mbox{and} \quad\tilde R_{n,p,\phi,\nu}(z):=  \frac{nz^2}{ \norm{e^{F_n}*\phi}_{2,\nu}^2} \wedge \frac{n^{\frac{p-1}{p}}z}{\norm{e^{F_n}*\phi}_{p,\nu}}. 
\end{flalign}
\end{subequations}
In particular, when  $\mu,\nu$ are compactly supported,  $\phi=\varphi_{\sigma}$ and
\begin{align}
\norm{F_n*\varphi_{\sigma}}_{2,\mu} \vee \norm{F_n*\varphi_{\sigma}}_{\infty,\mu} \vee \norm{e^{F_n}*\varphi_{\sigma}}_{2,\nu} \vee \norm{e^{F_n}*\varphi_{\sigma}}_{\infty,\nu} <\infty, \notag
\end{align}
\begin{align}
&\PP \left(  \abs{\kl{\hat \mu_n*\gamma_{\sigma}}{\hat \nu_n*\gamma_{\sigma}}-\kl{\mu*\gamma_{\sigma}}{\nu*\gamma_{\sigma}}} \gtrsim n^{-\frac 12} t_{n,\varphi_{\sigma},\mu,\nu} +z\right)  \notag \\
& \qquad \qquad \qquad \qquad  \leq e^{-\left(\frac{nz^2}{ \norm{F_n*\varphi_{\sigma}}_{2,\mu}^2} \wedge \frac{nz}{\norm{F_n*\varphi_{\sigma}}_{\infty,\mu}}\right)} + e^{-\left(\frac{nz^2}{ \norm{e^{F_n}*\varphi_{\sigma}}_{2,\mu}^2} \wedge \frac{nz}{\norm{e^{F_n}*\varphi_{\sigma}}_{\infty,\nu}}\right)}. \label{eq:KLdevineq-compsupp}
\end{align}
\end{lemma}
 The proof of Lemma \ref{Lem:concent-KL} is given in Section \ref{Lem:concent-KL-proof}. The key idea in its proof is to use the variational expression \eqref{eq:varexpkl} to upper bound the deviation $|\kl{\hat \mu_n*\Phi}{\hat \nu_n*\Phi}-\kl{\mu*\Phi}{\nu*\Phi}|$ in terms of sum of suprema of empirical processes indexed by $\cF_{n,\phi}$ and $\cE_{n,\phi}$. Then, \eqref{eq:KLdevineq-mombnd} and \eqref{eq:KLdevineq-compsupp} follow by invoking a Fuk-Nagaev type  inequality given in  \cite[Theorem 2]{Adamczak-2010} (see Theorem \ref{Thm:Admczak-tailbnd} below) and  Talagrand's concentration inequality given in \cite[Theorem 2.3]{BOUSQUET-2002} (see Theorem \ref{Thm:Talagrand-Bosquet} below), respectively.  While \eqref{eq:KLdevineq-compsupp} will suffice for our purposes, Lemma \ref{Lem:concent-KL} is more general both in terms of the flexibility of the smoothing kernel $\phi$ as well as  relaxed regularity conditions on the probability distributions. For instance, \eqref{eq:KLdevineq-mombnd} may be used to derive deviation inequalities in the two-sample Gaussian-smoothed setting when $\norm{X-Y}_2$ is bounded a.s. for $(X,Y) \sim \pi$ and the tails  of the marginals $\mu,\nu$  decay faster than sub-Gaussian. By following similar steps as in the proof below, it is easy to see that this case fits within the framework of Lemma \ref{Lem:concent-KL} since the aforementioned boundedness  and tail decay condition ensures that $f_{\hat \mu_n*\varphi,\hat \nu_n*\varphi}^{\star},f_{\mu*\varphi, \nu*\varphi}^{\star} \in \cF_n$ for an $\cF_n$ with envelope $F_n(x)=O_{\pi}(1+\norm{x}^2)$ such that the exponential moments of $F_n*\varphi$ exists. 
\begin{remark}[One-sample case] \label{Rem:onesamprate}
From the proof of Lemma \ref{Lem:concent-KL}, it is easy to see that \eqref{eq:KLdevineq-mombnd} and \eqref{eq:KLdevineq-compsupp} holds with $\nu*\Phi$ and $\hat \nu_n*\Phi$ replaced with $\nu$ such that $\mu*\Phi,\hat \mu_n*\Phi \ll \nu$ if the integrability conditions on only $F_n$ hold, i.e.,    $\norm{F_n*\phi}_{p,\mu}<\infty$ and $\norm{F_n*\phi}_{2,\mu} \vee \norm{F_n*\phi}_{\infty,\mu} <\infty$, respectively. For instance, \eqref{eq:KLdevineq-mombnd} simplifies to
\begin{align}
&\PP \left(  \abs{\kl{\hat \mu_n*\Phi}{ \nu}-\kl{\mu*\Phi}{\nu}} \gtrsim_q  n^{-\frac 12} 
J(1,F_n*\phi,\cF_{n,\phi},\mu) \norm{F_n*\phi}_{2,\mu} +z\right) \notag \\
 &\qquad \qquad \qquad \qquad\qquad \qquad\qquad \qquad\qquad \qquad\leq e^{-R_{n,p,\phi,\mu}(z)} +\norm{F_n*\phi}_{p,\mu}^p  n^{1-p} z^{-p}. \notag
\end{align}
Also, by considering the corresponding simplication for \eqref{eq:KLdevineq-compsupp}, the following improvement of  \eqref{eq:KLdevineq-compsupp-simp} holds in the one-sample setting where $\nu$ is known and only $\mu$ is estimated as $\hat \mu_n$:
\begin{flalign}
&\sup_{\mu,\nu \in \cP\left(\BB_d(r)\right)} \mspace{-3 mu}\PP \left(  \abs{\kl{\hat \mu_n*\gamma_{\sigma}}{\nu*\gamma_{\sigma}}-\kl{\mu*\gamma_{\sigma}}{\nu*\gamma_{\sigma}}} \gtrsim \bar{\xi}_{r,d,\sigma} \big(\tilde \xi_{r,d,\sigma}n^{-\frac 12}  +z\big)\right)  \mspace{-3 mu}\leq e^{-\left(nz^2 \wedge nz\right)}, \label{eq:onesampspecexp} 
\end{flalign}
where $\tilde \xi_{r,d,\sigma}$ is given in \eqref{eq:devterm2} below, and $\bar{\xi}_{r,d,\sigma}:=(r^2+4r)\big(1+r\sigma\sqrt{d}\big) \sigma^{-2}$.
\end{remark}
To prove Theorem \ref{Thm:KL_compsupp_sg}-$(i)$, fix $\mu,\nu \in \cP\big(\BB_d(r)\big)$. 
We may assume that  $r>0$ since \eqref{eq:KLdevineq-compsupp-simp} trivially holds for $r=0$.   Since $\supp(\mu),\supp(\nu) \subseteq \BB_d(r)$,  $\max_{1 \leq i \leq n} \norm{X_i} \vee \norm{Y_i}  \leq r$. Then  
\begin{align}
    \frac{\hat \mu_n * \varphi_{\sigma}(x)}{\hat \nu_n * \varphi_{\sigma}(x)}=\frac{\sum_{i=1}^n e^{-\frac{\norm{x-X_i}^2}{2 \sigma^2}}}{\sum_{i=1}^n e^{-\frac{\norm{x-Y_i}^2}{2 \sigma^2}}} \leq \frac{ e^{-\min_{i}\frac{\norm{x-X_i}^2}{2 \sigma^2}}}{ e^{-\max_{i}\frac{\norm{x-Y_i}^2}{2 \sigma^2}}} \leq \frac{ e^{-\min_{i}\frac{\norm{x}^2+\norm{X_i}^2-2\norm{x} \norm{X_i}}{2 \sigma^2}}}{ e^{-\max_{i}\frac{\norm{x}^2+\norm{Y_i}^2+2\norm{x} \norm{Y_i}}{2 \sigma^2}}} \leq e^{\frac{r^2+4r\norm{x}}{2 \sigma^2}}. \notag 
\end{align}
Similarly,  
\begin{align}
    \frac{\mu * \varphi_{\sigma}(x)}{\nu * \varphi_{\sigma}(x)}=\frac{ \int_{\RR^d} e^{-\frac{\norm{x-x'}^2}{2 \sigma^2}} d\mu(x')}{\int_{\RR^d} e^{-\frac{\norm{x-y}^2}{2 \sigma^2}}d\nu(y)} \leq \frac{ e^{-\min_{\norm{x'} \leq r }\frac{\norm{x-x'}^2}{2 \sigma^2}}}{ e^{-\max_{\norm{y} \leq r}\frac{\norm{x-y}^2}{2 \sigma^2}}} \leq \frac{ e^{-\frac{(\norm{x}^2-2r\norm{x}) }{2 \sigma^2}}}{ e^{-\frac{(\norm{x}^2+r^2+2\norm{x} r)}{2 \sigma^2}}} \leq e^{\frac{r^2+4r\norm{x}}{2 \sigma^2}}. \notag 
\end{align}
By similar steps, the same bounds also hold for  $\nu * \varphi_{\sigma}(x)/\mu * \varphi_{\sigma}(x)$ and $\hat \nu_n * \varphi_{\sigma}(x)/\hat \mu_n * \varphi_{\sigma}(x)$. Consequently, we have 
\begin{align}
    &\abs{f_{\hat \mu_n*\gamma_{\sigma},\hat \nu_n*\gamma_{\sigma}}^{\star}(x)} \vee \abs{f_{\mu*\gamma_{\sigma},\nu*\gamma_{\sigma}}^{\star}(x)} \leq \frac{r^2+4r\norm{x}}{2 \sigma^2}. \notag 
\end{align}
 Hence
 \begin{subequations}
\begin{align}\label{eq:loglikfncls} 
& f_{\hat \mu_n*\gamma_{\sigma},\hat \nu_n*\gamma_{\sigma}}^{\star}, f_{\mu*\gamma_{\sigma},\nu*\gamma_{\sigma}}^{\star} \in \cF\big(b_{r,\sigma},1\big) \mbox{ with } b_{r,\sigma}:=\frac{r^2+4r}{\sigma^2}, \\
   \mbox{where }& \cF(b,q):=\left\{f:\RR^d \rightarrow \RR, \abs{f(x)}  \leq 0.5 b(1+\norm{x}^q)\right\}.  
\end{align}      
 \end{subequations}

\noindent
For $\cX \subseteq \RR^d$, let
 \begin{align}
    v_{b,q}(\cX):= \sup_{x \in \cX} 0.5 b\big(1+\norm{x}^q\big).\label{eq:envelopmax}
\end{align} 
Define $F_{b,q}(x) :=0.5 b\big(1+\norm{x}^q\big)$,
  \begin{align}
 &  \cF_{\varphi_{\sigma}}(b,q):=\{f*\varphi_{\sigma}: f\in  \cF(b,q)\} \quad \mbox{ and } \quad 
\cE_{\varphi_{\sigma},\alpha}(b,q):=\{e^{\alpha f}*\varphi_{\sigma}: f\in  \cF(b,q)\}. \notag 
\end{align} 
Set $F_{b}(x)=0.5 b(1+\norm{x})$, $\cF_{\varphi_{\sigma}}(b):=\cF_{\varphi_{\sigma}}(b,1)$ and $   \cE_{\varphi_{\sigma}}(b):=\cE_{\varphi_{\sigma},1}(b,1)$. 

\medskip

For applying Lemma \ref{Lem:concent-KL}, we will require a bound on the entropy integrals,  $J(1,F_{b_{r,\sigma}}*\varphi_{\sigma},\cF_{\varphi_{\sigma}}(b_{r,\sigma}),\mu)$ and $ J(1,e^{F_{b_{r,\sigma}}}*\varphi_{\sigma},\cE_{\varphi_{\sigma}}(b_{r,\sigma}),\nu)$. To this end, we first obtain a slightly more general result bounding the covering entropy of $\cF_{\varphi_{\sigma}}(b,q)$. 
\begin{lemma}[Entropy bounds for smoothed function class]\label{lem:bndholdernorm}
Let $0<\sigma \leq 1$, $b \geq 0$, $0<\beta <\infty$, and $q \in \{0\} \cup [1,\infty)$. Then, for all $f \in \cF(b,q)$ and $\cX \subseteq \RR^d$ with non-empty interior,
  \begin{align}
\norm{f*\varphi_{\sigma}}_{\beta,\cX} &\leq \check c_{d,\beta,q}\mspace{2 mu}v_{b,q}(\cX)\sigma^{-(\ubar{\beta} \vee 1)}, \label{eq:Holdconstbndfclass}
\end{align}  
where $\check c_{d,\beta,q}$ is specified in \eqref{eq:constholdnorm} below. Further, for any $p \geq 1$ and $\gamma \in \cP\big(\BB_d(r)\big)$,
\begin{subequations}
\begin{align}
  &  \log N\big(\norm{F_{b,q}*\varphi_{\sigma}}_{\infty,\BB_d(r)}\epsilon,\cF_{\varphi_{\sigma}}\big(b,q\big),\norm{\cdot}_{\infty,\BB_d(r)}\big) \leq  c_{d,\beta,q}(1+r)^{d}(1+r^q)^{\frac{d}{\beta}} \mspace{2 mu} \epsilon^{-\frac{d}{\beta}}\sigma^{-\frac{(\ubar{\beta} \vee 1)d}{\beta}}, \label{eq:coventunif} \\
              &    \log N\big(\norm{F_{b,q}*\varphi_{\sigma}}_{p,\gamma}\epsilon,\cF_{\varphi_{\sigma}}\big(b,q\big),\norm{\cdot}_{p,\gamma}\big) \leq c_{d,\beta,q} (1+r)^{d}(1+r^q)^{\frac{d}{\beta}} \mspace{2 mu} \epsilon^{-\frac{d}{\beta}}\sigma^{-\frac{(\ubar{\beta} \vee 1)d}{\beta}}, \label{eq:coventholderprob} 
\end{align}
\end{subequations}
where $c_{d,\beta,q}$ is given in \eqref{eq:coventrconst} below. The same upper bounds in \eqref{eq:coventunif}-\eqref{eq:coventholderprob} also holds with  $\cF_{\varphi_{\sigma}}(b,q)$ and $F_{b,q}$  replaced by  $\cE_{\varphi_{\sigma},\alpha}(b,q)$ and $e^{\alpha F_{b,q}}$, respectively.  
Moreover, the above upper bounds with $q=1$ also applies to all $q \in [0,1]$ up to a universal multiplicative factor, and   $c_{d,\frac{d}{2}+1,q}=O\big(d^{6d+9}(9e^2)^d\big)$ for all $q \in [0,1]$.
\end{lemma}
The proof of Lemma \ref{lem:bndholdernorm} is given in Section \ref{Sec:lem:bndholdernorm-proof} and relies on a systematic analysis of the difference of derivatives of a function belonging to $\cF_{\varphi_{\sigma}}(b,q)$. 

\medskip

Let  $\beta=\frac{d}{2}+1$. It follows from \eqref{eq:coventholderprob} that for all probability measures $\gamma$ such that $\supp(\gamma) \subseteq \BB_d(r)$, 
\begin{align}
   \log N\Big(\epsilon \norm{F_{b_{r,\sigma}}*\varphi_{\sigma}}_{2,\gamma}, \cF_{\varphi_{\sigma}}(b_{r,\sigma}), \norm{\cdot}_{2,\gamma}\Big) \lesssim d^{6d+9}(9e^2)^d (r+1)^{d+2} \sigma^{-d} \epsilon^{-\frac{d}{\beta}}. \notag
\end{align}
Hence 
\begin{align}
J(1,F_{b_{r,\sigma}}*\varphi_{\sigma},\cF_{\varphi_{\sigma}}(b_{r,\sigma}),\mu) &:= \int_{0}^{1}\sup_{\gamma} \sqrt{\log N\Big(\epsilon \norm{F_{b_{r,\sigma}}*\varphi_{\sigma}}_{2,\gamma}, \cF_{\varphi_{\sigma}}(b_{r,\sigma}), \norm{\cdot}_{2,\gamma}\Big)} d\epsilon   \lesssim \tilde \xi_{r,d,\sigma}, \notag
\end{align}
where we used that $\int_{0}^1 \epsilon^{-\frac{d}{2\beta}} d\epsilon \lesssim d$ for $\beta=\frac{d}{2}+1$, and
\begin{align}
   \tilde \xi_{r,d,\sigma}:= (3ed)^{3(d+2)} (r+1)^{\frac d2+1} \sigma^{-\frac{d}{2}}. \label{eq:devterm2}
\end{align}
Similarly,
\begin{align}  J(1,e^{F_{b_{r,\sigma}}}*\varphi_{\sigma},\cE_{\varphi_{\sigma}}(b_{r,\sigma}),\nu) \lesssim \tilde \xi_{r,d,\sigma}. \label{eq:entintsmthfncls}
\end{align} 

We next bound $\norm{F_{b_{r,\sigma}}*\varphi_{\sigma}}_{2,\mu}$ and $\norm{e^{F_{b_{r,\sigma}}}*\varphi_{\sigma}}_{2,\nu}$. For this purpose, we will use the  following relation between sub-Gaussian and norm sub-Gaussian distributions. $\mu \in \cP\big(\RR^d\big)$ is $\sigma^2$-norm sub-Gaussian for $\sigma>0$ if $X \sim \mu$ satisfies
\begin{align}
    \mu\big(\norm{X-\EE [X]} > t\big) \leq 2e^{\frac{-t^2}{2\sigma^2}},\quad \forall~ t \in \RR. \notag
\end{align}
\begin{lemma}{\cite[Lemma 1, Lemma 2]{jin-2019-subgaussnorm}} \label{lem:normsubgauss}
If $\mu \in \cP\big(\RR^d\big)$ is $\sigma^2$-sub-Gaussian, then it is $8d \sigma^2$-norm sub-Gaussian. In particular, if $X \sim \mu$ and $\EE[X]=0$, then $\left(\EE[\norm{X}^p]\right)^{\frac1p} \lesssim \sqrt{d p}\sigma$ and $\EE\big[e^{t\norm{X}}\big] \leq e^{4t^2d\sigma^2 }$.
\end{lemma}
\noindent We have 
\begin{align}
F_{b_{r,\sigma}}*\varphi_{\sigma}(x)= \frac{1}{(2\pi)^{\frac d2}\sigma^d}\int_{\RR^d} F_{b_{r,\sigma}}(y)e^{-\frac{\norm{y-x}^2}{2\sigma^2}}dy  &= \frac{r^2+4r}{2(2\pi)^{\frac d2}\sigma^{d+2}}\int_{\RR^d}     (1+\norm{y}) e^{-\frac{\norm{y-x}^2}{2\sigma^2}}dy \notag \\
&\lesssim  \frac{(r^2+4r)(1+\norm{x}\sigma\sqrt{d})}{\sigma^2}, \notag 
\end{align}
where to obtain the last inequality, we used $\norm{y} \leq \norm{x}+\norm{y-x}$ and the moment bound in Lemma \ref{lem:normsubgauss} for $p=1$.  
Hence,
\begin{align}
\norm{F_{b_{r,\sigma}}*\varphi_{\sigma}}_{2,\mu} \vee \norm{F_{b_{r,\sigma}}*\varphi_{\sigma}}_{\infty,\mu} \lesssim   
\frac{(r^2+4r)\big(1+r\sigma\sqrt{d}\big)}{\sigma^2}. \notag
\end{align}
Similarly, 
\begin{align}
e^{F_{b_{r,\sigma}}}*\varphi_{\sigma}(x)= \frac{1}{(2\pi)^{\frac d2}\sigma^d}\int_{\RR^d} e^{F_{b_{r,\sigma}}(y)}e^{-\frac{\norm{y-x}^2}{2\sigma^2}}dy &\leq \frac{e^{\frac{(r^2+4r)(1+\norm{x})}{2\sigma^2}}}{(2\pi)^{\frac d2}\sigma^d}\int_{\RR^d} e^{\frac{(r^2+4r)\norm{y-x}-\norm{y-x}^2}{2\sigma^2}}dy\notag \\
&\leq e^{\frac{(r^2+4r)(1+\norm{x})+2d(r^2+4r)^2}{2\sigma^2}},\notag
\end{align}
where we used the exponential moment bound in Lemma \ref{lem:normsubgauss}. 
Hence
\begin{align}
\norm{e^{F_{b_{r,\sigma}}}*\varphi_{\sigma}}_{2,\nu} \vee \norm{e^{F_{b_{r,\sigma}}}*\varphi_{\sigma}}_{\infty,\nu} \leq e^{\frac{(r^2+4r)\left(2dr^2+(8d+1)r+1\right)}{2\sigma^2}}=: \xi_{r,d,\sigma}.\label{eq:devterm1}
\end{align}
Substituting the above bounds, we obtain for $\sigma \leq 1$ that
\begin{align}
&t_{n,\varphi_{\sigma},\mu,\nu}\notag \\
& \lesssim J(1,F_{b_{r,\sigma}}*\varphi_{\sigma},\cF_{\varphi_{\sigma}}(b_{r,\sigma}),\mu) \norm{F_{b_{r,\sigma}}*\varphi_{\sigma}}_{2,\mu}+   J(1,e^{F_{b_{r,\sigma}}}*\varphi_{\sigma},\cE_{\varphi_{\sigma}}(b_{r,\sigma}),\nu) \norm{e^{F_{b_{r,\sigma}}}*\varphi_{\sigma}}_{2,\nu}  \notag\\
& \lesssim (3ed)^{3(d+2)} (r+1)^{\frac d2+1} \sigma^{-\frac{d}{2}} \left(\frac{(r^2+4r)(1+r\sigma\sqrt{d})}{\sigma^2}+ e^{\frac{(r^2+4r)(2dr^2+(8d+1)r+1)}{2 \sigma^2}}\right)\notag \\
& \lesssim    \tilde \xi_{r,d,\sigma}  \xi_{r,d,\sigma}. \notag 
\end{align} 
Applying \eqref{eq:KLdevineq-compsupp} and taking supremum over $\mu,\nu \in \cP\big(\BB_d(r)\big)$, we obtain \eqref{eq:KLdevineq-compsupp-simp}.

\medskip

Next, consider Part $(ii)$. Fix $\mu,\nu \in \mathsf{SG}(L)$. 
We have 
\begin{align}
    \frac{\hat \mu_n * \varphi_{\sigma}(x)}{\nu * \varphi_{\sigma}(x)}\mspace{-2 mu}=\mspace{-2 mu}\frac{\frac{1}{n}\sum_{i=1}^n e^{-\frac{\norm{x-X_i}^2}{2 \sigma^2}}}{\int_{\RR^d} e^{-\frac{\norm{x-y}^2}{2 \sigma^2}}d \nu(y)} \mspace{-2 mu}\stackrel{(a)}{\leq} \mspace{-2 mu}\frac{ e^{-\min_{i}\frac{\norm{x}^2+\norm{X_i}^2-2\norm{x} \norm{X_i}}{2 \sigma^2}}}{ \EE_{\nu}\left[e^{-\frac{(\norm{x}^2+\norm{Y}^2-2x \cdot Y)}{2 \sigma^2}}\right]} \mspace{-2 mu} \leq \mspace{-2 mu}\frac{ e^{\frac{\norm{x} \max_{i}\norm{X_i}}{ \sigma^2}}}{ \EE_{\nu}\left[e^{\frac{2x \cdot Y-\norm{Y}^2}{2 \sigma^2}}\right]} \mspace{-4 mu}\stackrel{(b)}{\leq } \mspace{-4 mu}\frac{ e^{\frac{\norm{x} \max_{i}\norm{X_i}}{ \sigma^2}}}{ e^{\frac{2 x \cdot  \EE_{\nu}[Y] -\EE_{\nu}\left[\norm{Y}^2\right]}{2 \sigma^2}}}, \label{eq:subgaussrat1} 
\end{align}
where $(a)$ follows by upper bounding average by the maximum and Cauchy-Schwarz, and $(b)$ is via an application of Jensens inequality to convex function $e^{x}$. Hence, we have
\begin{align}
\log\left(\frac{\hat \mu_n * \varphi_{\sigma}(x)}{\nu * \varphi_{\sigma}(x)}\right) &\leq \frac{\norm{x} \max_{i}\norm{X_i}}{ \sigma^2}+ \frac{\EE_{\nu}[\norm{Y}^2]+2 x \cdot  \EE_{\nu}[Y] }{2 \sigma^2} \notag \\
&\leq \frac{\norm{x} \max_{i}\norm{X_i}+(L^2+16dL)+ \norm{x}L}{ \sigma^2}, \notag
\end{align}
where we used that since 
\begin{align}
  \EE_{\nu}[\norm{Y}^2] \leq 2\left(\norm{\EE_{\nu}[Y]}^2+\EE_{\nu}[\norm{Y-\EE_{\nu}[Y]}^2]\right) \leq 2(L^2+16dL). \notag 
\end{align}
On the other hand,
\begin{align}
    \frac{\nu * \varphi_{\sigma}(x)}{\hat \mu_n * \varphi_{\sigma}(x)} =\frac{\int_{\RR^d} e^{-\frac{\norm{x-y}^2}{2 \sigma^2}}d \nu(y)}{\frac{1}{n}\sum_{i=1}^n e^{-\frac{\norm{x-X_i}^2}{2 \sigma^2}}} \stackrel{(a)}{\leq} \frac{ \EE_{\nu}\left[e^{-\frac{(\norm{x}^2+\norm{Y}^2-2x \cdot Y)}{2 \sigma^2}}\right]}{ e^{-\max_{i}\frac{\norm{x}^2+\norm{X_i}^2-2 x\cdot X_i}{2 \sigma^2}}} &\leq   e^{\max_{i}\frac{(\norm{X_i}^2-2x\cdot X_i)}{ 2\sigma^2}} \EE_{\nu}\left[e^{\frac{x \cdot Y}{ \sigma^2}}\right] \notag \\
    &\stackrel{(b)}{\leq } e^{\max_{i}\frac{(\norm{X_i}^2-2x\cdot X_i)}{ 2\sigma^2}} e^{\frac{L\norm{x}^2}{2 \sigma^4}+\frac{L\norm{x}}{ \sigma^2}}, \label{eq:subgaussrat2}
\end{align}
where $(a)$ follows by lower bounding the average by the minimum, and $(b)$ is because $\nu \in \mathsf{SG}(L)$ implies
\begin{align}
\EE_{\nu}\left[e^{\frac{x \cdot Y}{ \sigma^2}}\right] \leq  e^{\frac{\norm{x}^2 L}{\sigma^4} +\frac{L\norm{x}}{\sigma^2} }. \notag
\end{align}

From \eqref{eq:subgaussrat1} and \eqref{eq:subgaussrat2}, 
\begin{align}
    \abs{\log \left(\frac{\hat \mu_n * \varphi_{\sigma}(x)}{\nu * \varphi_{\sigma}(x)}\right)} \leq \frac{L\norm{x}^2}{2\sigma^4}+\norm{x}\left(\frac{L+\max_i \norm{X_i}}{\sigma^2}\right)+\max_i \frac{\norm{X_i}^2}{2\sigma^2}+\frac{L^2+16dL}{\sigma^2}. \label{eq:bndloglik}
\end{align}
We will apply \eqref{eq:onesampspecexp}  to arrive at the desired conclusion. However, for this purpose, we require  the RHS of \eqref{eq:bndloglik} be bounded a.s. Since this does not hold for $\mu \in \mathsf{SG}(L)$ in general due to the presence of the term $\max_i \norm{X_i}^2$, we will use a truncation argument by considering the  event  
\begin{align}
    E_n(t):=\left\{\max_{1 \leq i \leq n} \norm{X_i} \leq t\right\}, \notag 
\end{align}
under which \eqref{eq:onesampspecexp} can be applied. The parameter $t$ is then scaled appropriately w.r.t. $n$ such that the probability of   $\bar E_n(t)$ (the complement event of $E_n(t)$) decays exponentially.

Since $X \sim \mu$ satisfies $\norm{X}_{\psi_p} \leq L$ for some $ p \geq 2$ and $ L\geq 1$, we have
from \cite[Lemma 2.2.2]{AVDV-book} (a prefactor of $K=2^{\frac 1p}$ suffices) that
\begin{align}
 \norm{\max_{1 \leq i \leq n} \norm{X_i}}_{\psi_p} \leq 2^{\frac 1 p} \left(\log(1+n)\right)^{\frac 1p} \norm{X}_{\psi_p}. \notag 
\end{align}
By Markov's inequality and definition of Orlicz norm $\psi_p$, 
\begin{align}
    \PP \left(\max_{1 \leq i \leq n} \norm{X_i} \geq t\right) \leq 2 e^{-\frac{t^p}{\norm{\max_{1 \leq i \leq n} \norm{X_i}}_{\psi_p}^p}}. \label{eq:orlicznormineq}
\end{align}
Setting 
\begin{align}
  t=T_{n,p,\kappa,L}:=2^{\frac 1 p}  L n^{\kappa}, \notag  
\end{align}
for some $\kappa \geq 0$, we obtain from \eqref{eq:orlicznormineq} that
\begin{align}
     \PP \big(\bar E_n(T_{n,p,\kappa,L})\big) \leq 2 e^{-\frac{n^{p\kappa }}{\log(1+n)}}. \notag   
\end{align}
For any $\bar t \geq 0$, we obtain
\begin{align}
& \PP\left(\big|\kl{\hat \mu_n*\gamma_{\sigma}}{ \nu*\gamma_{\sigma}}-\kl{\mu*\gamma_{\sigma}}{\nu*\gamma_{\sigma}}\big| \geq \bar t ~\right) \notag\\
& \leq  \PP\left(\big|\kl{\hat \mu_n*\gamma_{\sigma}}{\nu*\gamma_{\sigma}}-\kl{\mu*\gamma_{\sigma}}{\nu*\gamma_{\sigma}}\big| \geq \bar t ~\Big|    E_n\big(T_{n,p,\kappa,L}\big)\right) + \PP\left(\bar{E}_n\big(T_{n,p,\kappa,L}\big)\right)\notag \\
& \leq  \PP\left(\big|\kl{\hat \mu_n*\gamma_{\sigma}}{ \nu*\gamma_{\sigma}}-\kl{\mu*\gamma_{\sigma}}{\nu*\gamma_{\sigma}}\big| \geq \bar t ~ \Big|    E_n\big(T_{n,p,\kappa,L}\big)\right) +  2e^{-\frac{n^{p \kappa}}{\log(1+n)}}.\label{eq:unionbndorlicz}
\end{align}
Next, we bound the first term in the RHS above. Note that given  $E_n\big(T_{n,p,\kappa,L}\big)$ holds, we have  $X^n \sim \tilde \mu^{\otimes n}$, where 
\begin{align}
\tilde \mu(\cA)  = \frac{\mu\left(\cA \cap \BB_d(T_{n,p,\kappa,L})\right) }{\mu\left(\BB_d(T_{n,p,\kappa,L})\right)},~\forall ~\cA \in \cB(\RR^d). \notag
\end{align}
Then,  \eqref{eq:bndloglik} yields  that under this event
\begin{align}
\abs{f^{\star}_{\hat \mu_n * \varphi_{\sigma},\hat \nu_n * \varphi_{\sigma}}(x)} \leq \frac{L\norm{x}^2}{2\sigma^4}+\norm{x}\left(\frac{L+T_{n,p,\kappa,L}}{\sigma^2}\right)+ \frac{T_{n,p,\kappa,L}^2}{2\sigma^2}+\frac{L^2+16dL}{\sigma^2}. \notag
\end{align}
Further
\begin{align}
    \frac{\mu * \varphi_{\sigma}(x)}{\nu * \varphi_{\sigma}(x)} =\frac{\int_{\RR^d} e^{-\frac{\norm{x-y}^2}{2 \sigma^2}}d \mu(y)}{\int_{\RR^d} e^{-\frac{\norm{x-y}^2}{2 \sigma^2}}d \nu(y)} \leq \frac{ \EE_{\mu}\left[e^{\frac{x \cdot Y}{ \sigma^2}}\right]}{\EE_{\nu}\left[e^{\frac{2x \cdot Y-\norm{Y}^2}{2 \sigma^2}}\right]} &\stackrel{(a)}{\leq}  \frac{e^{\frac{\norm{x}^2L}{2\sigma^4}+\frac{\norm{x}L}{\sigma^2}}}{e^{\frac{2x\cdot \EE_{\nu}[Y]-\EE_{\nu}\left[\norm{Y}^2\right]}{2\sigma^2}}} \notag \\
    & = e^{\frac{\norm{x}^2L}{2\sigma^4}+\frac{\norm{x}L}{\sigma^2}+\frac{\EE_{\nu}\left[\norm{Y}^2\right]-2x\cdot \EE_{\nu}[Y] }{2 \sigma^2}} \notag \\
    &  \stackrel{(b)}{\leq}e^{\frac{\norm{x}^2L}{2\sigma^4}+\frac{(L^2+16dL)+2\norm{x}L }{\sigma^2}}, \notag
\end{align}
where $(a)$ used that $\mu \in \mathsf{SG}(L)$ as well as  Jensen's inequality applied to convex function $e^x$, while $(b)$ is due to $\nu \in \mathsf{SG}(L)$ . Similar bound also holds for $\nu * \varphi_{\sigma}(x)/\mu * \varphi_{\sigma}(x)$. Hence, given $E_n\big(T_{n,p,\kappa,L}\big)$ holds,  \eqref{eq:bndloglik} and above equation implies
\begin{align}
    \abs{f^{\star}_{\mu * \varphi_{\sigma},\nu * \varphi_{\sigma}}(x)} &\leq \frac{L\norm{x}^2}{2\sigma^4}+\frac{2\norm{x}(L+T_{n,p,\kappa,L})}{\sigma^2}+\frac{2L^2+32dL+T_{n,p,\kappa,L}^2 }{2\sigma^2} \notag \\
    &\leq 2\left(\frac{T_{n,p,\kappa,L}+8dL}{\sigma^2}\right)^2 \big(1+\norm{x}^2\big). \notag
\end{align}
It follows that 
\begin{align}
   f^{\star}_{\mu * \varphi_{\sigma},\nu * \varphi_{\sigma}},f^{\star}_{\mu * \varphi_{\sigma},\nu * \varphi_{\sigma}} \in \cF\big(\bar b_{n,d,p,\kappa,\sigma,L},2\big) \mbox{ with } \bar b_{n,d,p,\kappa,\sigma,L}=4 \left(\frac{T_{n,p,\kappa,L}+8dL}{\sigma^2}\right)^2.\notag
\end{align}
Let $\bar F_{n,d,p,\kappa,\sigma,L}:=0.5 \bar b_{n,d,p,\kappa,\sigma,L}\big(1+\norm{x}^2\big)$ be the envelope for $\cF\big(\bar b_{n,d,p,\kappa,\sigma,L},2\big)$ and set $\beta=\frac d2+1$. 
For every probability measure $\gamma$ such that $\supp(\gamma) \subseteq \supp(\tilde \mu)$, we have from \eqref{eq:coventholderprob} that
\begin{align}
   \log N\Big(\epsilon \norm{\bar F_{n,d,p,\kappa,\sigma,L}*\varphi_{\sigma}}_{2,\gamma}, \cF_{\varphi_{\sigma}}\big(\bar b_{n,d,p,\kappa,\sigma,L},2\big), \norm{\cdot}_{2,\gamma}\Big) \lesssim d^{6d+9}(9e^2)^d \big(T_{n,p,\kappa,L}+1\big)^{d+4} \sigma^{-d} \epsilon^{-\frac{d}{\beta}}. \notag
\end{align}
  Hence
\begin{align}
&J(1,\bar F_{n,d,p,\kappa,\sigma,L}*\varphi_{\sigma},\cF_{\varphi_{\sigma}}\big(\bar b_{n,d,p,\kappa,\sigma,L},2\big),\tilde \mu) \notag \\
&:= \int_{0}^{1}\sup_{\gamma} \sqrt{\log N\Big(\epsilon \norm{\bar F_{n,d,p,\kappa,\sigma,L}*\varphi_{\sigma}}_{2,\gamma}, \cF_{\varphi_{\sigma}}\big(\bar b_{n,d,p,\kappa,\sigma,L},2\big), \norm{\cdot}_{2,\gamma}\Big)} d\epsilon  \notag \\
 & \lesssim (3ed)^{3(d+2)} \big(T_{n,p,\kappa,L}+1\big)^{\frac d2+2} \sigma^{-\frac{d}{2}} \notag \\
 & \leq (3ed)^{3(d+2)} 2^{\frac{d}{2}+2} \Big(2^{\frac 1 p}L\Big)^{\frac{d}{2}+2} n^{\kappa\left(\frac{d}{2}+2\right)}\sigma^{-\frac{d}{2}}. \notag
\end{align}
By using $\norm{y}^2 \leq 2\big(\norm{x}^2+\norm{x-y}^2\big)$, we have
\begin{align}
    \abs{\bar F_{n,d,p,\kappa,\sigma,L}*\varphi_{\sigma}(x)} &\leq \int_{\RR^d} 0.5 \bar b_{n,d,p,\kappa,\sigma,L}\big(1+2\norm{x}^2+2\norm{x-y}^2\big) \varphi_{\sigma}(x-y) dy  \notag \\
    &\lesssim  \bar b_{n,d,p,\kappa,\sigma,L} \left(0.5+\norm{x}^2+d \sigma^2\right). \notag 
\end{align} 
Hence
\begin{align}
    \norm{\bar F_{n,d,p,\kappa,\sigma,L}*\varphi_{\sigma}}_{\infty,\tilde \mu} \leq \bar b_{n,d,p,\kappa,\sigma,L} \left(0.5+T_{n,p,\kappa,L}^2+d \sigma^2\right) \lesssim L^4 d^3 \sigma^{-4} n^{4\kappa}. \notag
\end{align}
Further, since
\begin{align}
&\norm{\bar F_{n,d,p,\kappa,\sigma,L}}_{p,\tilde \mu*\gamma_{\sigma}}^p \notag \\
&=\int \bar F_{n,d,p,\kappa,\sigma,L}^p(x)\mspace{4 mu} \tilde \mu*\varphi_{\sigma}(x) dx \notag  \\
& =\int_{\BB_d(T_{n,p,\kappa,L})}\int_{\RR^d}  \frac{\bar b_{n,d,p,\kappa,\sigma,L}^p}{2^p}\big(1+2\norm{y}^2+2\norm{x-y}^2\big)^p \varphi_{\sigma}(x-y) dx \mspace{2 mu}d\tilde \mu(y) \notag \\
& \leq \int_{\BB_d(T_{n,p,\kappa,L})}\int_{\RR^d}  \frac{\bar b_{n,d,p,\kappa,\sigma,L}^p}{2}\left(\big(1+2\norm{y}^2\big)^p+2^p\norm{x-y}^{2p} \right)\varphi_{\sigma}(x-y) dx \mspace{2 mu}d\tilde \mu(y), \notag \\
& \leq \frac{\bar b_{n,d,p,\kappa,\sigma,L}^p}{2}\big(1+2T_{n,p,\kappa,L}^2\big)^p+\int_{\BB_d(T_{n,p,\kappa,L})}\int_{\RR^d}  \frac{\bar b_{n,d,p,\kappa,\sigma,L}^p}{2^{1-p}}\norm{x-y}^{2p} \varphi_{\sigma}(x-y) dx \mspace{2 mu}d\tilde \mu(y) \notag \\
& \leq 2^{2p-1}\bar b_{n,d,p,\kappa,\sigma,L}^p T_{n,p,\kappa,L}^{2p}+2^{2p-1} \bar b_{n,d,p,\kappa,\sigma,L}^p (dp)^p \sigma^{2p}, \notag \\
& \lesssim \big(2^{12} L^4 d^3 p n^{4 \kappa}\sigma^{-4}\big)^p, \notag 
\end{align}
we have
\begin{align}
&\norm{\bar F_{n,d,p,\kappa,\sigma,L}*\varphi_{\sigma}}_{2,\tilde \mu}^2  \mspace{-4 mu}=\mspace{-4 mu}\int \mspace{-2 mu} \big(\bar F_{n,d,p,\kappa,\sigma,L} *\varphi_{\sigma}(x)\big)^2 d\tilde \mu(x) \mspace{-2 mu}\leq \mspace{-2 mu} \int \mspace{-2 mu}\bar F_{n,d,p,\kappa,\sigma,L}^2 *\varphi_{\sigma}(x) d\tilde \mu(x)\mspace{-2 mu}\lesssim \mspace{-2 mu} L^8 d^6  n^{8 \kappa}\sigma^{-8}, \notag \\
& J\big(1,\bar F_{n,d,p,\kappa,\sigma,L}*\varphi_{\sigma},\cF_{\varphi_{\sigma}}\big(\bar b_{n,d,p,\kappa,\sigma,L},2\big),\tilde \mu\big) \norm{\bar F_{n,d,p,\kappa,\sigma,L}*\varphi_{\sigma}}_{2,\tilde \mu}  \notag \\
& \qquad\qquad\qquad\qquad\qquad\qquad\qquad\qquad\qquad\qquad\lesssim (3ed)^{3(d+2)}4^{\frac d2+2} L^{\frac d2+6}d^3\sigma^{-\left(\frac{d}{2}+4\right)} n^{\kappa\left(\frac{d}{2}+6\right)}. \notag
\end{align}
Let $0 \leq \tau \leq 1 $. Set  $ \kappa_{d,\tau}^{\star}:=\frac{\tau}{d+12}$, 
  \begin{align}  
  \hat{\xi}_{d,L,\sigma}&:= (3ed)^{3(d+2)}(4L)^{\frac d2+2} \sigma^{-\frac{d}{2}}  \quad \mbox{and} \quad 
\check{\xi}_{d,L,\sigma}:=L^4d^3\sigma^{-4}. \label{eq:constrateexp} 
\end{align}   
Then, substituting $\kappa=\kappa_{d,\tau}^{\star}$ in  the above equations yields
\begin{align}
& J\big(1,\bar F_{n,d,p,\kappa,\sigma,L}*\varphi_{\sigma},\cF_{\varphi_{\sigma}}\big(\bar b_{n,d,p,\kappa,\sigma,L},2\big),\tilde \mu\big) \norm{\bar F_{n,d,p,\kappa,\sigma,L}*\varphi_{\sigma}}_{2,\tilde \mu} \lesssim    \hat{\xi}_{d,L,\sigma}  \check{\xi}_{d,L,\sigma}n^{\frac{\tau}{2}}. \notag  
\end{align}
Applying \eqref{eq:onesampspecexp} in Remark \ref{Rem:onesamprate} and simplifying, we obtain  
\begin{align}
&\PP\left(\big|\kl{\hat \mu_n*\gamma_{\sigma}}{\nu*\gamma_{\sigma}}-\kl{\mu*\gamma_{\sigma}}{\nu*\gamma_{\sigma}}\big| \gtrsim \check{\xi}_{d,L,\sigma}  \left(\hat{\xi}_{d,L,\sigma}n^{-\frac 12(1-\tau)}+z\right) \Big|    E_n\big(T_{n,p,\kappa_{d,\tau}^{\star},L}\big)\right) \notag \\
&\qquad \qquad\qquad\qquad\qquad\qquad\qquad\qquad\qquad\qquad\qquad\qquad\qquad\qquad\qquad\qquad\qquad \leq  e^{-\check{\zeta}_{n,d,\tau}(z)}, \notag
\end{align}
with \begin{align}
\check{\zeta}_{n,d,\tau}(z):=n^{\frac{d+12-8\tau}{d+12}}z^2 \wedge n^{\frac{d+12-4\tau}{d+12}}z. \notag
\end{align}
Finally, since $\hat{\xi}_{d,L,\sigma} \geq 1$, we obtain from the above equation and \eqref{eq:unionbndorlicz}  with   $\kappa=\kappa_{d,\tau}^{\star}$ and $\bar t =c\mspace{2 mu}\check{\xi}_{d,L,\sigma}  \big(\hat{\xi}_{d,L,\sigma}n^{-\frac 12(1-\tau)}+z\big)$  (for an appropriate constant $c$)  that 
\begin{align}
& \PP\left(\big|\kl{\hat \mu_n*\gamma_{\sigma}}{ \nu*\gamma_{\sigma}}-\kl{\mu*\gamma_{\sigma}}{\nu*\gamma_{\sigma}}\big| \gtrsim \check{\xi}_{d,L,\sigma}  \big(\hat{\xi}_{d,L,\sigma}n^{-\frac 12(1-\tau)}+z\big)\right) \leq 3 e^{-\hat{\zeta}_{n,d,p,\tau}(z)}, \notag 
\end{align}
where $\hat{\zeta}_{n,d,p,\tau}(z)$ is as given in \eqref{eq:rateorlicsg}.
Taking supremum over all $\mu,\nu \in \mathsf{SG}(L)$ such that $\norm{X}_{\psi_p} \leq L$ yields \eqref{eq:KLdevineq-subgauss}, thus completing the proof.
\subsection{Proof of Theorem \ref{Thm:KLdev-symmradexp}} \label{Sec:Thm:KLdev-symmradexp-proof}

\begin{lemma}[Deviation inequality in terms of Rademacher expectation]\label{Lem:statverdevineq-KL-gen}
Consider the setting of Lemma \ref{Lem:concent-KL}. Let
\begin{align} Z^{(\varepsilon)}_{n,\sigma} &:=\EE_{\varepsilon}\left[\sup_{f \in \cF_{n,\varphi_{\sigma}}}\abs{\sum_{i=1}^n \varepsilon_i f(X_i)} \right] \quad \mbox{ and } \quad 
\bar Z^{(\varepsilon)}_{n,\sigma} :=\EE_{\varepsilon}\left[\sup_{f \in \cE_{n,\varphi_{\sigma}}}\abs{\sum_{i=1}^n \varepsilon_i f(X_i)} \right]. \notag
\end{align}
 Then 
\begin{align}
&\PP \left(  \abs{\kl{\hat \mu_n*\gamma_{\sigma}}{\hat \nu_n*\gamma_{\sigma}}-\kl{\mu*\gamma_{\sigma}}{\nu*\gamma_{\sigma}}} \geq 7  n^{-1} \big(Z^{(\varepsilon)}_{n,\sigma}+\bar Z^{(\varepsilon)}_{n,\sigma}\big)+2z\right)  \notag \\
&   \leq  2 e^{-\Big(\frac{9n z^2}{128  \norm{F_n*\varphi_{\sigma}}_{2,\mu}^2} \wedge \frac{n^{\frac{p-1}{p}}z}{228\mspace{2 mu}\norm{F_n*\varphi_{\sigma}}_{p,\mu}}\Big)}+2e^{-\Big(\frac{9nz^2}{128  \norm{e^{F_n}*\varphi_{\sigma}}_{2,\mu}^2} \wedge \frac{nz}{228\norm{e^{F_n}*\varphi_{\sigma}}_{p,\nu}}\Big)}\notag \\
& \qquad \qquad \qquad \qquad \qquad \qquad \qquad  +2 (16)^p\left(\norm{F_n*\varphi_{\sigma}}_{p,\mu}^p+\norm{e^{F_n}*\varphi_{\sigma}}_{p,\nu}^p\right)n^{1-p}.\label{eq:KLdevineq-genOrlicz-radexp}
\end{align}    
In particular, for compactly supported $\mu,\nu$, 
\begin{align}
&\PP \left(  \abs{\kl{\hat \mu_n*\gamma_{\sigma}}{\hat \nu_n*\gamma_{\sigma}}-\kl{\mu*\gamma_{\sigma}}{\nu*\gamma_{\sigma}}} \geq 7n^{-1}\big(Z^{(\varepsilon)}_{n,\sigma}+\bar Z^{(\varepsilon)}_{n,\sigma}\big)  +2z\right)  \notag\\
& \qquad \qquad \qquad  \leq e^{-\left(\frac{nz^2}{ 8\norm{F_n*\varphi_{\sigma}}_{2,\mu}^2} \wedge \frac{nz}{22\norm{F_n*\varphi_{\sigma}}_{\infty,\mu}}\right)} + e^{-\left(\frac{nz^2}{ 8\norm{e^{F_n}*\varphi_{\sigma}}_{2,\mu}^2} \wedge \frac{nz}{22\norm{e^{F_n}*\varphi_{\sigma}}_{\infty,\nu}}\right)}. \label{eq:KLdevineq-compsupp-radexp}
\end{align}
\end{lemma}
The proof of Lemma \ref{Lem:statverdevineq-KL-gen} follows similar to Lemma \ref{Lem:concent-KL} up to \eqref{eq:devsplitineq}, and then applies \eqref{eq:statFN-int-env} along with $\overline{M_n} \leq \norm{M_n(X^n)}_{p,\mu} \leq n^{\frac 1p} \norm{F_n}_{p,\mu}$  to obtain \eqref{eq:KLdevineq-genOrlicz-radexp}, and \eqref{eq:eq:statist-fuk-nagaev-simp} with $\eta=1+\sqrt{6}$ to derive \eqref{eq:KLdevineq-compsupp-radexp}. 

\medskip

Theorem \ref{Thm:KLdev-symmradexp} follows from \eqref{eq:KLdevineq-compsupp-radexp} by noting that for $\mu,\nu \in \cP\left(\BB_d(r)\right)$,  
and  function classes $\cF_{\varphi_{\sigma}}(b_{r,\sigma})$ and $\cE_{\varphi_{\sigma}}(b_{r,\sigma})$ with envelopes $F_{b_{r,\sigma}}*\varphi_{\sigma}$ and $e^{F_{b_{r,\sigma}}}*\varphi_{\sigma}$, respectively, we have 
\begin{align}
\norm{F_{b_{r,\sigma}}*\varphi_{\sigma}}_{2,\mu} \vee \norm{F_{b_{r,\sigma}}*\varphi_{\sigma}}_{\infty,\mu} \vee \norm{e^{F_{b_{r,\sigma}}}*\varphi_{\sigma}}_{2,\nu} \vee  \norm{e^{F_{b_{r,\sigma}}}*\varphi_{\sigma}}_{\infty,\nu}    \leq \xi_{r,d,\sigma}, \notag 
\end{align}
where $\xi_{r,d,\sigma}$ is given in \eqref{eq:devterm1}.
\subsection{Proof of Theorem \ref{Thm:compactlysuppRen}}\label{Sec:Thm:compactlysuppRen-proof}
Fix $\mu,\nu \in \cP\left(\BB_d(r)\right)$. As shown in \eqref{eq:loglikfncls},  $f_{\hat \mu_n*\gamma_{\sigma},\hat \nu_n*\gamma_{\sigma}}^{\star}, f_{\mu*\gamma_{\sigma},\nu*\gamma_{\sigma}}^{\star} \in \cF\big(b_{r,\sigma}\big)$, where $b_{r,\sigma}:=\frac{r^2+4r}{\sigma^2}$. Also, recall that   $F_{b_{r,\sigma}}(x) :=0.5 b_{r,\sigma}\big(1+\norm{x}\big)$ is an envelope for $\cF\big(b_{r,\sigma}\big)$. 
Let
\begin{align}
  L_{\alpha,\sigma,\hat \mu_n}(f):=  \int_{\RR^d} \left(e^{\alpha f} *\varphi_{\sigma} \right)  \mspace{2 mu}d \hat \mu_n. \notag 
\end{align}
Note that for arbitrary functionals $A:\cF \rightarrow \RR$ and $B:\cF \rightarrow  \RR$ such that at least one of $\sup_{f \in \cF}A(f)$ and $\sup_{f \in \cF}B(f)$ is finite, we have
\begin{align}
    \abs{\sup_{f \in \cF}A(f)-\sup_{f \in \cF}B(f)} \leq \sup_{f \in \cF}\abs{A(f)-B(f)}. \label{eq:diffsupbnd} 
\end{align}
Using this and observing  that with probability one,
\begin{align}
  0 < \sup_{f \in \cF(b_{r,\sigma})} \log L_{\alpha-1,\sigma,\hat \mu_n}, \sup_{f \in \cF(b_{r,\sigma})} \log L_{\alpha-1,\sigma,\mu}, \sup_{f \in \cF(b_{r,\sigma})} \log L_{\alpha,\sigma,\hat \nu_n}, \sup_{f \in \cF(b_{r,\sigma})} \log L_{\alpha,\sigma,\nu}  <\infty, \notag
\end{align}
 we obtain 
\begin{align}
 & \abs{\rendiv{\hat \mu_n*\gamma_{\sigma}}{\hat \nu_n*\gamma_{\sigma}}{\alpha}-\rendiv{\mu*\gamma_{\sigma}}{\nu*\gamma_{\sigma}}{\alpha}} \notag \\
 &\leq \sup_{f \in \cF(b_{r,\sigma})} \frac{\alpha}{\abs{\alpha-1}} \abs{\log L_{\alpha-1,\sigma,\hat \mu_n}(f)-\log L_{\alpha-1,\sigma,\mu}(f)}  +\sup_{f \in \cF(b_{r,\sigma})}\abs{\log L_{\alpha,\sigma,\hat \nu_n}(f) -\log L_{\alpha,\sigma,\nu}(f)} \notag \\
  &\leq \frac{\alpha}{\abs{\alpha-1}}\sup_{f \in\cF(b_{r,\sigma})}  \frac{\abs{L_{\alpha-1,\sigma,\hat \mu_n}(f)-L_{\alpha-1,\sigma,\mu}(f)}}{L_{\alpha-1,\sigma,\hat \mu_n}(f) \wedge L_{\alpha-1,\sigma,\mu}(f)}+\sup_{f \in \cF(b_{r,\sigma})}  \frac{\abs{L_{\alpha,\sigma,\hat \nu_n}(f)-L_{\alpha,\sigma,\nu}(f)}}{L_{\alpha,\sigma,\hat \nu_n}(f) \wedge L_{\alpha,\sigma,\nu}(f)} \notag \\
&\leq \frac{\alpha}{\abs{\alpha-1}}  \frac{\sup_{f \in \cF(b_{r,\sigma})}\abs{L_{\alpha-1,\sigma,\hat \mu_n}(f)-L_{\alpha-1,\sigma,\mu}(f)}}{L_{-|\alpha-1|,\sigma,\hat \mu_n}(F_{b_{r,\sigma}}) \wedge L_{-|\alpha-1|,\sigma,\mu}(F_{b_{r,\sigma}})}+  \frac{\sup_{f \in \cF(b_{r,\sigma})}\abs{L_{\alpha,\sigma,\hat \nu_n}(f)-L_{\alpha,\sigma,\nu}(f)}}{L_{-\alpha,\sigma,\hat \nu_n}(F_{b_{r,\sigma}}) \wedge L_{-\alpha,\sigma,\nu}(F_{b_{r,\sigma}})}, \label{eq:diffrendiv}  
\end{align}
where the penultimate inequality used $\abs{\log \frac{x}{y}} \leq \frac{|x-y|}{x \wedge y}$. Note that since $\supp(\mu) \subseteq \BB_d(r)$, we have
\begin{align}
    L_{-|\alpha-1|,\sigma,\mu}\big(F_{b_{r,\sigma}}\big)&:=\int_{\RR^d} \left(e^{-\abs{\alpha-1}F_{b_{r,\sigma}}} *\varphi_{\sigma} \right)  \mspace{2 mu}d \mu \geq e^{-\abs{\alpha-1}\norm{F_{b_{r,\sigma}}}_{1, \mu*\gamma_{\sigma}}} \geq \frac{1}{\kappa_{|\alpha-1|,r,d,\sigma}}, \notag
\end{align}
 where we used that $\EE[\norm{Z}]\mspace{-2 mu}=\mspace{-2 mu}\sqrt{d}\sigma$ for $Z \mspace{-2 mu}\sim \mspace{-2 mu} N(0,\sigma^2I_d)$ and $\kappa_{\alpha,r,d,\sigma}\mspace{-2 mu}:=\mspace{-2 mu}e^{\frac{ \alpha b_{r,\sigma}(1+r+\sqrt{d}\sigma)}{2}}$. 
 Similarly, using  $\supp(\hat \mu_n),\supp(\hat \nu_n),$ $\supp(\nu) \subseteq \BB_d(r)$, we can lower bound the terms in the denominator of the expression in \eqref{eq:diffrendiv}, which upon substitution therein yields
 \begin{align}
 &  \abs{\rendiv{\hat \mu_n*\gamma_{\sigma}}{\hat \nu_n*\gamma_{\sigma}}{\alpha}-\rendiv{\mu*\gamma_{\sigma}}{\nu*\gamma_{\sigma}}{\alpha}} \notag \\
 &\leq  \frac{\alpha \kappa_{|\alpha-1|,r,d,\sigma}}{\abs{\alpha-1}}  \sup_{f \in \cF(b_{r,\sigma})}\abs{L_{\alpha-1,\sigma,\hat \mu_n}(f)-L_{\alpha-1,\sigma,\mu}(f)}+ \kappa_{\alpha,r,d,\sigma} \sup_{f \in \cF(b_{r,\sigma})}\abs{L_{\alpha,\sigma,\hat \nu_n}(f)-L_{\alpha,\sigma,\nu}(f)} \notag  \\
  &=  \sup_{f \in \bar \cE_{\varphi_{\sigma},\alpha,r,d}} \abs{(\hat \mu_n-\mu)(f)} 
  +  \sup_{f \in \tilde  \cE_{\varphi_{\sigma},\alpha,r,d}}\abs{(\hat \nu_n-\nu)(f)}. \notag
 \end{align}
Here 
\begin{align}
\bar \cE_{\varphi_{\sigma},\alpha,r,d}&:=\left\{\frac{\alpha \kappa_{|\alpha-1|,r,d,\sigma}}{\abs{\alpha-1}}e^{|\alpha-1| f}*\varphi_{\sigma}: f\in  \cF(b_{r,\sigma})\right\} \notag \\ \mbox{and} \quad 
\tilde  \cE_{\varphi_{\sigma},\alpha,r,d}&:=\left\{ \kappa_{\alpha,r,d,\sigma}e^{\alpha f}*\varphi_{\sigma}: f\in  \cF(b_{r,\sigma})\right\}, \notag
\end{align}
are Gaussian-smoothed function classes with envelopes $\bar E_{\sigma,\alpha,r,d}*\varphi_{\sigma}$ and $ \tilde E_{\sigma,\alpha,r,d}*\varphi_{\sigma}$, respectively, where 
\begin{align}
\bar E_{\sigma,\alpha,r,d}&:=\frac{\alpha \kappa_{|\alpha-1|,r,d,\sigma}}{\abs{\alpha-1}}e^{|\alpha-1| F_{b_{r,\sigma}}} \quad \mbox{and} \quad
    \tilde E_{\sigma,\alpha,r,d} :=\kappa_{\alpha,r,d,\sigma}e^{\alpha F_{b_{r,\sigma}}}.\notag 
\end{align} 
It follows analogous to \eqref{eq:entintsmthfncls} that  
\begin{flalign}  
J(1,\bar E_{\sigma,\alpha,r,d}*\varphi_{\sigma},\bar \cE_{\varphi_{\sigma},\alpha,r,d},\mu) \vee J(1,\tilde  E_{\sigma,\alpha,r,d}*\varphi_{\sigma},\tilde  \cE_{\varphi_{\sigma},\alpha,r,d},\nu)  &\lesssim \tilde \xi_{r,d,\sigma}, \notag  \notag 
\end{flalign}
where $\tilde \xi_{r,d,\sigma}$ is defined in \eqref{eq:devterm2}.  

Next, note that
\begin{align}
\bar E_{\sigma,\alpha,r,d}*\varphi_{\sigma}(x)&= \frac{1}{(2\pi)^{\frac d2}\sigma^d}\int_{\RR^d}\frac{\alpha \kappa_{|\alpha-1|,r,d,\sigma}}{\abs{\alpha-1}}e^{|\alpha-1| F_{b_{r,\sigma}}(y)}e^{-\frac{\norm{y-x}^2}{2\sigma^2}}dy \notag \\
&\leq \frac{\alpha \kappa_{|\alpha-1|,r,d,\sigma}}{\abs{\alpha-1}} e^{\frac{|\alpha-1|b_{r,\sigma}(1+\norm{x})}{2}} \EE\left[e^{\frac{|\alpha-1|b_{r,\sigma}\norm{Z}}{2}}\right]\notag \\
&\leq \frac{\alpha \kappa_{|\alpha-1|,r,d,\sigma}}{\abs{\alpha-1}} e^{\frac{|\alpha-1|b_{r,\sigma}(1+\norm{x})}{2}+d|\alpha-1|^2 b_{r,\sigma}^2 \sigma^2 } ,\notag 
\end{align}
where we used the exponential moment bound in Lemma \ref{lem:normsubgauss}. 
Hence
\begin{align}
\norm{\bar E_{\sigma,\alpha,r,d}*\varphi_{\sigma}}_{2,\mu} \vee \norm{\bar E_{\sigma,\alpha,r,d}*\varphi_{\sigma}}_{\infty,\mu} &\leq  \frac{\alpha \kappa_{|\alpha-1|,r,d,\sigma}}{\abs{\alpha-1}} e^{\frac{|\alpha-1|b_{r,\sigma}(1+r)}{2}+d|\alpha-1|^2 b_{r,\sigma}^2 \sigma^2 } \notag \\
&\leq  \frac{\alpha }{\abs{\alpha-1}} e^{|\alpha-1|b_{r,\sigma}\left(1+r+\frac{\sqrt{d}\sigma}{2}+d |\alpha-1|b_{r,\sigma}\sigma^2\right)}. \notag 
\end{align}
Likewise, we have
\begin{align}
   \norm{\tilde E_{\sigma,\alpha,r,d}*\varphi_{\sigma}}_{2,\nu} \vee \norm{\tilde E_{\sigma,\alpha,r,d}*\varphi_{\sigma}}_{\infty,\nu} & \leq \kappa_{\alpha,r,d,\sigma}e^{\frac{\alpha b_{r,\sigma}(1+r)}{2}+d \alpha^2 b_{r,\sigma}^2 \sigma^2 } \leq e^{\alpha b_{r,\sigma}\left(1+r+\frac{\sqrt{d}\sigma}{2}+d\alpha b_{r,\sigma} \sigma^2\right) }. \notag 
\end{align}
 Substituting the above bounds, we obtain for $\sigma \leq 1$ that
\begin{flalign}
&t_{n,\varphi_{\sigma},\mu,\nu} \notag \\
&:=\mspace{-2 mu}J(1,\bar E_{\sigma,\alpha,r,d}*\varphi_{\sigma},\bar \cE_{\varphi_{\sigma},\alpha,r,d},\mu) \norm{\bar E_{\sigma,\alpha,r,d}*\varphi_{\sigma}}_{2,\mu}  \mspace{-2 mu}  +\mspace{-2 mu} J(1,\tilde  E_{\sigma,\alpha,r,d}*\varphi_{\sigma},\tilde  \cE_{\varphi_{\sigma},\alpha,r,d},\nu)  \norm{\tilde E_{\sigma,\alpha,r,d}*\varphi_{\sigma}}_{2,\nu}  \notag\\
& \lesssim \tilde \xi_{r,d,\sigma} \left(\frac{\alpha }{\abs{\alpha-1}} e^{|\alpha-1|b_{r,\sigma}\left(1+r+\frac{\sqrt{d}\sigma}{2}+d |\alpha-1|b_{r,\sigma}\sigma^2\right)}+ e^{\alpha b_{r,\sigma}\left(1+r+\frac{\sqrt{d}\sigma}{2}+d\alpha b_{r,\sigma} \sigma^2\right) }\right)\notag \\
& = \tilde \xi_{r,d,\sigma}\mspace{2 mu}\lambda_{\alpha,d,r,\sigma},  \notag  &&
\end{flalign}
where
\begin{align}
   \lambda_{\alpha,d,r,\sigma}:= \left(\frac{\alpha }{\abs{\alpha-1}}+1\right)  e^{(\alpha+1) b_{r,\sigma}\left(1+r+\frac{\sqrt{d}\sigma}{2}+d(\alpha+1) b_{r,\sigma} \sigma^2\right) }. \label{eq:devtermren}
\end{align}
Since
\begin{align}
    \norm{\bar E_{\sigma,\alpha,r,d}*\varphi_{\sigma}}_{2,\mu} \vee \norm{\bar E_{\sigma,\alpha,r,d}*\varphi_{\sigma}}_{\infty,\mu} \vee \norm{\tilde E_{\sigma,\alpha,r,d}*\varphi_{\sigma}}_{2,\nu} \vee \norm{\tilde E_{\sigma,\alpha,r,d}*\varphi_{\sigma}}_{\infty,\nu} \leq \lambda_{\alpha,d,r,\sigma}, \notag
\end{align}
following similar steps as leading to \eqref{eq:KLdevineq-compsupp} and substituting the above bounds in the resulting inequality, we obtain \eqref{eq:Renconcineq-Orlicz}.
\subsection{Proof of Corollary \ref{Cor:Rendivsmoothrate}} \label{Sec:Cor:Rendivsmoothrate-proof}
From \eqref{eq:onesampspecexp}, we have for $0<\sigma \leq 1$ and every $\mu,\nu \in \cP\left(\BB_d(r)\right)$,
\begin{flalign}
&\PP \left(\abs{\kl{\hat \mu_n*\gamma_{\sigma}}{\nu*\gamma_{\sigma}}-\kl{\mu*\gamma_{\sigma}}{\nu*\gamma_{\sigma}}} \gtrsim \bar c_{d,r}\sigma^{-\left(\frac{d
}{2}+2\right)} n^{-\frac 12}  +\frac{(r^2+4r)(1+r\sqrt{d})}{\sigma^2}z\right)  \notag \\
& \qquad \qquad\qquad\qquad\qquad\qquad\qquad\qquad\qquad\qquad\qquad\qquad\qquad\qquad\qquad\qquad\quad\leq e^{-\left(nz^2 \wedge nz\right)}, \notag
\end{flalign}
where  $\bar c_{d,r}:=(3ed)^{3(d+2)}\sqrt{d}r(r+1)^{\frac{d}{2}+3}$.
On the other hand, \eqref{eq:KL-stability-ub} yields
\begin{align}
    \abs{\kl{\mu}{\nu}-\kl{\mu*\gamma_{\sigma}}{\nu*\gamma_{\sigma}}}\leq c_{d,s} M\left(M+1+\log M\right) \sigma^s \leq 2c_{d,s} M^2 \sigma^s. \notag
\end{align}
Combining these, and setting $\sigma=n^{-\frac{1}{d+2s+4}} \leq 1$ and $z=n^{\frac{-(s+2)}{d+2s+4}}\bar z$, \eqref{eq:KLconcineq-smth} follows via  triangle inequality.
\subsection{Proof of Proposition \ref{prop:smoothRen-stability}}\label{prop:smoothRen-stability-proof}
The inequality \eqref{eq:KL-stability-ub} follows via same steps as  in \cite[Lemma 3]{SGK-IT-2023} up to replacing Lebesgue measure by $\rho$ and taking $p_{\mu}$, $p_{\nu}$ as densities w.r.t. $\rho$. We next show \eqref{eq:ren-stability-ub} using the same techniques therein. For $\tau \in [0,1]$, let $z_{\mu,\tau}(x):=(1-\tau)p_{\mu}(x)+\tau \mu*\varphi_{\sigma}(x) $ and $z_{\nu,\tau}(x):=(1-\tau)p_{\nu}(x)+\tau\nu*\varphi_{\sigma}(x)$. By Taylor's theorem,  
\begin{align}
\left(\mu*\varphi_{\sigma}(x)\right)^{\alpha}\left(\nu*\varphi_{\sigma}(x)\right)^{1-\alpha}
&=p_{\mu}(x)^{\alpha}p_{\nu}(x)^{1-\alpha}+\alpha \big(\mu*\varphi_{\sigma}(x)-p_{\mu}(x)\big)\int_{0}^1\bigg(\frac{z_{\mu,\tau}(x)}{z_{\nu,\tau}(x)}\bigg)^{\alpha-1}d\tau \notag  \\
&\qquad \qquad  +(1-\alpha)\big(\nu*\varphi_{\sigma}(x)-p_{\nu}(x)\big)\int_{0}^1 \bigg(\frac{z_{\mu,\tau}(x)}{z_{\nu,\tau}(x)}\bigg)^{\alpha}d\tau.\label{eq:taylexpren} 
\end{align} 
Since  $1/M \leq p_{\mu}(x)/p_{\nu}(x) \leq M$ $\rho$-a.e. by assumption, we have $1/M \leq \mu*\varphi_{\sigma}(x)/\nu*\varphi_{\sigma}(x) \leq M$ $\rho$-a.e. Hence,  
\begin{align}
\frac 1M &\leq \frac{p_{\mu}(x)}{p_{\nu}(x)} \wedge \frac{\mu*\varphi_{\sigma}(x)}{\nu*\varphi_{\sigma}(x)} \leq \frac{z_{\mu,\tau}(x)}{z_{\nu,\tau}(x)} \leq \frac{p_{\mu}(x)}{p_{\nu}(x)} \vee \frac{\mu*\varphi_{\sigma}(x)}{\nu*\varphi_{\sigma}(x)}\leq M, ~~~\rho~\mbox{a.e.}\label{eq:bndratioden}
\end{align} 
Consider $\alpha \in (0,1)$.  Integrating w.r.t. $\rho$ in \eqref{eq:taylexpren}, we  obtain (note that $M \geq 1$)
\begin{align}
\abs{e^{(\alpha-1)\rendiv{\mu}{\nu}{\alpha}}-e^{(\alpha-1)\rendiv{\mu*\gamma_{\sigma}}{\nu*\gamma_{\sigma}}{\alpha}} } \mspace{-2 mu}
&\mspace{-2 mu}\leq \mspace{-2 mu}\alpha  M^{1-\alpha}\mspace{-4 mu}\int_{\RR^d}\mspace{-2 mu}\abs{\mu*\varphi_{\sigma}-p_{\mu}} d\rho  +M^{\alpha}(1-\alpha)  \mspace{-4 mu}\int_{\RR^d}\mspace{-2 mu}\abs{\nu*\varphi_{\sigma}-p_{\nu}} d\rho \notag  \\
& \leq c_{d,s}\big(\alpha M^{2-\alpha}+(1-\alpha)M^{\alpha+1}\big) \sigma^s. \label{eq:exprenapprox1}
\end{align}
In the final inequality, we used that
\begin{flalign}
 & \left(\int_{\RR^d} \abs{p_{\mu}-\mu*\varphi_{\sigma}}d \rho \right)\vee  \left(\int_{\RR^d} \abs{p_{\nu}-\nu*\varphi_{\sigma}}d\rho \right)\leq  c_{d,s} M  \sigma^s, \notag 
\end{flalign}
which follows under the Lipschitz assumption on the densities similarly to the proof of \cite[Lemma 3]{SGK-IT-2023}.
On the other hand, for $\alpha \in (1,\infty)$,  integrating w.r.t. $\rho$ in \eqref{eq:taylexpren} leads to
\begin{align}
\abs{e^{(\alpha-1)\rendiv{\mu}{\nu}{\alpha}}-e^{(\alpha-1)\rendiv{\mu*\gamma_{\sigma}}{\nu*\gamma_{\sigma}}{\alpha}} } 
&\mspace{-2 mu}\leq \mspace{-2 mu}\alpha  M^{\alpha-1}\mspace{-4 mu}\int_{\RR^d}\mspace{-2 mu}\abs{\mu*\varphi_{\sigma}-p_{\mu}} d\rho  \mspace{-2 mu}+\mspace{-2 mu}M^{\alpha}(\alpha-1) \mspace{-2 mu} \int_{\RR^d}\mspace{-2 mu}\abs{\nu*\varphi_{\sigma}-p_{\nu}} d\rho \notag \\
& \leq c_{d,s} \big(\alpha M^{\alpha}+(\alpha-1)M^{\alpha+1}\big) \sigma^s. \label{eq:exprenapprox2}
\end{align}
Next, observe that \eqref{eq:bndratioden} implies  
 \begin{align} \label{eq:exprenlb}
   & e^{(\alpha-1)\rendiv{\mu}{\nu}{\alpha}}=\int p_{\mu}(x)^{\alpha}p_{\nu}(x)^{1-\alpha} d\rho  \geq  M^{-\alpha} \quad \mbox{ and } \quad  e^{(\alpha-1)\rendiv{\mu*\gamma_{\sigma}}{\nu*\gamma_{\sigma}}{\alpha}} \geq M^{-\alpha}.
\end{align}   
 Applying the inequality $\big|\log x-\log y\big| \leq \abs{x-y}/(x \wedge y)$ for $x,y >0$  with $x=e^{(\alpha-1)\rendiv{\mu}{\nu}{\alpha}}$,  $y=e^{(\alpha-1)\rendiv{\mu*\gamma_{\sigma}}{\nu*\gamma_{\sigma}}{\alpha}}$ and using \eqref{eq:exprenapprox1}-\eqref{eq:exprenlb}  yields the desired claim. 

\subsection{Proof of Theorem \ref{Thm:devineqren}} \label{Sec:Thm:devineqren-proof}
\textbf{Part $(i)$:} 
Using \eqref{eq:varexpestklunreg} and \eqref{eq:diffsupbnd},
\begin{align}
 \abs{\klest{\hat \mu_n}{\hat \nu_n}{\cG_{b,\beta,\rho}}-\klest{\mu}{\nu}{\cG_{b,\beta,\rho}}}  &\leq \sup_{g \in \cG_{b,\beta,\rho}}\abs{(\hat \mu_n-\mu)(g)-(\hat \nu_n-\nu)(e^g)} \notag \\
& \leq \sup_{g \in \cG_{b,\beta,\rho}}\abs{(\hat \mu_n-\mu)(g)}+\sup_{g \in \cG_{b,\beta,\rho}}\abs{(\hat \nu_n-\nu)(e^g)} \notag \\
&=\frac{Z_{n,\cG_{b,\beta,\rho}}+\tilde Z_{n,\cE_{b,\beta,\rho}}}{n}, \label{eq:uppbnddevkl}
\end{align} 
where $\cG_{b,\beta,\rho}$ is as given in Theorem \ref{Thm:devineqren},  $\cE_{b,\beta,\rho}=\{e^g: g \in \cG_{b,\beta,\rho}\}$,  
\begin{align}
    &Z_{n,\cG_{b,\beta,\rho}}:=\sup_{g \in \cG_{b,\beta,\rho}}\abs{\sum_{i=1}^n \big(g(X_i)-\EE_{\mu}[g]\big)} \quad \mbox{ and } \quad \tilde Z_{n,\cE_{b,\beta,\rho}}:=\sup_{g \in \cE_{b,\beta,\rho}}\abs{\sum_{i=1}^n \big(g(Y_i)-\EE_{\nu}[g]\big)}. \label{eq:diffvalsemp} 
\end{align}
To obtain the desired deviation inequality, we  will use Theorem \ref{Thm:Talagrand-Bosquet}. To this end, we require an upper bound $\EE[Z_{n,\cG_{b,\beta,\rho}}]$ and $\EE[\tilde Z_{n,\cE_{b,\beta,\rho}}]$, for which we use the result stated next.
\begin{theorem}[see e.g., Theorem 2.14.1 in \cite{AVDV-book}] \label{Thm:Dudley-Pollard}
 Suppose  $\cF$ satisfies Assumption \ref{Assump:funclass} with envelope $F \in L^p(\mu)$ for some $p \geq 2$.  For $X^n \sim \mu^{\otimes n}$,  let $Z_n$ be as given in \eqref{eq:empsumfrmmean}  (with $\cF_n=\cF$ for all $n\in \NN$). Then
    \begin{align}
       \EE \left[Z_n\right] \leq \left(\EE[Z_n^p]\right)^{\frac 1p} \leq c_p 
 \sqrt{n}J(1,F,\cF,\mu) \norm{F}_{p,\mu},   \label{eq:entropyintmain}
    \end{align}
where  $c_p >0$ is a constant. 
\end{theorem}
Recall that $\mu,\nu \ll \rho$, $b=\log M$, $\beta>d/2$ and $a:=\frac{d}{2 \beta} <1$. From the definition of $J(1,F,\cF,\mu)$ (see \eqref{eq:coventrdef}) with $\cF=\cG_{b,\beta,\rho}$ and envelope $F$ such that $F=b$ $\rho$-a.e., we have using \eqref{eq:coveringbnd} that
\begin{align}
 J(1,b,\cG_{b,\beta,\rho},\mu) &:= \int_{0}^{1}\sup_{\gamma} \sqrt{\log N\Big(\epsilon b, \cG_{b,\beta,\rho}, \norm{\cdot}_{2,\gamma}\Big)} d\epsilon   \leq \sqrt{\hat c_{d,\beta}}\int_{0}^{1} \epsilon^{-\frac{d}{2\beta}} d\epsilon  \leq \frac{\sqrt{\hat c_{d,\beta}}}{1-a}.\notag
\end{align}
 It follows from \eqref{eq:entropyintmain} that
\begin{align}
\EE\left[Z_{n,\cG_{b,\beta,\rho}}\right] &\lesssim \sqrt{n}  J(1,b,\cG_{b,\beta,\rho},\mu) \mspace{2 mu}b\leq \frac{\sqrt{n\hat c_{d,\beta}}b}{1-a} . \notag
\end{align}
Observe that for $z_1,z_2 \in \RR$ such that $\abs{z_1} \vee \abs{z_2} \leq b$, we have $\abs{e^{z_1}-e^{z_2}} \leq e^b\abs{z_1-z_2}$. Since this implies that for $g_1,g_2 \in \cG_{b,\beta,\rho}$, $\norm{e^{g_1}-e^{g_2}}_{2,\gamma}  \leq e^b \norm{g_1-g_2}_{2,\gamma}$ for any probability measure $\gamma$ with $\supp(\gamma) \subseteq \supp(\rho)$, we obtain 
\begin{align}
    N\big(\epsilon e^b, \cE_{b,\beta,\rho}, \norm{\cdot}_{2,\gamma}\big) \leq N\big(\epsilon, \cG_{b,\beta,\rho}, \norm{\cdot}_{2,\gamma}\big). \notag
\end{align}
Consequently
\begin{align}
J(1,e^b,\cE_{b,\beta,\rho},\nu) &:= \int_{0}^{1}\sup_{\gamma} \sqrt{\log N\Big(\epsilon e^b, \cE_{b,\beta,\rho}, \norm{\cdot}_{2,\gamma}\Big)} d\epsilon \notag \\
 &\leq \int_{0}^{1}\sup_{\gamma} \sqrt{\log N\Big(\epsilon, \cG_{b,\beta,\rho}, \norm{\cdot}_{2,\gamma}\Big)} d\epsilon \notag \\
 &\leq \frac{\sqrt{\hat c_{d,\beta}} b^a}{1-a}  ,\notag
\end{align}
where we again used \eqref{eq:coveringbnd} to obtain the last inequality. Then,  \eqref{eq:entropyintmain} yields 
\begin{align} \EE\left[\tilde Z_{n,\cE_{b,\beta,\rho}}\right] \lesssim   \sqrt{n}J\big(1,e^b,\cE_{b,\beta,\rho},\nu\big) e^b \leq  \frac{\sqrt{n\hat c_{d,\beta}} b^a e^b}{1-a}. \notag
\end{align}
In order to apply \eqref{eq:talagrandconcmain}, we also use the following bounds: 
\begin{align}
    \sqrt{V_n}&:=\sqrt{n \theta^2+2b \EE\left[Z_{n,\cG_{b,\beta,\rho}}\right]} \leq \sqrt{n} \theta+ \sqrt{2b\EE\left[Z_{n,\cG_{b,\beta,\rho}}\right]} \lesssim \sqrt{n} b+ \sqrt{2}b (\hat c_{d,\beta})^{\frac 14}(1-a)^{-\frac 12} n^{\frac 14}, \notag \\ 
    \sqrt{\tilde V_n} &:=\sqrt{n \tilde \theta^2+2e^b\EE[\tilde Z_{n,\cE_{b,\beta,\rho}}]} \leq \sqrt{n} \tilde \theta+ \sqrt{2e^b\EE[\tilde Z_{n,\cE_{b,\beta,\rho}}]}   \leq \sqrt{n} e^b+ \sqrt{2}e^b b^{\frac{a}{2}}(\hat c_{d,\beta})^{\frac 14}(1-a)^{-\frac 12}  n^{\frac 14}. \notag 
\end{align}
Setting $c_{d,\beta}:=c\frac{\sqrt{\hat c_{d,\beta}}}{1-a}$ for an appropriate constant $c \geq 1$ and noting that the above bounds yields $\sqrt{V_n} \leq \sqrt{c_{d,\beta}n} b$ and  $\sqrt{\tilde V_n} \leq \sqrt{c_{d,\beta}n} e^b(1+b^{\frac a2})$, 
\begin{align}
 &  \PP\left(\abs{\klest{\hat \mu_n}{\hat \nu_n}{\cG_{b,\beta,\rho}}-\klest{ \mu}{ \nu}{\cG_{b,\beta,\rho}}} \geq c_{d,\beta} (b+b^ae^b)n^{-\frac{1}{2}} +4\sqrt{2c_{d,\beta}}e^b\big(1+b^{\frac a2}\big)z \right) \notag \\
  &  \leq \PP\left(Z_{n,\cG_{b,\beta,\rho}}+\tilde Z_{n,\cE_{b,\beta,\rho}} \geq c_{d,\beta} (b+b^ae^b)\sqrt{n} +4\sqrt{2c_{d,\beta}}e^b\big(1+b^{\frac a2}\big)nz \right) \notag \\
    &  \leq \PP\left(Z_{n,\cG_{b,\beta,\rho}}-\EE\big[Z_{n,\cG_{b,\beta,\rho}}\big] \geq 2\sqrt{2c_{d,\beta}}bnz\right)+ \PP\left(\tilde Z_{n,\cE_{b,\beta,\rho}}-\EE\big[\tilde Z_{n,\cE_{b,\beta,\rho}}\big] \geq  2\sqrt{2c_{d,\beta}}e^b\big(1+b^{\frac a2}\big)nz\right)  \notag \\
 & \leq  2 e^{-(nz^2 \wedge nz)}. \label{eq:kldivbndtal}
\end{align}
Here,  the first inequality follows from \eqref{eq:uppbnddevkl} while  the final step is via an application of   \eqref{eq:talagrandconcmain}.

Next, observe that   for $(\mu,\nu) \in \bar \cP_{M,\rho}(\cG_{b,\beta,\rho},\delta)$,  there exists   $f \in \cG_{b,\beta,\rho}$ (by definition) such that $\norm{f_{\mu,\nu}^{\star}-f}_{\infty,\rho} \leq \delta$. This implies 
\begin{align}
\abs{\klest{ \mu}{ \nu}{\cG_{b,\beta,\rho}}-\kl{\mu}{\nu}}&=\kl{\mu}{\nu}-\klest{ \mu}{ \nu}{\cG_{b,\beta,\rho}} \notag \\
&= \EE_{\mu}[f_{\mu,\nu}^{\star}]-\EE_{\nu}[e^{f_{\mu,\nu}^{\star}}]+1-\klest{ \mu}{ \nu}{\cG_{b,\beta,\rho}} \notag \\
& \leq \EE_{\mu}[f_{\mu,\nu}^{\star}-f]-\EE_{\nu}[e^{f_{\mu,\nu}^{\star}}-e^f] \notag \\
& \leq \EE_{\mu}\big[\abs{f_{\mu,\nu}^{\star}-f}\big]+M \EE_{\nu}\left[\abs{f_{\mu,\nu}^{\star}-f}\right] \notag \\
&\leq (M+1)\delta, \notag
\end{align}
where in the penultimate inequality, we  used  $\norm{f}_{\infty,\rho} \leq b=\log M$ and the fact that the Lipschitz constant of  $e^x$ is bounded by $M$ for $\abs{x} \leq \log M$.  
Hence, using triangle inequality and adjusting the constant $c_{d,\beta}$, we have
\begin{align}
 &  \PP\left(  \abs{\klest{\hat \mu_n}{\hat \nu_n}{\cG_{b,\beta,\rho}}-\kl{\mu}{\nu}} \geq (M+1)\delta+ c_{d,\beta} M\left(1+(\log M)^a\right)\big(n^{-\frac 12}+z\big)  \right)  \leq 2 e^{-( nz^2 \wedge nz)}. \notag
\end{align}
Taking supremum w.r.t. $(\mu,\nu) \in\bar \cP_{M,\rho}(\cG_{b,\beta,\rho},\delta)$ completes the proof of \eqref{eq:devineqkl}.

\medskip
\noindent
\textbf{Part $(ii)$:}
We have
\begin{align}
  & \abs{\renest{\hat \mu_n}{\hat \nu_n}{\alpha}{\cG_{b,\beta,\rho}}-\renest{\mu}{\nu}{\alpha}{\cG_{b,\beta,\rho}}} \notag \\
    &  \leq \sup_{g \in\cG_{b,\beta,\rho}}\frac{\alpha}{\abs{\alpha-1}}\abs{\log \left(\frac{\EE_{\hat \mu_n}\big[ e^{(\alpha-1)g}\big]}{\EE_{\mu}\big[ e^{(\alpha-1)g}\big]}\right)}+\sup_{g \in \cG_{b,\beta,\rho}}\abs{\log \left(\frac{\EE_{\hat \nu_n}\big[ e^{\alpha g}\big]}{\EE_{\nu}\big[ e^{\alpha g}\big]}\right)} \notag \\
  & \leq \sup_{g \in \cG_{b,\beta,\rho}}\frac{\alpha e^{\abs{\alpha-1}b}}{\abs{\alpha-1}}\abs{\EE_{\hat \mu_n}\big[ e^{(\alpha-1)g}\big]-\EE_{\mu}\big[ e^{(\alpha-1)g}\big]}+e^{\alpha b}\sup_{g \in \cG_{b,\beta,\rho}}\abs{\EE_{\hat \nu_n}\big[ e^{\alpha g}\big]-\EE_{\nu}\big[ e^{\alpha g}\big]} \notag \\
   &=\sup_{g \in \bar \cG_{b,\alpha,\beta,\rho}} \abs{(\hat \mu_n-\mu)(g)} +\sup_{g \in  \tilde \cG_{b,\alpha,\beta,\rho}} \abs{(\hat \nu_n-\nu)(g)} \notag \\
   &=\frac{Z_{n,\bar \cG_{b,\alpha,\beta,\rho}}+\tilde Z_{n,\tilde \cG_{b,\alpha,\beta,\rho}}}{n},\notag
\end{align}
where $Z_{n,\bar \cG_{b,\alpha,\beta,\rho}}$ and $\tilde Z_{n,\tilde \cG_{b,\alpha,\beta,\rho}}$ are as defined in \eqref{eq:diffvalsemp} with
  \begin{align} \notag
\bar \cG_{b,\alpha,\beta,\rho}&:= \left\{\frac{\alpha e^{\abs{\alpha-1}b}}{\abs{\alpha-1}} e^{(\alpha-1)g}: g \in \cG_{b,\beta,\rho}\right\} \quad \mbox{and} \quad
 \tilde \cG_{b,\alpha,\beta,\rho}:=\big\{e^{\alpha (b+g)}:g \in \cG_{b,\beta,\rho}\big\}. 
\end{align}  
In the first inequality above, we used  \eqref{eq:varexpestrenunreg} and  \eqref{eq:diffsupbnd}. In the second inequality, we used $\abs{\log \frac{x}{y}} \leq \frac{|x-y|}{x \wedge y}$ and that $\norm{g}_{\infty,\rho} \leq b=\log M $ for $g \in \cG_{b,\beta,\rho}$. 
Note that for any $c \in \RR$ and $x,y \in \RR$ such that $\abs{x} \vee \abs{y} \leq b$, $\abs{e^{c x}-e^{c y}} \leq  \abs{c}e^{\abs{c}b} \abs{x-y}$. Hence, for any probability measure $\gamma$ such that $\supp(\gamma) \subseteq \supp(\rho)$, 
  \begin{align} \notag
 N\Big(\alpha e^{2\abs{\alpha-1}b} \epsilon , \bar \cG_{b,\alpha,\beta,\rho}, \norm{\cdot}_{2,\gamma}\Big) &\leq N\Big(\epsilon , \cG_{b,\beta,\rho}, \norm{\cdot}_{2,\gamma}\Big), \notag \\
\mbox{and} \quad
 N\Big(\alpha e^{2\alpha b} \epsilon , \tilde \cG_{b,\alpha,\beta,\rho}, \norm{\cdot}_{2,\gamma}\Big) &\leq N\Big(\epsilon , \cG_{b,\beta,\rho}, \norm{\cdot}_{2,\gamma}\Big). \notag 
\end{align}  
Let $F$ be an envelope for $\bar \cG_{b,\alpha,\beta,\rho}$ such that $F=(\alpha/\abs{\alpha-1}) e^{2\abs{\alpha-1}b}$ $\rho$-a.e. Then, we have using \eqref{eq:coveringbnd} that for $\beta >d/2$,
\begin{align}
 J(1,F,\bar \cG_{b,\alpha,\beta,\rho},\mu) &:= \int_{0}^{1}\sup_{\gamma} \sqrt{\log N\Big(\epsilon F, \bar \cG_{b,\alpha,\beta,\rho}, \norm{\cdot}_{2,\gamma}\Big)} d\epsilon   \leq \frac{\sqrt{\hat c_{d,\beta}}}{1-a} \big(b\abs{\alpha-1}\big)^{a}. \notag
\end{align}
Similarly, considering $F$ such that $F= e^{2\alpha b}$ $\rho$-a.e. to be an envelope for $\tilde \cG_{b,\alpha,\beta,\rho}$,  we obtain for $\beta >d/2$ that
\begin{align}
  J(1,F,\tilde \cG_{b,\alpha,\beta,\rho},\nu) &:= \int_{0}^{1}\sup_{\gamma} \sqrt{\log N\Big(\epsilon F, \tilde  \cG_{b,\alpha,\beta}, \norm{\cdot}_{2,\gamma}\Big)} d\epsilon   \leq \frac{\sqrt{\hat c_{d,\beta}}}{1-a} (b\alpha)^{a}.  \notag 
\end{align}
Consequently,  \eqref{eq:entropyintmain}  yields (with respective envelopes $F$) that
\begin{align}
\EE\left[Z_{n,\bar \cG_{b,\alpha,\beta,\rho}}\right] &\lesssim  \sqrt{n} J\big(1,F,\bar \cG_{b,\alpha,\beta,\rho},\mu\big) \norm{F}_{2,\mu}\leq \frac{\sqrt{\hat c_{d,\beta}}}{1-a} \sqrt{n}\alpha b^{a}\abs{\alpha-1}^{a-1}e^{2\abs{\alpha-1}b}, \notag \\
\EE\left[\tilde Z_{n,\tilde \cG_{b,\alpha,\beta,\rho}}\right] &\lesssim \sqrt{n}  J\big(1,F,\tilde \cG_{b,\alpha,\beta,\rho},\nu\big) \norm{F}_{2,\mu}\leq \frac{\sqrt{\hat c_{d,\beta}}}{1-a} \sqrt{n}(b \alpha)^{a}e^{2\alpha b}. \notag 
\end{align}
Moreover
\begin{align}
    \sqrt{V_n} \leq \sqrt{n} \frac{\alpha e^{2\abs{\alpha-1}b}}{\abs{\alpha-1}} + \sqrt{2\frac{\alpha e^{2\abs{\alpha-1}b}}{\abs{\alpha-1}}\EE\big[Z_{n,\bar \cG_{b,\alpha,\beta,\rho}}\big]} 
    &\leq \sqrt{n}\frac{\alpha e^{2\abs{\alpha-1}b}}{\abs{\alpha-1}} + \frac{\sqrt{2 c_{d,\beta}}\alpha e^{2\abs{\alpha-1}b}b^{\frac a2}n^{\frac 14}}{\abs{\alpha-1}^{1-\frac a2}} \notag \\
    &\leq  \frac{\sqrt{nc_{d,\beta}}\alpha e^{2\abs{\alpha-1}b}}{\abs{\alpha-1}}\left(1+(b\abs{\alpha-1})^{\frac a2}\right), \notag
\end{align}
and \begin{align}
   &  \sqrt{\tilde V_n} \leq \sqrt{n} e^{2\alpha b} + \sqrt{2b\EE\big[\tilde Z_{n,\tilde  \cG_{b,\alpha,\beta}}\big]} \leq \sqrt{n} e^{2\alpha b} + \sqrt{2  c_{d,\beta}}e^{2\alpha b}(b\alpha)^{\frac a2}n^{\frac 14}  \leq  \sqrt{n c_{d,\beta}}e^{2\alpha b} \big(1+(b\alpha)^{\frac a2}\big) . \notag
\end{align}
Applying \eqref{eq:talagrandconcmain}, it follows similar to \eqref{eq:kldivbndtal} that for $z \geq 0\mspace{2 mu}$:
\begin{align}
 &  \PP\left(\abs{\renest{\hat \mu_n}{\hat \nu_n}{\alpha}{\cG_{b,\beta,\rho}}-\renest{ \mu}{ \nu}{\alpha}{\cG_{b,\beta,\rho}}} \geq c_{d,\beta} \big(n^{-\frac 12}+z\big) v_{e^b,\alpha,a}\right) \leq 2e^{-( nz^2 \wedge nz )},  \label{eq:talgrndineqrenunreg}
\end{align}
where $v_{e^b,\alpha,a}:=\frac{\alpha e^{2\abs{\alpha-1}b}}{\abs{\alpha-1}}\left(1+(b\abs{\alpha-1})^{\frac a2}\right)+ e^{2\alpha b}\big(1+(b\alpha)^a\big)$.

Next, for $(\mu,\nu) \in \bar \cP_{M,\rho}(\cG_{b,\beta,\rho},\delta)$,  there exists an $f \in \cG_{b,\beta,\rho}$ such that $\norm{f_{\mu,\nu}^{\star}-f}_{\infty,\rho} \leq \delta$. This implies
\begin{align}
&\abs{\renest{ \mu}{ \nu}{\alpha}{\cG_{b,\beta,\rho}}-\rendiv{\mu}{\nu}{\alpha}}\notag \\
&=\rendiv{\mu}{\nu}{\alpha}-\renest{ \mu}{ \nu}{\alpha}{\cG_{b,\beta,\rho}} \notag \\
& \leq \frac{\alpha}{\alpha-1} \left(\log \EE_{\hat \mu}\big[ e^{(\alpha-1)f_{\mu,\nu}^{\star}}\big]-\log \EE_{\hat \mu}\big[ e^{(\alpha-1)f}\big]\right)-\left(\log \EE_{\nu}\big[e^{\alpha f_{\mu,\nu}^{\star}}\big]-\log \EE_{\nu}\big[e^{\alpha f}\big]\right) \notag \\
& \leq \alpha e^{2\abs{\alpha-1}b} \EE_{\mu}\big[\abs{f_{\mu,\nu}^{\star}-f}\big]+\alpha e^{2\alpha b} \EE_{\nu}\left[\abs{f_{\mu,\nu}^{\star}-f}\right] \notag \\
&\leq \alpha  (M^{2\abs{\alpha-1}}+M^{2\alpha})\delta, \notag 
\end{align}
where in the penultimate inequality, we used that the Lipschitz constant of $\log x$  is bounded by $1/\abs{c}$ for $x\geq c>0$ and that of $e^x$ is bounded by $e^c$ for $x \in [-c,c]$, and in the last equation, we used that $\norm{f_{\mu,\nu}^{\star}-f}_{\infty,\rho} \leq \delta$ and $b=\log M$. The claim in \eqref{eq:devineqren} follows via triangle inequality from \eqref{eq:talgrndineqrenunreg} and the above equation after simplification.

\subsection{Proof of Proposition \ref{prop:ht-audit}} \label{Sec:prop:ht-audit-proof}
Let $0 \leq  \tau  \leq 1$. Consider the test statistic $T_n=\renest{\hat \mu_n}{\hat \nu_n}{\alpha}{\cG_n}$ with critical value
\begin{align}
    t_n=\epsilon+\big(M^{2\alpha}+M^{2\abs{\alpha-1}}\big)\delta_n+ c_{d,\beta}  v_{M,\alpha,a}\big(n^{-\frac 12}+n^{0.5(\tau-1)}\big). \notag
\end{align}
 Define the event  $E_{n}:=\left\{\renest{\hat \mu_n}{\hat \nu_n}{\alpha}{\cG_n} > t_n\right\}$
and let $\bar E_{n}$ denote its compliment. 
Since $\rendiv{\mu_0}{\nu_0}{\alpha} \leq \epsilon$ by assumption of the hypotheses, we have
\begin{align}
&e_{1,n}(T_n,t_n)\notag \\
&=\PP\big(E_{n}|H_0\big)\notag \\
 & =  \PP\left(\renest{\hat \mu_n}{\hat \nu_n}{\alpha}{\cG_n} > \epsilon+\big(M^{2\alpha}+M^{2\abs{\alpha-1}}\big)\delta_n+ c_{d,\beta}  v_{M,\alpha,a}\big(n^{-\frac 12} +n^{0.5(\tau-1)}\big) \big| H_0\right) \notag \\
  & \leq  \PP\left(\abs{\renest{\hat \mu_n}{\hat \nu_n}{\alpha}{\cG_n}-\rendiv{\mu_0}{\nu_0}{\alpha}} \geq \big(M^{2\alpha}+M^{2\abs{\alpha-1}}\big)\delta_n+ c_{d,\beta}  v_{M,\alpha,a}\big(n^{-\frac 12} +n^{0.5(\tau-1)}\big)\big| H_0 \right) \notag \\
& \leq 2e^{-\left(n^{\tau} \wedge n^{0.5(\tau+1)} \right)},\notag
\end{align}
where the last inequality follows using \eqref{eq:devineqren}. Note that $n^{\tau} \wedge n^{0.5(\tau+1)}=n^{\tau}$ for $0\leq \tau \leq 1$.

On the other hand, since $\rendiv{\mu_1}{\nu_1}{\alpha} > \epsilon$ by assumption of the hypotheses,
\begin{align}
&e_{2,n}(T_n,t_n)\notag \\
 & =  \PP\left(\renest{\hat \mu_n}{\hat \nu_n}{\alpha}{\cG_n} \leq t_n\big| H_1\right) \notag \\
  & \leq  \PP\left(\renest{\hat \mu_n}{\hat \nu_n}{\alpha}{\cG_n}-\rendiv{\mu_1}{\nu_1}{\alpha} \leq - (\rendiv{\mu_1}{\nu_1}{\alpha}-t_n)\big| H_1\right) \notag \\
  & \leq  \PP\left(\abs{\renest{\hat \mu_n}{\hat \nu_n}{\alpha}{\cG_n}-\rendiv{\mu_1}{\nu_1}{\alpha}} \geq \rendiv{\mu_1}{\nu_1}{\alpha}-t_n\big| H_1 \right) \notag \\
    & =  \PP\left(\abs{\renest{\hat \mu_n}{\hat \nu_n}{\alpha}{\cG_n}-\rendiv{\mu_1}{\nu_1}{\alpha}} \geq \big(M^{2\alpha}+M^{2\abs{\alpha-1}}\big)\delta_n+ c_{d,\beta}  v_{M,\alpha,a}\big(n^{-\frac 12} +\bar \vartheta_{n,\alpha,\delta_n}(\tau)\big) \big| H_1 \right) \notag \\
& \leq 2e^{-\vartheta_{n,\alpha,\delta_n}(\tau)}, \notag
\end{align}
where in the final inequality, we applied \eqref{eq:devineqren} with $z=\bar \vartheta_{n,\alpha,\delta_n}(\tau)$,  and
\begin{subequations}\label{eq:expt2err}
 \begin{align}
\vartheta_{n,\alpha,\delta_n}(\tau)&:=  n \big(\bar \vartheta_{n,\alpha,\delta_n}(\tau) \vee 0\big)^2 \wedge n  \big(\bar\vartheta_{n,\alpha,\delta_n}(\tau) \vee 0),   \\
\bar \vartheta_{n,\alpha,\delta_n}(\tau)&:=\frac{1}{c_{d,\beta} v_{M,\alpha,a}}\left(  \rendiv{\mu_1}{\nu_1}{\alpha}-\epsilon -2\big(M^{2\alpha}+M^{2\abs{\alpha-1}}\big)\delta_n -2c_{d,\beta} v_{M,\alpha,a}n^{-\frac 12} \right)-  n^{0.5(\tau-1)} . 
\end{align}   
\end{subequations}
\section{Proofs of  Lemmas}
\subsection{Proof of Lemma \ref{Lem:concent-KL}} \label{Lem:concent-KL-proof}
To handle the case of compactly supported distributions, we  will use a version of Talagrand's concentration inequality given in \cite[Theorem 2.3]{BOUSQUET-2002}.
\begin{theorem}[Theorem 2.3 in \cite{BOUSQUET-2002}] \label{Thm:Talagrand-Bosquet}
Suppose that $\cF$ satisfies Assumption \ref{Assump:funclass} and   
$\norm{f}_{\infty,\mu} \leq b$ for all $f \in \cF$. For $X^n \sim \mu^{\otimes n}$, let  $Z_n$ be as given in \eqref{eq:empsumfrmmean} (with $\cF_n=\cF$ for all $n\in \NN$). Let  $\theta^2:=\sup_{f \in \cF} \EE_{\mu }\big[\big(f-\EE_{\mu}[f]\big)^2\big]$ and  $V_n:=n \theta^2+ 2b \EE[Z_n]$. Then, for all $z \geq 0$,
\begin{subequations}
      \begin{align}
      \PP \left(Z_n \geq \EE[Z_n]+\sqrt{2zV_n}+\frac{bz}{3} \right) \leq e^{-z}. \label{eq:Tal-bosqversion}
\end{align}
In particular, for every $z \geq 0$,
 \begin{align}
   &   \PP \left(Z_n \geq \EE[Z_n]+z \right) \leq e^{-\big(\frac{z^2}{8V_n} \wedge \frac{3z}{2b}\big)}, \label{eq:talagrandconcmain} \\
   & \PP \left(Z_n \geq 2\EE[Z_n]+z \right) \leq e^{-\big(\frac{z^2}{8n\theta^2} \wedge \frac{3z}{8b}\big)}. \label{eq:talagrandconctrunc}
 \end{align} 
 \end{subequations}
\end{theorem}
 Equation \eqref{eq:talagrandconcmain}  follows from \eqref{eq:Tal-bosqversion} by considering the cases $\sqrt{2zV_n} \geq bz/3$ and its opposite. Equation \eqref{eq:talagrandconctrunc} follows from \eqref{eq:Tal-bosqversion} by noting that $\sqrt{2zV_n} \leq \theta \sqrt{2nz} +2\sqrt{bz\EE[Z_n]} \leq \theta \sqrt{2nz}+bz+\EE[Z_n]$, where we used $\sqrt{x+y} \leq \sqrt{x}+\sqrt{y}$ for $x,y \geq 0$ and $2xy \leq x^2+y^2$ for $x,y \in \RR$, and considering the cases $\sqrt{2zn\theta^2} \geq 4bz/3$ and its opposite.

To show \eqref{eq:KLdevineq-mombnd}, we will use the following Fuk-Nagaev
type inequality which is applicable for a sequence of function classes $(\cF_n)_{n \in \NN}$ with unbounded envelope. It is a  restatement of  \cite[Theorem 2]{Adamczak-2010} in a form with explicit constants which is useful for our purposes.
\begin{theorem}[Theorem 2 in \cite{Adamczak-2010}] \label{Thm:Admczak-tailbnd} 
In the setting of Lemma \ref{Lem:statistver:Fuk-Nagaev}, if $F_n$ satisfies the weaker integrability condition $\norm{F_n}_{p,\mu}<\infty$ for some $p\geq 2$ (instead of $\norm{F_n}_{\psi_q}<\infty$ for some $q \in (0,1]$), 
     \begin{align}
      \PP \Big(Z_n \geq 2\EE[Z_n]+z \Big) \leq e^{-\big(\frac{z^2}{128 n \norm{F_n}_{2,\mu}^2} \wedge \frac{3z}{256\overline{M_n}}\big)} +2(16)^p \EE\big[\big(M_n(X^n)\big)^p\big] z^{-p}, \label{eq:tailbndepmomcond}
  \end{align}
where $\overline{M_n} \leq \norm{M_n(X^n)}_{p,\mu} \leq n^{\frac 1p} \norm{F_n}_{p,\mu}$.
\end{theorem}
The constants  in Theorem \ref{Thm:Admczak-tailbnd} are determined similar to that in Lemma \ref{Lem:statistver:Fuk-Nagaev}, and the upper bound on $\norm{M_n(X^n)}_{p,\mu}$ follows due to the maximal inequality for Orlicz norm with $\psi(x)=x^p$ (see equation above \cite[Lemma 2.2.2]{AVDV-book}).

\medskip

\noindent
Proceeding with the proof of \eqref{eq:KLdevineq-mombnd}, let  $X^n \sim \mu^{\otimes n}$, $Y^n \sim \nu^{\otimes n}$,  and  $\cF_n$ be the function class having the properties specified in Lemma \ref{Lem:concent-KL}. Recall that $f_{\mu*\Phi,\nu*\Phi}^{\star}:= \log \left(\mu *\phi/\nu *\phi \right)$. Since $f_{\hat \mu_n*\Phi,\hat \nu_n*\Phi}^{\star}, f_{\mu*\Phi,\nu*\Phi}^{\star} \in \cF_n$ by assumption, 
\begin{align}
 \sup_{f \in \cF_n}\mspace{-2 mu}\EE_{\hat \mu_n*\Phi}[f]-\EE_{\hat \nu_n*\Phi}[e^f]+1 \mspace{-2 mu}& \mspace{-2 mu}\leq  \kl{\hat \mu_n*\Phi}{\hat \nu_n*\Phi}\mspace{-2 mu}=\mspace{-2 mu}\EE_{\hat \mu_n*\Phi}\big[f_{\hat \mu_n*\Phi,\hat \nu_n*\Phi}^{\star}\big]\mspace{-2 mu}-\mspace{-2 mu}\EE_{\hat \nu_n*\Phi}\big[e^{f_{\hat \mu_n*\Phi,\hat \nu_n*\Phi}^{\star}}\big]\mspace{-2 mu}+\mspace{-2 mu}1, \notag \\
 \sup_{f \in \cF_n}\EE_{\mu*\Phi}[f]-\EE_{\nu*\Phi}[e^f]+1 &\leq  \kl{\mu*\Phi}{\nu*\Phi} =\EE_{\mu*\Phi}\big[f_{\mu*\Phi,\nu*\Phi}^{\star}\big]-\EE_{\nu*\Phi}\big[e^{f_{\mu*\Phi,\nu*\Phi}^{\star}}\big]+1. \notag
\end{align}
Hence
\begin{subequations} \label{eq:klvarfeq}
\begin{align}
  & \kl{\hat \mu_n*\Phi}{\hat \nu_n*\Phi}=\sup_{f \in \cF_n}\EE_{\hat \mu_n*\Phi}[f]-\EE_{\hat \nu_n*\Phi}[e^f]+1,   \label{eq:klvarfeq1} \\
  &  \kl{\mu*\Phi}{\nu*\Phi}=\sup_{f \in \cF_n}\EE_{\mu*\Phi}[f]-\EE_{\nu*\Phi}[e^f]+1. \label{eq:klvarfeq2}   
\end{align}
\end{subequations}
Moreover, since $\abs{f_{\mu*\Phi,\nu*\Phi}^{\star}(x) } \leq F_n(x)$ for all $x \in \RR^d$ and $\norm{F_n*\phi}_{p,\mu} <\infty$ for some $p \geq 2$ implies that $\norm{F_n}_{1,\mu*\Phi} =\norm{F_n*\phi}_{1,\mu}<\infty$, $\kl{\mu*\Phi}{\nu*\Phi} <\infty$. Then, from \eqref{eq:diffsupbnd} and \eqref{eq:klvarfeq}, we obtain
\begin{align}
 &\big|\kl{\hat \mu_n*\Phi}{\hat \nu_n*\Phi}-\kl{\mu*\Phi}{\nu*\Phi}\big| \notag \\
 &\leq \sup_{f \in \cF_n}\abs{(\hat \mu_n*\phi-\mu*\phi)(f)-(\hat \nu_n*\phi-\nu*\phi)(e^f)} \notag \\
& \leq \sup_{f \in \cF_n}\abs{(\hat \mu_n*\phi-\mu*\phi)(f)}+\sup_{f \in \cF_n}\abs{(\hat \nu_n*\phi-\nu*\phi)(e^f)} \notag \\
&= \frac{Z_{n,\cF_{n,\phi}} + \tilde Z_{n,\cE_{n,\phi}}}{n}, \notag
\end{align}
where 
\begin{align}
    &Z_{n,\cF_{n,\phi}}:=\sup_{f \in \cF_{n,\phi}}\abs{\sum_{i=1}^n \big(f(X_i)-\EE_{\mu}[f]\big)} \quad \mbox{ and } \quad \tilde Z_{n,\cE_{n,\phi}}:=\sup_{f \in \cE_{n,\phi}}\abs{\sum_{i=1}^n \big(f(X_i)-\EE_{\nu}[f]\big)}, \notag 
\end{align}
and  
$\cF_{n,\phi}$ and $\cE_{n,\phi}$ are as defined in Lemma \ref{Lem:concent-KL}. 
By an application of  Theorem \ref{Thm:Dudley-Pollard}, we obtain 
\begin{align} 
   \EE[Z_{n,\cF_{n,\phi}}] & \leq c  \sqrt{n} J(1,F_n*\phi,\cF_{n,\phi},\mu) \norm{F_n*\phi}_{2,\mu}, \notag \\
    \mbox{and} \quad
   \EE[\tilde Z_{n,\cE_{n,\phi}}]  &\leq c  \sqrt{n} J(1,e^{F_n}*\phi,\cE_{n,\phi},\nu) \norm{e^{F_n}*\phi}_{2,\nu}. \notag
\end{align}
 Now, we can upper bound the probability of deviation  as follows:
\begin{align}
&\PP\left(\big|\kl{\hat \mu_n*\Phi}{\hat \nu_n*\Phi}-\kl{\mu*\Phi}{\nu*\Phi}\big| \geq 2n^{-\frac 12}t_{n,\varphi_{\sigma},\mu,\nu} +z\right) \notag \\
& \leq \PP\left(\big|\kl{\hat \mu_n*\Phi}{\hat \nu_n*\Phi}-\kl{\mu*\Phi}{\nu*\Phi}\big| \geq 2n^{-1}\EE\left[Z_{n,\cF_{n,\phi}}\right]+2n^{-1}\EE\left[\tilde Z_{n,\cE_{n,\phi}}\right]+z\right) \notag \\
& \leq \PP\left(Z_{n,\cF_{n,\phi}} + \tilde Z_{n,\cE_{n,\phi}} \geq 2\EE\left[Z_{n,\cF_{n,\phi}}\right]+2\EE\left[\tilde Z_{n,\cE_{n,\phi}}\right]+nz\right) \notag \\
& \leq \PP\left(Z_{n,\cF_{n,\phi}} \geq 2\EE\left[Z_{n,\cF_{n,\phi}}\right]+\frac{nz}{2}\right) + \PP\left(\tilde Z_{n,\cE_{n,\phi}} \geq 2\EE\left[\tilde Z_{n,\cE_{n,\phi}}\right]+\frac{nz}{2}\right). \label{eq:devsplitineq}   
\end{align}
Applying   \eqref{eq:tailbndepmomcond} in Theorem \ref{Thm:Admczak-tailbnd} to each term above and scaling $z$ by appropriate constants (depending on $p$),  we obtain \eqref{eq:KLdevineq-mombnd}. 

The proof of \eqref{eq:KLdevineq-compsupp} is analogous up to \eqref{eq:devsplitineq}, and the conclusion then follows by applying \eqref{eq:talagrandconctrunc} instead, noting that $\sup_{f \in \cF_{n,\phi}} \EE_{\mu }\big[\big(f-\EE_{\mu}[f]\big)^2\big]\leq   \norm{F_n*\phi}_{2,\mu}^2 $  and scaling $z$ appropriately. 
\subsection{Proof of Lemma \ref{lem:bndholdernorm}}\label{Sec:lem:bndholdernorm-proof}
We may assume without loss of generality that $b>0$, since otherwise, the claims hold trivially. Consider  $q =0$ or $q \in [1,\infty)$. For $f \in \cF(b,q)$, set  $f_{\sigma}:=f*\varphi_{\sigma}$. Let $H_m(z)$, $m \in \ZZ_{\geq 0}$, denote the Hermite polynomial of degree $m$ defined by
\begin{align}
    H_m(z)=(-1)^m e^{\frac{z^2}{2}} \left[\frac{d^m}{dz^m}e^{-\frac{z^2}{2}}\right]. \notag
\end{align}
We have
\begin{align}
  D^kf_{\sigma}(x)=(-\sigma)^{-|k|}\int f(y) \varphi_{\sigma}(x-y)\underbrace{\left[\prod_{j=1}^d H_{k_j}\left(\frac{x_j-y_j}{\sigma}\right)\right]}_{B_{k,\sigma}(x-y)}dy.   \label{eq:bksigma}
\end{align}
Hence
\begin{align}
&\abs{D^kf_{\sigma}(x)-D^kf_{\sigma}(z)} \leq \sigma^{-|k|}\int \abs{f(y)} \abs{\varphi_{\sigma}(x-y)B_{k,\sigma}(x-y)-\varphi_{\sigma}(z-y)B_{k,\sigma}(z-y)}dy. \label{eq:bndderexp}
\end{align}
By triangle inequality for absolute value, we have
\begin{align}
    &\abs{\varphi_{\sigma}(x-y)B_{k,\sigma}(x-y)-\varphi_{\sigma}(z-y)B_{k,\sigma}(z-y)} \notag \\
    &\leq \varphi_{\sigma}(x-y) \abs{B_{k,\sigma}(x-y)-B_{k,\sigma}(z-y)}  + \abs{B_{k,\sigma}(z-y)}\abs{\varphi_{\sigma}(x-y)-\varphi_{\sigma}(z-y)} \notag \\
    &\leq  \abs{B_{k,\sigma}(x-y)-B_{k,\sigma}(z-y)} \big(\varphi_{\sigma}(x-y)+\abs{\varphi_{\sigma}(x-y)-\varphi_{\sigma}(z-y)}\big) \notag \\
    & \qquad \qquad + \abs{B_{k,\sigma}(x-y)}\abs{\varphi_{\sigma}(x-y)-\varphi_{\sigma}(z-y)}. \label{eq:intermdiff1}
\end{align}
In the second inequality, we used $\abs{B_{k,\sigma}(z-y)} \leq \abs{B_{k,\sigma}(x-y)}+\abs{B_{k,\sigma}(x-y)-B_{k,\sigma}(z-y)}$. 

Next observe that
\begin{align}
  \varphi_{\sigma}(x-y)-\varphi_{\sigma}(z-y)&=\varphi_{\sigma}(x-y)\left( 1-e^{-\frac{\big(\norm{z-y}^2-\norm{x-y}^2\big)}{2 \sigma^2} } \right) \notag \\
  & \stackrel{(a)}{\leq} \varphi_{\sigma}(x-y)\left( 1-e^{-\frac{\big(\norm{z-y}+\norm{x-y}\big)\norm{x-z}}{2 \sigma^2} } \right) \notag \\
  & \stackrel{(b)}{\leq}  \frac{\norm{x-z}}{ 2\sigma^2} \varphi_{\sigma}(x-y)\big(\norm{z-y}+\norm{x-y}\big) \notag \\
  & \stackrel{(c)}{\leq} \frac{\norm{x-z}^2}{2\sigma^2} \varphi_{\sigma}(x-y)+\frac{\norm{x-z}\norm{x-y}}{ \sigma^2}\varphi_{\sigma}(x-y), \label{eq:intermdiff2}
\end{align}
where
\begin{enumerate}[(a)]
\item follows by using $a^2-b^2=(a+b)(a-b)$ for $a,b \in \RR$ and  $\norm{u}-\norm{v} \leq \norm{u-v}$ for $u,v \in \RR^d$;
\item is due to $1-e^{-x} \leq x$ for $x \in \RR$;
\item is since $\norm{z-y}+\norm{x-y} \leq \norm{x-z}+2\norm{x-y}$ for $x,y,z \in \RR^d$.
\end{enumerate}
Substituting \eqref{eq:intermdiff2} in \eqref{eq:intermdiff1}, we obtain
\begin{align}
  & \abs{\varphi_{\sigma}(x-y)B_{k,\sigma}(x-y)-\varphi_{\sigma}(z-y)B_{k,\sigma}(z-y)}\notag \\
 &\qquad \qquad \qquad \qquad\leq   \varphi_{\sigma}(x-y)\abs{B_{k,\sigma}(x-y)-B_{k,\sigma}(z-y)} \left(1+\frac{\norm{x-z}^2}{2 \sigma^2} +\frac{\norm{x-z}\norm{x-y}}{ \sigma^2}\right) \notag \\
 &\qquad \qquad \qquad\qquad \qquad \qquad \qquad \qquad + \varphi_{\sigma}(x-y)\abs{B_{k,\sigma}(x-y)}\norm{x-z}\left( \frac{\norm{x-z}}{2 \sigma^2} +\frac{\norm{x-y}}{ \sigma^2}\right). \notag 
\end{align}
Consider the term $B_{k,\sigma}(x-y)$. We have
\begin{align}
B_{k,\sigma}(x-y)&:= \prod_{j=1}^d H_{k_j}\left(\frac{x_j-y_j}{\sigma}\right)=\sum_{l=1}^{|k|}\sigma^{-l} \sum_{\substack{s \in \ZZ_{\geq 0}^d:\\|v|=l}} \tilde c_{k,s} \prod_{j=1}^d(x_j-y_j)^{v_j}, \notag
\end{align}
for some constants $\tilde c_{k,s}$. A crude upper bound on these constants can be obtained by using the closed form expression for the coefficients of a Hermite polynomial as
\begin{align}
\max_{\abs{s} \leq \abs{k}} \abs{\tilde c_{k,s}} \leq \max_{\abs{s} \leq \abs{k}} \frac{|k|!}{|s|!(|k|-2|s|)! 2^{|s|}} \leq |k|!\left(\left\lfloor\frac{\abs{k}}{3} \right\rfloor !\right)^{-2} 2^{-\left\lfloor\frac{\abs{k}}{3} \right\rfloor} \leq 3 \abs{k}^{\frac 32} (2\abs{k})^{\frac{\abs{k}}{3}}.\notag
\end{align}
Hence,
\begin{align}
   \abs{B_{k,\sigma}(x-y)-B_{k,\sigma}(z-y)} \leq \sum_{l=1}^{|k|}\sigma^{-l} \sum_{\substack{s \in \ZZ_{\geq 0}^d:\\|s|=l}} \abs{\tilde c_{k,s}} \abs{\prod_{j=1}^d(x_j-y_j)^{s_j}-\prod_{j=1}^d(z_j-y_j)^{s_j}}. \label{eq:difftermb} 
\end{align}
Now, note that for any $u,v \in \RR^i$, we have
\begin{flalign}
 \prod_{j=1}^{i} u_j-\prod_{j=1}^i v_j&=\sum_{l=0}^{i-1}\left(\prod_{j=1}^{i-l} u_j\prod_{j=i-l+1}^{i} v_j-\prod_{j=1}^{i-l-1} u_j\prod_{j=i-l}^{i} v_j \right)  \notag \\
 &=\sum_{l=0}^{i-1}(u_{i-l}-v_{i-l}) \left(\prod_{j=1}^{i-l-1} u_j\right) \left(\prod_{j=i-l+1}^{i} v_j\right). \notag 
\end{flalign}
Hence, we have
\begin{align}
\abs{\prod_{j=1}^{i} u_j-\prod_{j=1}^i v_j} \leq \sum_{l=0}^{i-1} \norm{u-v}_{\infty} \left(\norm{u}_{\infty}+\norm{v}_{\infty}\right)^{i} \leq i 2^{i-1}\norm{u-v}_{\infty} \big(\norm{u}_{\infty}^i+\norm{v}_{\infty}^i\big).  \notag 
\end{align}
Substituting in \eqref{eq:difftermb}, we obtain
\begin{align}
   \abs{B_{k,\sigma}(x-y)-B_{k,\sigma}(z-y)} &\leq \sum_{l=1}^{|k|}\sigma^{-l} \sum_{\substack{s \in \ZZ_{\geq 0}^d:|s|=l}} \abs{\tilde c_{k,s}} |s|2^{|s|-1}\norm{x-z} \left(\norm{x-y}^{|s|}+\norm{z-y}^{|s|}\right) \notag \\
   & \leq \hat c_{d,\abs{k}}  \sigma^{-|k|} \norm{x-z}  \left(1 + \norm{x-y}^{|k|}+\norm{z-y}^{|k|}\right) \notag \\
   &   \leq \hat c_{d,\abs{k}} \sigma^{-|k|} \norm{x-z}  \left(1 + \norm{x-y}^{|k|}+\norm{x-z}^{|k|}\right), \notag 
\end{align}
where 
\begin{align}
    \hat c_{d,\abs{k}}:=|k|!(\abs{k}+1)^d \left(\left\lfloor\frac{\abs{k}}{3} \right\rfloor !\right)^{-2} 2^{-\left\lfloor\frac{\abs{k}}{3} \right\rfloor}. \notag
\end{align}
Hence, since $\sigma^{-|k|} \geq \sigma^{{-|k'|}}$ for all $\sigma \leq 1$ and $|k| \geq |k'| $, we have for such $\sigma$ that 
\begin{align}
  & \abs{\varphi_{\sigma}(x-y)B_{k,\sigma}(x-y)-\varphi_{\sigma}(z-y)B_{k,\sigma}(z-y)}\notag \\
 & \leq \hat c_{d,\abs{k}}  \sigma^{-|k|} \varphi_{\sigma}(x-y)\norm{x-z}  \left(1 + \norm{x-y}^{|k|}+\norm{x-z}^{|k|}\right) \left(1+\frac{\norm{x-z}^2}{2 \sigma^2} +\frac{\norm{x-z}\norm{x-y}}{ \sigma^2}\right) \notag \\
 &\qquad \qquad\qquad\qquad\qquad\qquad+ \varphi_{\sigma}(x-y)\abs{B_{k,\sigma}(x-y)}\norm{x-z}\left( \frac{\norm{x-z}}{2 \sigma^2} +\frac{\norm{x-y}}{ \sigma^2}\right). \notag
\end{align}
Next observe that for $q \geq 1$ or $q=0$ and $f \in \cF(b,q)$, 
\begin{align}
\abs{f(y)} \leq 0.5 b (1+\norm{y}^q) \leq 0.5 b (1+2^{q-1}\norm{x-y}^q+2^{q-1}\norm{x}^q ). \notag     
\end{align}
Hence, for $x \in \cX$, we obtain using the definition of $v_{b,q}(\cX)$ (see \eqref{eq:envelopmax}) that
\begin{align}
\abs{f(y)} \leq  2^{q-1}  \left(v_{b,q}(\cX)+b\norm{x-y}^q\right). \notag     
\end{align}
Using this in the above equation yields for $x \in \cX$ that
\begin{align}
  & \abs{f(y)}\abs{\varphi_{\sigma}(x-y)B_{k,\sigma}(x-y)-\varphi_{\sigma}(z-y)B_{k,\sigma}(z-y)}\notag \\
 & \leq 2^{q-1} \hat c_{d,\abs{k}}\big(v_{b,q}(\cX)+b\norm{x-y}^q\big)\varphi_{\sigma}(x-y)\norm{x-z} \Bigg(  \sigma^{-|k|}  \left(1 + \norm{x-y}^{|k|}+\norm{x-z}^{|k|}\right) \notag \\
 &\qquad \qquad \qquad \left(1+\frac{\norm{x-z}^2}{2 \sigma^2} +\frac{\norm{x-z}\norm{x-y}}{ \sigma^2}\right)  + \abs{B_{k,\sigma}(x-y)}\left( \frac{\norm{x-z}}{2 \sigma^2} +\frac{\norm{x-y}}{ \sigma^2}\right)\Bigg) \notag \\
 & \leq 2^{q-1}\hat c_{d,\abs{k}}\big(v_{b,q}(\cX)+b\norm{x-y}^q\big)\varphi_{\sigma}(x-y)\norm{x-z}  \bigg( \sigma^{-|k|}\Big(1 + \norm{x-y}^{|k|}+\norm{x-z}^{|k|}\Big) \notag \\
 & \qquad \qquad\qquad \qquad +\abs{B_{k,\sigma}(x-y)}\frac{\norm{x-y}}{ \sigma^2}\bigg) + \norm{x-z}^2 R_{d,q,\sigma,|k|},\label{eq:integrndub}
\end{align}
where 
\begin{align}
R_{d,q,\sigma,|k|}&:=    2^{q-1}\hat c_{d,\abs{k}}\big(v_{b,q}(\cX)+b\norm{x-y}^q\big)\varphi_{\sigma}(x-y)\bigg(\abs{B_{k,\sigma}(x-y)}+ \sigma^{-(|k|+2)}    \notag \\
&\qquad \qquad\qquad\qquad\qquad\quad \times \Big(1 + \norm{x-y}^{|k|}+\norm{x-z}^{|k|}\Big)\left(\norm{x-z}+\norm{x-y}\right)\bigg).\notag
\end{align}
For  $U \sim N(0,\sigma^2)$ and $V=N(0,\sigma^2I_d)$, a straightforward computation yields
\begin{subequations}
    \begin{align}
    &\big(\EE[H_m(U)])^2 \leq \EE[H_m^2(U)]=m !, \label{eq:bndherm1}\\
&\EE[\norm{V}^q]^{\frac 1q} \leq \sqrt{qd}\sigma.\label{eq:bndmomnt1}
    \end{align}   
\end{subequations}
    From these and $\prod_{j=1}^d k_j! \leq |k|!$, we obtain using the definition of $B_{k,\sigma}(x-y)$ (see \eqref{eq:bksigma}) that
\begin{subequations}\label{eq:integralvals}
    \begin{align}
  &  \int_{\RR^d} \norm{x-y}^a   \varphi_{\sigma}(x-y) dy=\left(\sqrt{ad}\sigma\right)^a ,  \\
       &  \int_{\RR^d}     \abs{B_{k,\sigma}(x-y)}\varphi_{\sigma}(x-y) dy\leq \left(\prod_{j=1}^d\EE\Big[H_{k_j}^2(U)\Big]\right)^{\frac 12} \leq \sqrt{|k|!},  
    \end{align}
    and 
    \begin{align}
     \int_{\RR^d} \norm{x-y}^a    \abs{B_{k,\sigma}(x-y)}\varphi_{\sigma}(x-y) dy\leq \left(\EE\Big[\norm{V}^{2a}\Big]\right)^{\frac 12}  \left(\prod_{j=1}^d\EE\Big[H_{k_j}^2(U)\Big]\right)^{\frac 12} \leq \big(\sqrt{2ad}\sigma\big)^{a}\sqrt{|k|!}, 
    \end{align}
        \end{subequations}
where the penultimate inequality follows from  Cauchy-Schwarz inequality. 
Substituting \eqref{eq:integrndub} in \eqref{eq:bndderexp}, integrating w.r.t. $y$ and using \eqref{eq:integralvals}, we obtain for $x,z \in \cX$ and $\sigma \leq 1$ that 
\begin{align}
\abs{D^kf_{\sigma}(x)-D^kf_{\sigma}(z)} &\leq \bar c_{d,|k|,q} v_{b,q}(\cX)  \sigma^{-(|k| \vee 1)}\norm{x-z} +O_{d,\sigma,q,|k|}\left(\norm{x-z}^2\right), \notag 
\end{align}
where 
\begin{align}
 \bar c_{d,|k|,q}&:=   2^{q-1}  \hat c_{d,\abs{k}}\left(\sqrt{(q+1)|k|!} \vee \sqrt{(q+|k|)^{q+|k|}}\right)\left( \sqrt{d} \vee \sqrt{d^{q+|k|}}\right). \notag
\end{align}
From this, we obtain for $|k| \geq 1$, $x \in \cX$ and $\sigma \leq 1$ that 
\begin{align}
   \abs{D^kf_{\sigma}(x)} &\leq \bar c_{d,|k|-1,q} v_{b,q}(\cX)  \sigma^{-\big((|k|-1) \vee 1\big)}.\notag
\end{align}
On the other hand, for $x \in \cX$, we have for $q \geq 1$ and $\sigma \leq 1$ that
\begin{align}
  \abs{f_{\sigma}(x)}  \leq \int_{\RR^d} \abs{f(y)} \varphi_{\sigma}(x-y) dy 
  &\leq \int_{\RR^d} 0.5b(1+\norm{y}^q) \varphi_{\sigma}(x-y) dy \notag \\
   &\leq \int_{\RR^d} 0.5b(1+2^{q-1}\norm{x}^q+2^{q-1}\norm{x-y}^q) \varphi_{\sigma}(x-y) dy \notag \\
   & \leq 2^{q-1}v_{b,q}(\cX)\left(1+\big(\sqrt{qd}\big)^q\right), \label{eq:funcubndcls} 
\end{align}
where to obtain \eqref{eq:funcubndcls}, we used the definition of $v_{b,q}(\cX)$ and \eqref{eq:bndmomnt1}. 
For $f \in \cF(b)$ (corresponds to $q=0$) and any $x \in \RR^d$,
\begin{align}
  \abs{f_{\sigma}(x)}  &\leq \int_{\RR^d} \abs{f(y)} \varphi_{\sigma}(x-y) dy  \leq b \int_{\RR^d}  \varphi_{\sigma}(x-y) dy \leq b. \notag 
\end{align}
Note that the RHS of \eqref{eq:funcubndcls} simplifies to $b$ when $q=0$, and hence we may use this expression as a common upper bound for $q \in \{0\} \cup[1,\infty)$. Combining all the above, we obtain that for $|k| \in \ZZ_{\geq 0}$ and $\sigma \leq 1$,
\begin{align}
  \sup_{x \in \cX} \abs{D^kf_{\sigma}(x)} &\leq \bar c_{d,|k|,q} v_{b,q}(\cX)  \sigma^{-\big((|k|-1) \vee 1\big)}.\label{eq:bndlips1}
\end{align}
Finally, we have
\begin{align}
 & \max_{|k|=\ubar \beta} \sup_{x,z} \frac{\abs{D^{k}f_{\sigma}(x)-D^{k}f_{\sigma}(z)}}{\norm{x-z}^{\beta-\ubar{\beta}}} \notag \\
 &\leq   \max_{|k|=\ubar \beta} \sup_{x,z} \left(\frac{\abs{D^{k}f_{\sigma}(x)-D^{k}f_{\sigma}(z)}}{\norm{x-z}}\right)^{\beta-\ubar{\beta}} \max_{|k|=\ubar \beta} \sup_{x,z} \abs{D^{k}f_{\sigma}(x)-D^{k}f_{\sigma}(z)}^{1-\beta+\ubar{\beta}} \notag \\
  &\leq  2 \max_{|k|=\ubar \beta} \sup_{x,z} \left(\frac{\abs{D^{k}f_{\sigma}(x)-D^{k}f_{\sigma}(z)}}{\norm{x-z}}\right)^{\beta-\ubar{\beta}} \max_{|k|=\ubar \beta} \sup_{x} \abs{D^{k}f_{\sigma}(x)}^{1-\beta+\ubar{\beta}} \notag \\
  & \stackrel{(a)}{\leq}\left(\sqrt{d} \bar c_{d,\ubar{\beta},q} v_{b,q}(\cX)  \sigma^{-(\ubar{\beta} \vee 1)}\right)^{\beta-\ubar{\beta}}\left(\bar c_{d,\ubar{\beta},q} v_{b,q}(\cX)  \sigma^{-\big((\ubar{\beta}-1) \vee 1\big)}\right)^{1-\beta+\ubar{\beta}} \notag \\
  & \leq \check c_{d,\beta,q}v_{b,q}(\cX)\sigma^{-(\ubar{\beta} \vee 1)}, \label{eq:bndlips2}
\end{align}
where the suprema above are over $x,z$ in the interior of $\cX$, and 
\begin{align}
  \check c_{d,\beta,q}&:= 2^{q-1}  \hat c_{d,\ubar{\beta}}\left(\sqrt{(q+1)\ubar{\beta}!} \vee \sqrt{(q+\ubar{\beta})^{q+\ubar{\beta}}}\right) \left(d \vee \sqrt{d^{\ubar{\beta}+q+1}}\right) \notag \\
  & \leq  2^{q-1} \ubar{\beta}!(\ubar{\beta}+1)^d\left(\left\lfloor\frac{\ubar{\beta}}{3} \right\rfloor !\right)^{-2} 2^{-\left\lfloor\frac{\ubar{\beta}}{3} \right\rfloor}\left(\sqrt{(q+1)\ubar{\beta}!} \vee \sqrt{(q+\ubar{\beta})^{q+\ubar{\beta}}}\right) \left(d \vee \sqrt{d^{\ubar{\beta}+q+1}}\right). \label{eq:constholdnorm}
\end{align}
In $(a)$, we used \eqref{eq:bndlips1} and 
\begin{align}
    \sup_{x,z} \frac{\abs{D^{k}f_{\sigma}(x)-D^{k}f_{\sigma}(z)}}{\norm{x-z}} \leq \sup_{x}\norm{\nabla\big( D^{k}f_{\sigma}(x)\big)} \leq  \sqrt{d} \bar c_{d,|k|,q} v_{b,q}(\cX)  \sigma^{-(|k| \vee 1)}. \notag
\end{align}
Here, the first inequality  follows by Lagrange's mean value theorem for multiple variables which applies due to differentiability of $D^kf_{\sigma}$ everywhere, and the second inequality uses  \eqref{eq:bndlips1} along with $\norm{x} \leq \sqrt{d} \norm{x}_{\infty}$ for $x \in \RR^d$.

Combining  \eqref{eq:bndlips1}, \eqref{eq:bndlips2} and the definition of H\"{o}lder norm given in \eqref{eq:defHoldnorm}, we obtain \eqref{eq:Holdconstbndfclass}. 
Since $\norm{x}^q \leq 1+\norm{x}$ for all $q \in [0,1]$, the upper bound for $q=1$ (up to a universal multiplicative  constant) also applies  to $q \in [0,1)$. 

\medskip

Equipped with \eqref{eq:Holdconstbndfclass}, to prove \eqref{eq:coventunif}-\eqref{eq:coventholderprob},  we will use the following result from \cite{AVDV-book}  that provides covering entropy bounds for H\"{o}lder ball of functions with bounded support. 
    \begin{theorem}[Theorem 2.7.1 and Corollary 2.7.2 in \cite{AVDV-book}] \label{Thm:coventconst}
        Let $\cX$ be a bounded Borel-measurable convex subset of $\RR^d$ with non-empty interior. For $\beta,L>0$, let $C_L^{\beta}(\cX)$ denote the functions $f:\cX \rightarrow \RR$ such that $\norm{f}_{\beta,\cX} \leq L$.  Then 
        \begin{align}
            \log N(L\epsilon,C_L^{\beta}(\cX),\norm{\cdot}_{\infty,\cX}) &\leq c_{d,\beta,\cX}  \mspace{2 mu} \epsilon^{-\frac{d}{\beta}}, \label{eq:coventholder} \\
                 \log N\big(L\epsilon \gamma(\cX)^{1/r},C_L^{\beta}(\cX),\norm{\cdot}_{r,\gamma}\big) &\leq c_{d,\beta,\cX}  \mspace{2 mu}\epsilon^{-\frac{d}{\beta}}, ~\forall~r \geq 1, \gamma \in \cB(\cX), \label{eq:coventholderarbmeas} 
        \end{align}
        where  $\cB(\cX)$ denotes the set of positive Borel measures on $\cX$. For $\cX=\BB_d(r)$, the above inequalities hold with $c_{d,\beta,\BB_d(r)}$ as given in \eqref{eq:constcoventbnd} below. 
    \end{theorem}
   The proof of Theorem \ref{Thm:coventconst}  given in \cite{AVDV-book} establishes \eqref{eq:coventholder}, albeit without specifying the dependence of $d,\beta$ in the pre factor  $ c_{d,\beta,\cX}$.  In Appendix  \ref{App:Thm:coventconst-proof}, we provide an upper bound on this constant when  $\cX=\BB_d(r)$, which will be used below to complete the proof of Lemma \ref{lem:bndholdernorm}. 

    \medskip

 Let $\cF_{\varphi_{\sigma}}(b,q,\cX)$ denote the class of functions obtained by restricting the domain of functions in $\cF_{\varphi_{\sigma}}(b,q)$ to $\X$. It follows from \eqref{eq:Holdconstbndfclass} that $\cF_{\varphi_{\sigma}}\big(b,q,\BB_d(r)\big) \subseteq C_L^{\beta}\big(\BB_d(r)\big)$  for any $\beta>0$ with $L:=\check c_{d, \beta,q} b(1+r^q)\sigma^{-(\ubar{\beta} \vee 1)}$, where   
$\check c_{d, \beta,q}$ is given in \eqref{eq:constholdnorm}.   Then,  applying \eqref{eq:coventholder}  with  $L$ as above and noting that $\norm{F_{b,q}*\varphi_{\sigma}}_{\infty,\BB_d(r)} \geq 0.5 b$ yields 
    \begin{align}
      \log N\big(\norm{F_{b,q}*\varphi_{\sigma}}_{\infty,\BB_d(r)}\epsilon,\cF_{\varphi_{\sigma}}\big(b,q\big),\norm{\cdot}_{\infty,\BB_d(r)}\big) & \leq c_{d,\beta}(1+r)^{d}L^{\frac{d}{\beta}}\left(\norm{F_{b,q}*\varphi_{\sigma}}_{\infty,\BB_d(r)}\epsilon\right)^{-\frac{d}{\beta}} \notag \\ &\leq c_{d,\beta,q}(1+r)^{d}(1+r^q)^{\frac{d}{\beta}} \mspace{2 mu} \epsilon^{-\frac{d}{\beta}}\sigma^{-\frac{(\ubar{\beta} \vee 1)d}{\beta}},\notag     
    \end{align}
      where $c_{d,\beta}$ is as given in \eqref{eq:constcoventbnd}, and 
      \begin{align}
    c_{d,\beta,q}&:=c_{d,\beta}  \big(2 \check c_{d, \beta,q}\big)^{\frac{d}{\beta}} \notag \\
    &= (\ubar{\beta}+1)^d 2^{\frac d \beta} d^{\frac d 2} (e^d+1)^{\frac{d}{ \beta}}\left(\frac{\beta 2^{\frac{d}{\beta}}}{d} \vee 3^d \log \Big(2^{\beta+1} d^{\frac {\beta}{2}} (e^d+1)+1\Big)\right)\notag \\
    & \quad \left(2^{q} \ubar{\beta}!(\ubar{\beta}+1)^d \left(\left\lfloor\frac{\ubar{\beta}}{3} \right\rfloor !\right)^{-2} 2^{-\left\lfloor\frac{\ubar{\beta}}{3} \right\rfloor}\left(\sqrt{(q+1)\ubar{\beta}!} \vee \sqrt{(q+\ubar{\beta})^{q+\ubar{\beta}}}\right) \left(d \vee \sqrt{d^{\ubar{\beta}+q+1}}\right)\right)^{\frac{d}{\beta}}.\label{eq:coventrconst} 
      \end{align}
      For $\beta=\frac{d}{2}+1$ and $q \in [0,1]$, simplifying  the above expression  yields that  $c_{d,\frac{d}{2}+1,q}=O\big(d^{6d+9}(9e^2)^d\big)$, where we used $\log x \leq x-1$ for $x >0$ and Stirling's approximation  for factorial given in \cite{robbins1955}:
      \begin{align}
       \sqrt{2\pi n}\left(\frac{n}{e}\right)^n e^{\frac{1}{12n+1}}  \leq  n! \leq \sqrt{2\pi n}\left(\frac{n}{e}\right)^n e^{\frac{1}{12n}}. \notag
      \end{align}
      This proves \eqref{eq:coventunif}. 
    Similarly,  \eqref{eq:coventholderarbmeas} yields \eqref{eq:coventholderprob} by using that $\norm{F_{b,q}*\varphi_{\sigma}}_{p,\gamma} \geq 0.5 b$  for $p \geq 1$.   
    Finally, the proof of the penultimate claim for $\cE_{\varphi_{\sigma},\alpha}(b,q)$ follows straightforwardly via the same proof by replacing $v_{b,q}(\cX)$ by $e^{\alpha v_{b,q}(\cX)}$. 
    \subsection{Proof of Lemma \ref{Lem:statistver:Fuk-Nagaev}} \label{Sec:Lem:statistver:Fuk-Nagaev-proof} 
    We will show \eqref{eq:tailbndstatver} with 
\begin{align}
    c_q= 4 \tilde c_q \left(2^{\frac 1q-1}\left(1+\left\lceil\frac{1}{q} \right\rceil ! \mspace{2 mu}(\log 2)^{-\frac 1 q} \right)+16 \left\lceil\frac{1}{q} \right\rceil !\right), \label{eq:Orlicznormconstfin}
\end{align}
where $\tilde c_q$ is the maximum of the constant appearing in \cite[Theorem 6.21]{LT-1991} and $1$. 
We will use the following statistical version of Talagrand's inequality given in \cite{Bartlett-2005}.
\begin{theorem}[Theorem 2.1 in \cite{Bartlett-2005}]
Consider the setting of Theorem \ref{Thm:Talagrand-Bosquet}. Let $\{\varepsilon_i\}_{i=1}^n$ be i.i.d. Rademacher variables independent of $X^n$, and set $Z^{(\varepsilon)}_n:=\EE_{\varepsilon}\left[\sup_{f \in \cF}\abs{\sum_{i=1}^n \varepsilon_i f(X_i)} \right]$. Then, for all $z \geq 0$,
\begin{align}
\PP\left(Z_n \geq \inf_{\eta \in (1,\infty)} \left\{2 \eta Z^{(\varepsilon)}_n+\theta \sqrt{2nz}+\frac{(3\eta^2+11 \eta+4)bz}{3(\eta-1)}\right\} \right) \leq 2e^{-z}. \label{eq:statist-fuk-nagaev}
\end{align}
In particular, for every $z \geq 0$ and $\eta \in (1,\infty)$,
\begin{align}
\PP\left(Z_n \geq 2 \eta Z^{(\varepsilon)}_n+z \right) \leq 2e^{-\left(\frac{z^2}{8n \hat \theta_n^2} \wedge \frac{3z(\eta-1)}{2(11 \eta+3\eta^2+4)b}\right)}. \label{eq:eq:statist-fuk-nagaev-simp}
\end{align}
\end{theorem}
Equation \eqref{eq:statist-fuk-nagaev} follows from \cite[Theorem 2.1]{Bartlett-2005} by setting $\eta=\frac{1+\alpha}{1-\alpha}$.

\medskip

Let  $T_n$ be a truncation level (to be specified later) for the truncated function class  defined in Lemma \ref{Lem:statistver:Fuk-Nagaev}. 
We have
\begin{align}
     Z_n \leq  Z_n'+ Z_n'' , \label{eq:triangineq1}
\end{align}
where 
\begin{align}
Z_n&:=\sup_{f \in \cF_n}\abs{\sum_{i=1}^n \big(f(X_i)-\EE_{\mu}[f]\big)}, \notag \\
    Z_n'&:= \sup_{f \in \cF_n}\abs{\sum_{i=1}^n \big(f\ind_{\{F_n \leq T_n\}}(X_i)-\EE_{\mu} \left[f\ind_{\{F_n \leq T_n\}}\right]\big)} = \sup_{f \in \cF_n(T_n)}\abs{\sum_{i=1}^n \big(f(X_i)-\EE_{\mu}[f]\big)}, \notag \\
    Z_n''&:=\sup_{f \in \cF_n}\abs{\sum_{i=1}^n \Big(f\ind_{\{F_n>T_n\}}(X_i)-\EE_{\mu}\left[f\ind_{\{F_n >T_n\}}\right]\Big)}. \notag 
\end{align}
Set
\begin{align} Z_n^{'\varepsilon}:=\EE_{\varepsilon}\left[\sup_{f \in \cF_n(T_n)}\abs{\sum_{i=1}^n \varepsilon_i f(X_i)}\right]. \notag 
\end{align}
Let $\eta \in (1,\infty)$.  From the above and \eqref{eq:triangineq1}, it follows that 
\begin{align}
    \PP \left(Z_n \geq 2\eta Z_n^{'\varepsilon}+z \right) &\leq   \PP \left(Z_n'+Z_n'' \geq  2\eta Z_n^{'\varepsilon}+z \right)  \leq \PP \left(Z_n' \geq  2\eta Z_n^{'\varepsilon}+\frac{3z}{4} \right)+\PP \left(Z_n'' \geq  \frac{z}{4} \right). \label{eq:sumbndsplit}
\end{align}
By symmetrization inequality (see e.g. \cite[Lemma 2.3.1]{AVDV-book}), we have
\begin{align}
    \EE \big[Z_n''\big] \leq 2 \EE \left[\sup_{f \in \cF_n}\abs{\sum_{i=1}^n \varepsilon_i f\ind_{\{F_n>T_n\}}(X_i)}\right]. \label{eq:remaindertermenv}
\end{align}
 Then, applying the Hoffmann-J{\o}rgensen inequality, we obtain
\begin{align}
   \EE \left[\sup_{f \in \cF_n}\abs{\sum_{i=1}^n \varepsilon_i f\ind_{\{F_n>T_n\}}(X_i)}\right]  &\leq 8 \EE \left[ \max_{1 \leq i \leq n}\sup_{f \in \cF_n}\abs{  f\ind_{\{F_n>T_n\}}(X_i)}\right] +8 t_0(T_n) \leq 8 \EE \left[M_n\right] +8 t_0(T_n), \notag
\end{align}
where $M_n:= \max_{1 \leq i \leq n}F_n(X_i)$ and
\begin{align}
    t_0(T_n):=\inf \left\{t>0: \PP\left(\max_{1 \leq k \leq n}\sup_{f \in \cF_n}\abs{\sum_{i=1}^k \varepsilon_i f\ind_{\{F_n>T_n\}}(X_i)}>t\right) \leq 1/8\right\}. \notag
\end{align}
For $T_n=T_n^{\star}:=8\EE[M_n]$,  Markov's inequality yields
\begin{align}
    \PP\left(\max_{1 \leq k \leq n}\sup_{f \in \cF_n}\abs{\sum_{i=1}^k \varepsilon_i f\ind_{\{F_n>T_n\}}(X_i)}>0\right) \leq \PP\left(M_n > T_n\right) \leq \frac{\EE[M_n]}{T_n} =\frac{1}{8}. \notag
\end{align}
Hence, $t_0\big(T_n^{\star}\big)=0$. Combining the above and substituting in \eqref{eq:remaindertermenv}, we obtain
\begin{align}
\norm{Z_n''}_{1}= \EE \big[Z_n''\big] \leq 16 \EE \left[M_n\right] \stackrel{(a)}{\leq} 16 \left\lceil\frac{1}{q} \right \rceil ! \norm{M_n}_{\psi_q}, \label{eq:expectbndrem}
\end{align}
where $(a)$ used \eqref{eq:expbndorlicznorm} given below.

Next, applying \eqref{eq:eq:statist-fuk-nagaev-simp} to the bounded function class $\cF_n(T_n^{\star})$ with $b=8\EE[M_n]$, we have for every $\eta \geq 1$ that
\begin{align}
    \PP \left(Z_n' \geq  2\eta Z_n^{'\varepsilon}+\frac{3z}{4} \right)  \leq e^{-\big(\frac{9z^2}{128 n \norm{F_n}_{2,\mu}^2} \wedge \frac{9z(\eta-1)}{64(11 \eta+3\eta^2+4)\EE[M_n]}\big)}, \notag
\end{align}
where we used 
\begin{align}
    \hat \theta_n^2 :=\sup_{f \in \cF_n(T_n^{\star})} \EE_{\mu }\big[\big(f-\EE_{\mu}[f]\big)^2\big]\leq   \sup_{f \in \cF_n(T_n^{\star})} \EE_{\mu }\big[f^2\big] \leq  \sup_{f \in \cF_n} \EE_{\mu }\big[f^2\big] \leq  \norm{F_n}_{2,\mu}^2. \notag
\end{align}
Choosing $\eta=1+\sqrt{6}$ to maximize the second exponent, we obtain
\begin{align}
    \PP \left(Z_n' \geq  7 Z_n^{'\varepsilon}+\frac{3z}{4} \right)  \leq e^{-\big(\frac{9z^2}{128 n \norm{F_n}_{2,\mu}^2} \wedge \frac{z}{228\mspace{2 mu}\EE[M_n]}\big)}, \label{eq:talagrandineqappl}
\end{align}
For controlling the second term in \eqref{eq:sumbndsplit}, note that for  $q \in (0,1]$,  
\begin{align}
\PP \left(Z_n'' \geq  \frac{z}{4} \right) \leq 2e^{\frac{-z^q}{4^q\norm{Z_n''}_{\psi_q}^q}}. \label{eq:exporlicznorm}
\end{align}
From \cite[Theorem 6.21]{LT-1991}, we obtain
\begin{align}
\norm{Z_n''}_{\psi_q} \leq \tilde c_q \left(\norm{Z_n''}_{1}+\bigg\|\max_{1 \leq i \leq n} \sup_{f \in \cF_n}\abs{f\ind_{\{F_n>T_n\}}(X_i)-\EE_{\mu}\left[f\ind_{\{F_n >T_n\}}\right]}\bigg\|_{\psi_q} \right).\label{Orlicznormtailbnd} 
\end{align}
We can upper bound the last term in the above equation as follows:
\begin{align}
   &\bigg\|\max_{1 \leq i \leq n} \sup_{f \in \cF_n}\abs{f\ind_{\{F_n>T_n\}}(X_i)-\EE_{\mu}\big[f\ind_{\{F_n >T_n\}}\big]}\bigg\|_{\psi_q} \notag \\
   &\stackrel{(a)}{\leq} 2^{\frac{1}{q}-1} \bigg\|\max_{1 \leq i \leq n} \sup_{f \in \cF_n}\abs{f\ind_{\{F_n>T_n\}}(X_i)}\bigg\|_{\psi_q}+2^{\frac{1}{q}-1}\bigg\|\max_{1 \leq i \leq n} \sup_{f \in \cF_n}\abs{\EE_{\mu}\big[f\ind_{\{F_n >T_n\}}\big]}\bigg\|_{\psi_q} \notag \\
     &\stackrel{(b)}{\leq}2^{\frac{1}{q}-1} \left(1+\left\lceil\frac{1}{q} \right\rceil !\mspace{2 mu}(\log 2)^{-\frac 1 q} \right)\bigg\|\max_{1 \leq i \leq n} \sup_{f \in \cF_n}\abs{f\ind_{\{F_n>T_n\}}(X_i)}\bigg\|_{\psi_q} \notag \\
&\leq 2^{\frac{1}{q}-1} \left(1+\left\lceil\frac{1}{q} \right\rceil ! \mspace{2 mu}(\log 2)^{-\frac 1 q} \right)\norm{ M_n}_{\psi_q}, \notag 
\end{align}
where 
   \begin{enumerate}[(a)]
   \item follows from $\norm{W_1+W_2}_{\psi_q} \leq 2^{\frac{1}{q}-1}\big(\norm{W_1}_{\psi_q}+\norm{W_2}_{\psi_q}\big)$ for $0 < q \leq 1$ as shown in \eqref{eq:orliczconsttriangle} below;
   \item is due to $\norm{\EE_{\mu}[W]}_{\psi_q} \leq \left\lceil\frac{1}{q} \right\rceil ! (\log 2)^{-\frac 1 q}  \norm{W}_{\psi_q}$ for $W \sim \mu$
as shown in \eqref{eq:orlicconst2}.
   \end{enumerate}
Using the above bound and \eqref{eq:expectbndrem} in \eqref{Orlicznormtailbnd} and substituting the resulting bound in  \eqref{eq:exporlicznorm} yields 
\begin{align}
    \PP \left(Z_n'' \geq  \frac{z}{4} \right) \leq 2e^{\frac{-z^q}{(c_q\norm{M_n}_{\psi_q})^q}}, \notag 
\end{align}
with $c_q$ as given in \eqref{eq:Orlicznormconstfin}.   
By substituting the above inequality to control the last term in \eqref{eq:sumbndsplit} and using \eqref{eq:talagrandineqappl} to upper bound the first term therein, we obtain
\begin{align}
    \PP \left(Z_n \geq  7 Z^{'(\varepsilon)}_n+z \right)  \leq 2e^{-\big(\frac{9z^2}{128 n \norm{F_n}_{2,\mu}^2} \wedge \frac{z}{228\mspace{2 mu}\overline{M_n}}\big)}+2e^{\frac{-z^q}{(c_q\norm{M_n}_{\psi_q})^q}}.\notag 
\end{align}
To obtain upper bounds on $\overline{M_n}$ and $\norm{M_n}_{\psi_q}$,  we will use the following maximal inequality upper bounding the  Orlicz norm of maximum of finitely many random variables.
\begin{lemma}[Orlicz norm of maxima] \label{Lem:maxineqorliczvar}
 For any $0 <q <\infty$ and $\psi_q(x)=e^{x^q}-1$, we have 
 \begin{align}
     \norm{\max_{1 \leq i \leq n} X_i}_{\psi_q} \leq \big(5\log(1+2n)\big)^{\frac{1}{q}} \max_{1 \leq i \leq n} \norm{X_i}_{\psi_q}. \label{eq:maxineqorlicz}
 \end{align}
\end{lemma}
The proof of \eqref{eq:maxineqorlicz}  is given in Appendix \ref{App:Lem:maxineqorliczvar-proof}, which is similar in spirit to the  proof of \cite[Lemma 2.2.2.]{AVDV-book}, albeit without invoking  convexity of $\psi_q$ (required here since $\psi_q=e^{x^q}-1$  is not convex near origin for $0<q < 1$). 
\medskip

From \eqref{eq:maxineqorlicz}, it follows  that $\norm{M_n}_{\psi_q} \leq  \big(5\log(2n+1)\big)^{ 1/q} \norm{F_n}_{\psi_q}<\infty.$
 Further, by using \eqref{eq:expbndorlicznorm} given in Appendix \ref{Sec:Orlicz-elem-ineq} followed by \eqref{eq:maxineqorlicz}, 
\begin{align}
  &  \overline{M}_n \leq \left\lceil\frac{1}{q} \right \rceil ! \norm{M_n}_{\psi_q} \leq  \left\lceil\frac{1}{q} \right \rceil ! \big(5\log(2n+1)\big)^{\frac 1q} \norm{F_n}_{\psi_q}<\infty. \notag 
\end{align}

To complete the proof, we will show that $Z^{'(\varepsilon)}_n \leq Z^{(\varepsilon)}_n$ for $X^n$ surely. Fix $X^n$ and set $\cI(X^n):=\{i \in [1:n], F_n(X_i) \leq T_n \}$, $\bar \cI(X^n):=\{1,\ldots,n\} \setminus \cI(X^n)$, and $\varepsilon_{\cI(X^n)}:=\{\varepsilon_i,~i \in \cI(X^n)\}$. Then
\begin{align}   Z^{'(\varepsilon)}_n&= \EE_{\varepsilon}\left[\sup_{f \in \cF}\Bigg|\sum_{i=1}^n \varepsilon_i f\ind_{\{F_n \leq T_n\}}(X_i)\Bigg| \right] \notag \\
&=\EE_{\varepsilon_{\cI(X^n)}}\left[\sup_{f \in \cF}\Bigg|\sum_{i \in \cI(X^n)} \varepsilon_i f\ind_{\{F_n \leq T_n\}}(X_i)\Bigg| \right] \notag \\
&\stackrel{(a)}{=}\EE_{\varepsilon_{\cI(X^n)}}\left[\sup_{f \in \cF}\Bigg|\sum_{i \in \cI(X^n)} \varepsilon_i f\ind_{\{F_n \leq T_n\}}(X_i)+ \sum_{i \in \bar{\cI}(X^n)}\EE_{\varepsilon_i}\left[\varepsilon_i f(X_i)\right]\Bigg| \right] \notag \\
&\stackrel{(b)}{\leq}\EE_{\varepsilon}\left[\sup_{f \in \cF}\Bigg|\sum_{i \in \cI(X^n)} \varepsilon_i f\ind_{\{F_n \leq T_n\}}(X_i)+ \sum_{i \in \bar{\cI}(X^n)}\varepsilon_i f(X_i)\Bigg| \right] \notag \\
&=\EE_{\varepsilon}\left[\sup_{f \in \cF}\Bigg|\sum_{i \in \cI(X^n)} \varepsilon_i f(X_i)+ \sum_{i \in \bar{\cI}(X^n)}\varepsilon_i f(X_i)\Bigg| \right] \notag \\
&=Z^{(\varepsilon)}_n, \notag
\end{align}
where $(a)$ is because  $\EE_{\varepsilon_i}\left[\varepsilon_i f(X_i)\right]=\EE_{\varepsilon_i}\left[\varepsilon_i\right] f(X_i)=0$ by independence of $\varepsilon_i$ and $X_i$, and $\EE_{\varepsilon_i}\left[\varepsilon_i\right]=0$, and $(b)$ is via Jensen's inequality. This completes the proof.
\section{Concluding Remarks}
We established exponential deviation inequalities for smoothed plug-in estimators and neural estimators of R\'{e}nyi divergences.  The derived inequalities apply to smoothed KL divergence estimate when the unsmoothed distributions are compactly supported or sub-Gaussian, and to smoothed R\'enyi divergence estimate of order $\alpha \neq 1$  when the unsmoothed distributions are compactly supported. It is worth emphasizing that the  dependence of data dimension and distributional  parameters are specified  in our bounds up to universal constants. However, these dependencies are possibly sub-optimal, and a tighter analysis quantifying their impact maybe warranted in scenarios where parameters are not fixed. For instance,   such a  scenario naturally arises in physical-layer security,  where the dimension (codeword length) scales as a function of sample size (codebook size).  A plausible approach in this regard is to find a suitable class of smooth functions with a better covering entropy to which the supremum in the variational form of R\'{e}nyi divergences can be restricted to. Another interesting direction to pursue is deviation inequalities for  kernel density estimators via the variational approach.  We hope that our results will be a starting point for further investigations along these lines.
\appendix

\section{Variational Expression for R\'{e}nyi Divergences}\label{Sec:lem:varexpren-proof}
The following variational expression holds for R\'{e}nyi divergences similar in spirit to the expression \eqref{eq:varexpkl} for KL divergence, which has the advantage of linearizing the second term in \eqref{eq:rendonvarform}.
\begin{lemma}[Variational expression for R\'{e}nyi divergence]\label{lem:varexpren}
Let $\alpha \in (0,1) \cup (1,\infty)$. Then
\begin{align}
    \rendiv{\mu}{\nu}{\alpha}= \sup_{ f} \frac{\alpha}{\alpha-1}\log \left(\int_{\RR^d} e^{(\alpha-1)f} d\mu\right)-\int_{\RR^d} e^{\alpha f} d\nu+1, \label{eq:eqvanrbndren}
\end{align}
where the supremum is over all measurable functions such that the second integral is finite.  When $\mu \ll \nu$ and $\left(\frac{d\mu}{d\nu}\right)^{\alpha} \in L_1(\nu)$, the supremum is achieved by $f_{\alpha}^{\star}=\log \frac{d\mu}{d\nu}+\frac{1-\alpha}{\alpha}\rendiv{\mu}{\nu}{\alpha}$. 
\end{lemma}
\begin{proof}
We start with the variational expression given in \cite{Birell-2021}:
\begin{align}
     \rendiv{\mu}{\nu}{\alpha}=\sup_{f} \frac{\alpha}{\alpha-1} \log \left(\int_{\RR^d} e^{(\alpha-1)f} d\mu\right)-\log \left(\int_{\RR^d} e^{\alpha f} d\nu\right), \label{eq:varexpren}
\end{align}
where the supremum is taken over the set of all bounded  continuous functions. 
Substituting $g=e^f$, we obtain
\begin{align}
     \rendiv{\mu}{\nu}{\alpha}=\sup_{g > 0} \frac{\alpha}{\alpha-1} \log \left(\int_{\RR^d} g^{(\alpha-1)} d\mu\right)-\log \left(\int_{\RR^d} g^{\alpha} d\nu\right), \notag
\end{align}
where the supremum is over bounded continuous functions $g>0$. 
 Note that the expression on the RHS is invariant to scaling the function $g$ by a positive constant $c$. Hence, we have
\begin{align}
     \rendiv{\mu}{\nu}{\alpha}=\sup_{ g > 0} \frac{\alpha}{\alpha-1} \log \left(\int_{\RR^d} (cg)^{(\alpha-1)} d\mu\right)-\log \left(\int_{\RR^d} (cg)^{\alpha} d\nu\right). \notag
\end{align} 
Since $\int_{\RR^d} g^{\alpha} d\nu>0$ for any $g>0$, we may choose $c>0$ such that $\int_{\RR^d} (cg)^{\alpha} d\nu=1$. Then
\begin{align}
     \rendiv{\mu}{\nu}{\alpha}&=\sup_{ \substack{g > 0:\\ \int_{\RR^d} g^{\alpha} d\nu=1}} \frac{\alpha}{\alpha-1} \log \left(\int_{\RR^d} g^{(\alpha-1)} d\mu\right)-\log \left(\int_{\RR^d} g^{\alpha} d\nu\right) \notag \\
     &=\sup_{ \substack{g > 0:\\ \int_{\RR^d} g^{\alpha} d\nu=1}} \frac{\alpha}{\alpha-1} \log \left(\int_{\RR^d} g^{(\alpha-1)} d\mu\right)-\int_{\RR^d} g^{\alpha} d\nu +1 \notag \\
     & \leq \sup_{ g > 0} \frac{\alpha}{\alpha-1} \log \left(\int_{\RR^d} g^{(\alpha-1)} d\mu\right)-\int_{\RR^d} g^{\alpha} d\nu +1 \notag \\
          &  \leq  \sup_{f} \frac{\alpha}{\alpha-1} \log \left(\int_{\RR^d} e^{(\alpha-1)f} d\mu\right)-\int_{\RR^d} e^{\alpha f} d\nu +1, \notag 
\end{align}
where the supremum is over all bounded continuous functions $f$.
On the other hand, by applying  $\log x \leq x-1$ for  $x \geq 0$ to the last term in \eqref{eq:varexpren}, we obtain the reverse inequality, thus proving \eqref{eq:eqvanrbndren}. When $\mu \ll \nu$ and $\left(\frac{d\mu}{d\nu}\right)^{\alpha} \in L_1(\nu)$, $\rendiv{\mu}{\nu}{\alpha}<\infty$ and the integrals in \eqref{eq:eqvanrbndren} are finite. Moreover, it can be seen by substitution that the supremum in \eqref{eq:eqvanrbndren} is achieved by $f_{\alpha}^{\star}=\log \frac{d\mu}{d\nu}+\frac{1-\alpha}{\alpha}\rendiv{\mu}{\nu}{\alpha}$. 
\end{proof}

\section{Determination of Prefactor in Theorem \ref{Thm:coventconst}}\label{App:Thm:coventconst-proof}
The claim in \eqref{eq:coventholder} will follow by determining the constants appearing in the proof of \cite[Theorem 2.7.1]{AVDV-book}. 
It is sufficient to consider the case $L=1$ since the claim for general $L$ follows by scaling. 
Let $\delta=\epsilon^{\frac 1 \beta}$, $\norm{\cdot}$ denote the Euclidean norm on $\RR^d$, and  $\{x_{(i)}\}_{i=1}^m$, with $m=N(\delta,\BB_d(r),\norm{\cdot})$, denote the set of points in the interior of $\BB_d(r)$ such that balls of radius $\delta$ centered at these points cover $\BB_d(r)$. 
We will use the fact that $m \leq \left(1+\frac{2r}{\delta}\right)^d$ for $\delta < r$ and $m=1$ for $\delta \geq r$, which implies that $m \leq 1 \vee \left(\frac{3r}{\delta}\right)^d$.   Recall the differential operator $D^k$ for a $d$-dimensional vector $k$ with non-negative integers as elements (see Definition  \ref{def:Holderclass}). Set $\frac{x^k}{k!}:=\prod_{i=1}^{d} \frac{x_i^{k_i}}{k_i!}$. 
For $k$ with $\abs{k} \leq \ubar{\beta}$, and for each $f$, consider the  vector 
\begin{align}
    A_k f=\left(\left\lfloor \frac{D^k f\big(x_{(1)}\big)}{\delta^{\beta-|k|}} \right\rfloor,\cdots, \left\lfloor \frac{D^k f\big(x_{(m)}\big)}{\delta^{\beta-|k|}} \right\rfloor\right). \notag
\end{align}
For any $x \in \BB_d(r)$, there exists an $x_{(i)}$ for some $1 \leq i \leq m$ such that $\norm{x-x_{(i)}} \leq \delta$. 
By multivariate Taylor's theorem with remainder estimates given in \cite[Equation 8]{Folland-1990}, for $f,g \in C_1^{\beta}\big(\BB_d(r)\big)$ and $x \in \mathrm{int}\big(\BB_d(r)\big)$, we have 
\begin{align}
    (f-g)(x)=\sum_{|k| \leq \ubar{\beta}} D^k(f-g)\big(x_{(i)}\big)\frac{(x-x_{(i)})^k}{k!}+R_{\beta}\big(f-g,x_{(i)},x-x_{(i)}\big), \label{eq:taylorexp}
\end{align}
where 
\begin{align}
  R_{\beta}(f,x_{(i)},x-x_{(i)}) = \ubar{\beta} \sum_{|k|=\ubar{\beta}} \frac{\big(x-x_{(i)}\big)^k}{k!} \int_0^1 (1-t)^{\ubar{\beta}-1} \Big(D^k f\big((1-t)x_{(i)}+tx\big)-D^k f\big(x_{(i)}\big)\Big)dt. \label{eq:taylorexprem}
\end{align}
By definition of $C_1^{\beta}\big(\BB_d(r)\big)$, for $k$ such that $\abs{k}=\ubar{\beta}$, 
\begin{align}
   \abs{ D^k f\big((1-t)x_{(i)}+tx\big)-D^k f\big(x_{(i)}\big)} \leq \norm{x-x_{(i)}}^{\beta-\ubar{\beta}}. \notag
\end{align}
Then, we obtain via triangle inequality that
\begin{align}
   \abs{ D^k f\big((1-t)x_{(i)}+tx\big)-D^k f\big(x_{(i)}\big)-D^k g\big((1-t)x_{(i)}+tx\big)+D^k g\big(x_{(i)}\big)} \leq 2\norm{x-x_{(i)}}^{\beta-\ubar{\beta}}. \notag
\end{align}
This yields as shown in the last equation in \cite{Folland-1990} that
\begin{align}
    \abs{R_{\beta}(f-g,x_{(i)},x-x_{(i)})} \leq 2\frac{\norm{x-x_{(i)}}_1^{\ubar{\beta}}\norm{x-x_{(i)}}^{\beta-\ubar{\beta}}}{\ubar{\beta} !} \leq 2d^{\frac{\ubar{\beta}}{2 }}\frac{\norm{x-x_{(i)}}^{\beta}}{\ubar{\beta} !} \leq \frac{2d^{\frac{\ubar{\beta}}{2 }}\delta^{\beta}}{\ubar{\beta} !}. \notag
\end{align}
Thus, if $A_kf =A_k g$ for all $k$ such that $\abs{k} \leq \ubar{\beta}$, we obtain from \eqref{eq:taylorexp}-\eqref{eq:taylorexprem} that
\begin{align}
  \abs{f(x)-g(x)} \leq 2 \sum_{|k| \leq \ubar{\beta}} d^{\frac{|k|}{2 }}\left(\delta^{\beta-|k|} \frac{\delta^{|k|}}{k!}\right) +\frac{2d^{\frac{\ubar{\beta}}{2 }}\delta^{\beta}}{\ubar{\beta} !} \leq  2 d^{\frac{\ubar{\beta}}{2 }}\delta^{\beta} \left(e^d+\frac{1}{\ubar{\beta} !}\right)\leq  2 d^{\frac{\ubar{\beta}}{2 }} \left(e^d+\frac{1}{\ubar{\beta} !}\right)\epsilon. \notag 
\end{align}
Setting $\hat c_{d,\beta}:=2 d^{\frac{\ubar{\beta}}{2 }} \left(e^d+1\right)$, it follows  that the covering number $N\big(\epsilon \mspace{2 mu}\hat c_{d,\beta}, C_1^{\beta}\big(\BB_d(r)\big),\norm{\cdot}_{\infty}\big) $ is upper bounded by the number of different matrices $Af$ as $f$ ranges over $C_1^{\beta}\big(\BB_d(r)\big)$, where $Af$ is formed by stacking $A_k f$ for $\abs{k} \leq \ubar{\beta}$ into different rows. Since the number of rows is  upper bounded by $(\ubar{\beta}+1)^d$ and the number of possible values of each entry in the row $A_kf$ is bounded by $2/\delta^{\beta-|k|} \vee 1$ because $\max_{1 \leq i \leq m }\abs{D^kf(x_i)} \leq 1$, each column of such matrices can at most have $1 \vee (2\delta^{-\beta})^{(\ubar{\beta}+1)^d}$ values.

Next, assume without loss of generality that  $\{x_{(i)}\}_{i=1}^m$ is enumerated such that for each $j>1$, there is an index $i <j$ such that $\norm{x_{(i)}-x_{(j)}}\leq 2\delta$. This follows by considering a sorting algorithm that iteratively builds the desired enumerated list by choosing a point (from the remaining unlisted points) and adding them at the rightmost possible location in the list such that the aforementioned property holds at each iteration. The feasibility of this procedure at each step is guaranteed. Otherwise the union of $\delta$ balls centered at points in the list would be disjoint with a strictly positive gap from the union of $\delta$ balls centered at points outside the list, thus contradicting the assumption that $\{x_{(i)}\}_{i=1}^m$ forms a $\delta$-covering of $\BB_d(r)$.

Fixing the first column of $Af(x_{(1)})$, the smoothness properties of $f$ enforces limitations on the values the other columns may take. To see this, for each column indexed by $j>1$, choose an $i<j$ such that $\norm{x_{(i)}-x_{(j)}}\leq 2\delta$. By Taylor's theorem,
\begin{align}
    D^kf\big(x_{(j)}\big)=\sum_{ \abs{l} \leq \ubar{\beta}-|k|} D^{k+l} f\big(x_{(i)}\big)\frac{\big(x_{(j)}-x_{(i)}\big)^l}{l!}+R_{\beta-\abs{k}}\big(f,x_{(i)},x_{(j)}-x_{(i)}\big), \label{eq:taylorexp2}
\end{align}
with
\begin{align}
 & R_{\beta-\abs{k}}\big(f,x_{(i)},x_{(j)}-x_{(i)}\big)\notag \\
 &=  (\ubar{\beta}-\abs{k}) \sum_{|l|=\ubar{\beta}-\abs{k}} \frac{\big(x_{(j)}-x_{(i)}\big)^l}{l!} \int_0^1 (1-t)^{\ubar{\beta}-\abs{k}-1} \Big(D^{k+l} f\big((1-t)x_{(i)}+tx_{(j)}\big)-D^{k+l} f\big(x_{(i)}\big)\Big)dt. \notag
\end{align}
Let $B_{k}f:= \delta^{\beta-|k|} A_kf$ and $Bf$ be the matrix obtained by stacking $B_{k}f$ as rows. Suppose $A f\big(x_{(i)}\big)=A g\big(x_{(i)}\big)$. Then  $B f\big(x_{(i)}\big)=B g\big(x_{(i)}\big)$. Hence
\begin{align}
 & \abs{D^kf\big(x_{(j)}\big)-D^kg\big(x_{(j)}\big)} \notag \\
 & \stackrel{(a)}{\leq} \mspace{-2 mu}  \abs{D^kf\big(x_{(j)}\big)\mspace{-2 mu}-\mspace{-2 mu}\sum_{ \abs{l} \leq \ubar{\beta}-|k|}\mspace{-2 mu} B^{k+l} f\big(x_{(i)}\big)\frac{\big(x_{(j)}-x_{(i)}\big)^l}{l!}}\mspace{-2 mu}+\mspace{-2 mu}\abs{D^kg\big(x_{(j)}\big)\mspace{-2 mu}-\mspace{-2 mu}\sum_{ \abs{l} \leq \ubar{\beta}-|k|} \mspace{-2 mu}B^{k+l} g\big(x_{(i)}\big)\frac{\big(x_{(j)}-x_{(i)}\big)^l}{l!}} \notag \\
 & \stackrel{(b)}{\leq}\mspace{-4 mu}  \sum_{ \abs{l} \leq \ubar{\beta}-|k|}\mspace{-2 mu}\abs{D^{k+l}f\big(x_{(i)}\big)- B^{k+l} f\big(x_{(i)}\big)}\frac{\abs{\big(x_{(j)}-x_{(i)}\big)^l}}{l!}\mspace{-2 mu}+\mspace{-2 mu}\abs{D^{k+l}g\big(x_{(i)}\big)- B^{k+l} g\big(x_{(i)}\big)}\frac{\abs{\big(x_{(j)}-x_{(i)}\big)^l}}{l!} \notag \\
 & \leq 2^{\beta-|k|+1} \sum_{ \abs{l} \leq \ubar{\beta}-|k|}d^{\frac{|l|}{2}}\delta^{\beta-|k|-|l|} \frac{\delta^{|l|}}{l!}+R_{\beta-\abs{k}}\big(f,x_{(i)},x_{(j)}-x_{(i)}\big)+R_{\beta-\abs{k}}\big(g,x_{(i)},x_{(j)}-x_{(i)}\big) \notag \\
 & \leq 2^{\beta-|k|+1} d^{\frac{\ubar{\beta}-|k|}{2}}\delta^{\beta-|k|}e^d+2^{\beta-|k|+1} d^{\frac{\ubar{\beta}-|k|}{2}}\frac{\delta^{\beta-|k|}}{(\ubar{\beta}-|k|)!} \notag \\
&  \leq 2^{\beta-|k|}\hat c_{d,\beta-|k|}\delta^{\beta-|k|}, \notag
\end{align}
where $(a)$ follows via triangle inequality and $(b)$ is by substitution of \eqref{eq:taylorexp2}. This means that fixing the values of the $i^{th}$ column of $Af$, each value in the $j^{th}$ column of $Af$ ranges over integers in an interval of length at most $2^{\beta-|k|}\hat c_{d,\beta-|k|}+1$. Using the upper bound $1 \vee (2\delta^{-\beta})^{(\ubar{\beta}+1)^d}$ for the number of different values in the first column of $Af$, it follows that the number of possible matrices $Af$ and hence the covering number is upper bounded as
\begin{align}
N\big(\epsilon \mspace{2 mu}\hat c_{d,\beta}, C_1^{\beta}\big(\BB_d(r)\big),\norm{\cdot}_{\infty}\big)  \leq     1 \vee (2\epsilon^{-1})^{(\ubar{\beta}+1)^d}~\big(2^{\beta}\hat c_{d,\beta}+1\big)^{(m-1)(\ubar{\beta}+1)^d}.\notag 
\end{align}
Taking logarithms, we obtain
\begin{align}
\log N\big(\epsilon \mspace{2 mu}\hat c_{d,\beta}, C_1^{\beta}\big(\BB_d(r)\big),\norm{\cdot}_{\infty}\big)  &\leq 0 \vee   (\ubar{\beta}+1)^d \frac{\beta}{d}  \log \left(2^{\frac{d}{\beta}}\epsilon^{-\frac{d}{\beta}}\right)+(m-1)(\ubar{\beta}+1)^d \log \big(2^{\beta}\hat c_{d,\beta}+1\big) \notag \\
&\leq    (\ubar{\beta}+1)^d \frac{\beta}{d}  2^{\frac{d}{\beta}}\epsilon^{-\frac{d}{\beta}}+\left(\frac{3r}{\delta}\right)^d(\ubar{\beta}+1)^d \log \big(2^{\beta}\hat c_{d,\beta}+1\big) \notag \\
& =: (\ubar{\beta}+1)^d\left(\frac{\beta 2^{\frac{d}{\beta}}}{d}+(3r)^d \log \big(2^{\beta}\hat c_{d,\beta}+1\big)\right)\epsilon^{-\frac{d}{\beta}}.\notag 
\end{align}
Hence, by scaling $\epsilon$, \eqref{eq:coventholder}  follows with
\begin{align}
   c_{d,\beta,\BB_d(r)} &:=(\ubar{\beta}+1)^d\hat c_{d,\beta}^{\frac{d}{\beta}}\left(\frac{\beta 2^{\frac{d}{\beta}}}{d}+(3r)^d \log \big(2^{\beta}\hat c_{d,\beta}+1\big)\right) \notag \\
   &=(\ubar{\beta}+1)^d 2^{\frac d \beta} d^{\frac d 2} (e^d+1)^{\frac{d}{ \beta}}\left(\frac{\beta 2^{\frac{d}{\beta}}}{d}+(3r)^d \log \Big(2^{\beta+1} d^{\frac {\beta}{2}} (e^d+1)+1\Big)\right) \notag \\
   &= \underbrace{(\ubar{\beta}+1)^d 2^{\frac d \beta} d^{\frac d 2} (e^d+1)^{\frac{d}{ \beta}}\left(\frac{\beta 2^{\frac{d}{\beta}}}{d} \vee 3^d \log \Big(2^{\beta+1} d^{\frac {\beta}{2}} (e^d+1)+1\Big)\right)}_{c_{d,\beta}}(1+r^d)\label{eq:constcoventbnd}
\end{align}
Given \eqref{eq:coventholder}, the upper bound in \eqref{eq:coventholderarbmeas} follows via similar steps given in \cite[Corollary 2.7.2]{AVDV-book}, which deals with bracketing entropy in place of covering entropy. 
\section{Auxiliary Results Related to Orlicz Norm}
\subsection{Elementary Inequalities Involving Orlicz Norm}\label{Sec:Orlicz-elem-ineq}
Let $0 <q \leq 1$ and $W_1,W_2$ be Banach space valued random variables. We will show that
\begin{align}
    \norm{W_1+W_2}_{\psi_q} \leq 2^{\frac{1}{q}-1}\big(\norm{W_1}_{\psi_q}+\norm{W_2}_{\psi_q}\big). \label{eq:orliczconsttriangle}
\end{align}
 Note that for $a,b\geq 0$ and $0 \leq q \leq 1$, we have
 \begin{align}
   2^{q-1}(a^q+b^q) \leq (a+b)^q = \frac{a+b}{(a+b)^{1-q}} \leq \frac{a}{(a+b)^{1-q}}+\frac{b}{(a+b)^{1-q}} \leq a^q+b^q, \notag  
 \end{align}
where the first inequality follows from Jensen's inequality applied to the concave map $x \mapsto x^q$. Denoting by $\norm{\cdot}_{\mathfrak{B}}$ the relevant Banach space norm and setting  $a=\norm{W_1}_{\mathfrak{B}}$, $b=\norm{W_2}_{\mathfrak{B}}$, $\bar a=\norm{W_1}_{\psi_q}$, $\bar b=\norm{W_2}_{\psi_q}$ and $\bar c_q=2^{\frac{1}{q}-1}$,  we obtain by using the above inequalities that
    \begin{align}
\EE\left[e^{\left(\frac{\norm{W_1+W_2}_{\mathfrak{B}}}{\bar c_q(\bar a+\bar b)}\right)^q}\right] \leq \EE\left[e^{\frac{(a+b)^q}{\bar c_q^q(\bar a+\bar b)^q}}\right]\leq \EE\left[e^{\frac{a^q+b^q}{\bar c_q^q(\bar a+\bar b)^q}}\right] \leq \EE\left[e^{\frac{a^q+b^q}{\bar c_q^q 2^{q-1}(\bar a^q+\bar b^q)}}\right]=\EE\left[e^{\frac{a^q+b^q}{\bar a^q+\bar b^q}}\right]. \label{eq:orliczconst1}
    \end{align}
Continuing,  we have
\begin{align}
  \EE\left[e^{\frac{a^q+b^q}{\bar a^q+\bar b^q}}\right]  \stackrel{(a)}{\leq} \frac{\bar a^q}{\bar a^q+\bar b^q}\EE\left[e^{\left(\frac{\norm{W_1}_{\mathfrak{B}}}{\bar a}\right)^q}\right] +\frac{\bar b^q}{\bar a^q+\bar b^q}\EE\left[e^{\left(\frac{\norm{W_2}_{\mathfrak{B}}}{\bar b}\right)^q}\right] \stackrel{(b)}{\leq}\frac{2\bar a^q}{\bar a^q+\bar b^q}+\frac{2\bar b^q}{\bar a^q+\bar b^q} \leq 2, \label{eq:orliczconst2}
\end{align}
where $(a)$ is via Jensen's inequality applied to the convex map $x \mapsto e^x$ and $(b)$ follows from the definition of Orlicz norm $\norm{\cdot}_{\psi_q}$ which implies 
\begin{align}
\EE\left[e^{\left(\frac{\norm{W_1}_{\mathfrak{B}}}{\bar a}\right)^q}\right] \vee \EE\left[e^{\left(\frac{\norm{W_2}_{\mathfrak{B}}}{\bar b}\right)^q}\right] \leq 2. \notag  
\end{align} 
Using the definition of Orlicz norm $\norm{\cdot}_{\psi_q}$ again, we have $  \norm{W_1+W_2}_{\psi_q} \leq c_q(\bar a+\bar b)$ since combining \eqref{eq:orliczconst1}-\eqref{eq:orliczconst2} yields
  \begin{align}
\EE\left[e^{\left(\frac{\norm{W_1+W_2}_{\mathfrak{B}}}{c_q(\bar a+\bar b)}\right)^q}\right] \leq 2.
\notag  
\end{align} 

\medskip

Next, we show that for a Banach space valued random variable $W \sim \mu$   such that $\EE_{\mu}[W]$ exists (as a Bochner integral), 
\begin{align}
\norm{\EE_{\mu}[W]}_{\psi_q} \leq \left\lceil\frac{1}{q} \right \rceil !~(\log 2)^{-\frac 1q}  \norm{W}_{\psi_q}.  \label{eq:orlicconst2}
\end{align}
We may assume that $\norm{\EE_{\mu}[W]}_{\psi_q}>0$ since otherwise \eqref{eq:orlicconst2} holds trivially. In this case, by continuity of $\psi_q$ and the definition of Orlicz norm $\norm{\cdot}_{\psi_q}$, we have 
\begin{align}
 e^{\left(\frac{\norm{\EE_{\mu}[W]}_{\mathfrak{B}}}{\norm{\EE_{\mu}[W]}_{\psi_q}}\right)^q} = 2.   \notag
\end{align}
This implies that
\begin{align}
    \norm{\EE_{\mu}[W]}_{\psi_q} = (\log 2)^{-\frac 1q} \norm{\EE_{\mu}[W]}_{\mathfrak{B}} \leq \left\lceil\frac{1}{q} \right \rceil ! (\log 2)^{-\frac 1q} \norm{W}_{\psi_q}, \notag
\end{align}
where in the last inequality, we used 
\begin{align}
  \norm{\EE_{\mu}[W]}_{\mathfrak{B}} \leq \EE_{\mu}\big[\norm{W}_{\mathfrak{B}}\big] \leq \left\lceil\frac{1}{q} \right \rceil ! \norm{W}_{\psi_q}. \label{eq:expbndorlicznorm}  
\end{align}
 The first inequality here, in turn,  follows from properties of Bochner integral and the second inequality is a consequence of the fact that for a non-negative valued random variable $X$, 
\begin{align}
 \EE_{\mu}\big[X\big] \leq \EE_{\mu}\big[X^{q \left\lceil\frac{1}{q} \right \rceil}\big] \leq \left\lceil\frac{1}{q} \right \rceil !\mspace{2 mu}\EE\big[ e^{X^{q}}-1 \big]. \notag 
\end{align}
\subsection{Proof of Lemma \ref{Lem:maxineqorliczvar}} \label{App:Lem:maxineqorliczvar-proof}
We will follow the proof given in \cite[Lemma 2.2.2.]{AVDV-book}, by avoiding arguments that require convexity of $\psi$.  For $\psi_q(x)=e^{x^q}-1$, let $x^{\star}_q=\big(\log (3/2)\big)^{1/q}$.  Observe that for $x,y \geq x^{\star}_q$ and $c_q' \geq \big(2/\log(3/2)\big)^{1/q}$, we have $\psi_q(x)\psi_q(y) \leq \psi_q(c_q'xy)$. Consequently, for any $y \geq x^{\star}_q$ and $x \geq x^{\star}_q y$, 
\begin{align}
   \psi_q\left(\frac{x}{y}\right) \leq \frac{\psi_q(c_q'x)}{\psi_q(y)}. \label{bndratioorlicnorm} 
\end{align}
Then, for  any  $y \geq x^{\star}_q$ and $C>0$, we have 
\begin{align}
    \max_i \psi_q\left(\frac{\abs{X_i}}{Cy}\right) &=  \max_i \left(\psi_q\left(\frac{\abs{X_i}} {Cy}\right) \ind_{\big\{\abs{X_i} \geq Cy x^{\star}_q \big\} }+ \psi_q\left(\frac{\abs{X_i}} {Cy}\right) \ind_{\big\{\abs{X_i} < Cy x^{\star}_q \big\}} \right) \notag \\
    & \stackrel{(a)}{\leq}  \max_i \left(\frac{\psi_q\left(\frac{c_q'\abs{X_i}} {C}\right)}{\psi_q(y)} + \psi_q\left(\frac{\abs{X_i}} {Cy}\right) \ind_{\big\{\abs{X_i} < Cy x^{\star}_q \big\}} \right) \notag \\
    & \stackrel{(b)}{\leq} \max_i\frac{\psi_q\left(\frac{c_q'\abs{X_i}} {C}\right)}{\psi_q(y)} +  \psi_q(x^{\star}_q)\notag \\
    & \stackrel{(c)}{\leq} \sum_{i}^n \frac{\psi_q\left(\frac{c_q'\abs{X_i}} {C}\right)}{\psi_q(y)}+\frac 12,\notag
\end{align}
where in $(a)$, we applied \eqref{bndratioorlicnorm} with $x=\abs{X_i}/C$ under the event that $\abs{X_i} \geq Cyx^{\star}_q$, $(b)$ used that $\psi_q$ is monotone increasing, and in $(c)$ follows by upper bounding $\max$ by sum, and $\psi_q(x^{\star}_q)=0.5$. Setting $C=c_q' \max_i \norm{X_i}_{\psi_q}$ and taking expectation, we obtain for $y \geq x^{\star}_q$
\begin{align}
 \EE\left[  \psi_q\left(\frac{\max_i\abs{X_i}}{c_q' \max_i \norm{X_i}_{\psi_q}y}\right)\right]   \leq \frac{n}{\psi_q(y)}+\frac 12. \notag
\end{align}
Choosing $y=\psi_q^{-1}(2n) \geq x^{\star}_q $, the RHS is less than 1, thus implying that (by choosing $c_q'=\big(2/\log(3/2)\big)^{1/q} \leq 5^{\frac 1q}$)
\begin{align*}
    \norm{\max_i\abs{X_i}}_{\psi_q} \leq 5^{\frac 1q}\psi_q^{-1}(2n) \max_i \norm{X_i}_{\psi_q} \leq  \big(5\log(1+2n)\big)^{\frac{1}{q}} \max_i \norm{X_i}_{\psi_q}.
\end{align*}

\bibliographystyle{abbrv}
\bibliography{ref}

\begin{thebibliography}{10}

\bibitem{Jayadev-2019}
J.~Acharya, S.~Bhadane, P.~Indyk, and Z.~Sun.
\newblock Estimating entropy of distributions in constant space.
\newblock In {\em Advances in Neural Information Processing Systems}, volume~32. Curran Associates, Inc., 2019.

\bibitem{Adamczak-2008}
R.~Adamczak.
\newblock {A tail inequality for suprema of unbounded empirical processes with applications to Markov chains}.
\newblock {\em Electronic Journal of Probability}, 13:1000 -- 1034, 2008.

\bibitem{Adamczak-2010}
R.~Adamczak.
\newblock A few remarks on the operator norm of random {Toeplitz} matrices.
\newblock {\em Journal of Theoretical Probability}, 23:85 -- 108, 04 2010.

\bibitem{Agrawal-2020}
R.~Agrawal.
\newblock Finite-sample concentration of the multinomial in relative entropy.
\newblock {\em IEEE Transactions on Information Theory}, 66(10):6297--6302, 2020.

\bibitem{Antos-Ioannis-2001}
A.~Antos and I.~Kontoyiannis.
\newblock Convergence properties of functional estimates for discrete distributions.
\newblock {\em Random Structures \& Algorithms}, 19(3-4):163--193, 2001.

\bibitem{Barron_1993}
A.~R. {Barron}.
\newblock Universal approximation bounds for superpositions of a sigmoidal function.
\newblock {\em IEEE Transactions on Information Theory}, 39(3):930--945, May 1993.

\bibitem{Bartlett-2005}
P.~L. Bartlett, O.~Bousquet, and S.~Mendelson.
\newblock {Local Rademacher complexities}.
\newblock {\em The Annals of Statistics}, 33(4):1497 -- 1537, 2005.

\bibitem{belghazi2018}
M.~I. Belghazi, A.~Baratin, S.~Rajeshwar, S.~Ozair, Y.~Bengio, A.~Courville, and D.~Hjelm.
\newblock Mutual information neural estimation.
\newblock In {\em Proceedings of the 35th International Conference on Machine Learning}, volume~80, pages 531--540, Stockholm Sweden, Jul. 2018.

\bibitem{berrett2023-twosampfunc}
T.~B. Berrett and R.~J. Samworth.
\newblock {Efficient functional estimation and the super-oracle phenomenon}.
\newblock {\em The Annals of Statistics}, 51(2):668 -- 690, 2023.

\bibitem{berrett2019-knn}
T.~B. Berrett, R.~J. Samworth, and M.~Yuan.
\newblock Efficient multivariate entropy estimation via $k$-nearest neighbour distances.
\newblock {\em The Annals of Statistics}, 47(1):288--318, Feb. 2019.

\bibitem{Bickel-Ritov-88}
P.~J. Bickel and Y.~Ritov.
\newblock Estimating integrated squared density derivatives: Sharp best order of convergence estimates.
\newblock {\em Sankhyā: The Indian Journal of Statistics, Series A (1961-2002)}, 50(3):381--393, 1988.

\bibitem{Birge-Massart-95}
L.~Birge and P.~Massart.
\newblock Estimation of integral functionals of a density.
\newblock {\em The Annals of Statistics}, 23(1):11 -- 29, 1995.

\bibitem{Birell-2021}
J.~Birrell, P.~Dupuis, M.~A. Katsoulakis, L.~Rey-Bellet, and J.~Wang.
\newblock Variational representations and neural network estimation of {R}ényi divergences.
\newblock {\em SIAM Journal on Mathematics of Data Science}, 3(4):1093--1116, 2021.

\bibitem{AZYA-2022}
A.~Block, Z.~Jia, Y.~Polyanskiy, and A.~Rakhlin.
\newblock Rate of convergence of the smoothed empirical {W}asserstein distance.
\newblock {\em Annales de l’Institut Henri Poincar\'{e}, to appear}, 2025.

\bibitem{BOUSQUET-2002}
O.~Bousquet.
\newblock A {B}ennett concentration inequality and its application to suprema of empirical processes.
\newblock {\em Comptes Rendus Mathematique}, 334(6):495--500, 2002.

\bibitem{Bu-2018}
Y.~Bu, S.~Zou, Y.~Liang, and V.~V. Veeravalli.
\newblock Estimation of {KL} divergence: optimal minimax rate.
\newblock {\em IEEE Transactions on Information Theory}, 64(4):2648--2674, 2018.

\bibitem{Bulinksi-2021}
A.~Bulinski and D.~Dimitrov.
\newblock Statistical estimation of the {Kullback-Leibler} divergence.
\newblock {\em Mathematics}, 9(5):1--36, March 2021.

\bibitem{Cai-2006}
H.~Cai, S.~Kulkarni, and S.~Verdu.
\newblock Universal divergence estimation for finite-alphabet sources.
\newblock {\em IEEE Transactions on Information Theory}, 52(8):3456--3475, 2006.

\bibitem{CoverThomas}
T.~M. Cover and J.~A. Thomas.
\newblock {\em Elements of Information Theory}.
\newblock NewYork: Wiley, 1991.

\bibitem{Darbellay-1999}
G.~A. {Darbellay} and I.~{Vajda}.
\newblock Estimation of the information by an adaptive partitioning of the observation space.
\newblock {\em IEEE Transactions on Information Theory}, 45(4):1315--1321, 1999.

\bibitem{Delatre-Fournier-2017}
S.~Delattre and N.~Fournier.
\newblock On the {Kozachenko–Leonenko} entropy estimator.
\newblock {\em Journal of Statistical Planning and Inference}, 185:69--93, 2017.

\bibitem{DL-1993}
R.~A. DeVore and G.~G. Lorentz.
\newblock {\em Constructive Approximation}.
\newblock Springer Berlin, Heidelberg, 1993.

\bibitem{domingo2022auditing}
C.~Domingo-Enrich and Y.~Mroueh.
\newblock Auditing differential privacy in high dimensions with the kernel quantum {R}\'{e}nyi divergence.
\newblock {\em arXiv:2205.13941}, 2022.

\bibitem{Donsker-1975}
M.~D. Donsker and S.~R.~S. Varadhan.
\newblock Asymptotic evaluation of certain {Markov} process expectations for large time, i.
\newblock {\em Communications on Pure and Applied Mathematics}, 28(1):1--47, 1975.

\bibitem{DKMMN-2006}
C.~Dwork, K.~Kenthapadi, F.~McSherry, I.~Mironov, and M.~Naor.
\newblock Our data, ourselves: privacy via distributed noise generation.
\newblock In {\em Advances in Cryptology - EUROCRYPT 2006}, pages 486--503. Springer Berlin Heidelberg, 2006.

\bibitem{DMKS-2006}
C.~Dwork, F.~McSherry, K.~Nissim, and A.~Smith.
\newblock Calibrating noise to sensitivity in private data analysis.
\newblock In {\em Proceedings of the Third Theory of Cryptography Conference}, pages 265--284. Springer Berlin Heidelberg, 2006.

\bibitem{Dennis-2021}
D.~Elbr\"{a}chter, D.~Perekrestenko, P.~Grohs, and H.~Bölcskei.
\newblock Deep neural network approximation theory.
\newblock {\em IEEE Transactions on Information Theory}, 67(5):2581--2623, 2021.

\bibitem{Folland-1990}
G.~B. Folland.
\newblock Remainder estimates in {Taylor's} theorem.
\newblock {\em The American Mathematical Monthly}, 97(3):233--235, 1990.

\bibitem{Fuk-Nagaev}
D.~K. Fuk and S.~V. Nagaev.
\newblock Probability inequalities for sums of independent random variables.
\newblock {\em Theory of Probability \& Its Applications}, 16(4):643--660, 1971.

\bibitem{Gao-2018}
W.~Gao, S.~Oh, and P.~Viswanath.
\newblock Demystifying fixed $k$ -nearest neighbor information estimators.
\newblock {\em IEEE Transactions on Information Theory}, 64(8):5629--5661, 2018.

\bibitem{Goldfeld-2020-wiretap}
Z.~Goldfeld, P.~Cuff, and H.~H. Permuter.
\newblock Wiretap channels with random states non-causally available at the encoder.
\newblock {\em IEEE Transactions on Information Theory}, 66(3):1497--1519, 2020.

\bibitem{Goldfeld-2020-smoothemp}
Z.~Goldfeld, K.~Greenewald, J.~Niles-Weed, and Y.~Polyanskiy.
\newblock Convergence of smoothed empirical measures with applications to entropy estimation.
\newblock {\em IEEE Transactions on Information Theory}, 66(7):4368--4391, Jul. 2020.

\bibitem{Goria-2005}
M.~N. Goria, N.~N. Leonenko, V.~V. Mergel, and P.~L.~N. Inverardi.
\newblock A new class of random vector entropy estimators and its applications in testing statistical hypotheses.
\newblock {\em Journal of Nonparametric Statistics}, 17(3):277--297, 2005.

\bibitem{Guo-Richardson-2021}
F.~R. Guo and T.~S. Richardson.
\newblock Chernoff-type concentration of empirical probabilities in relative entropy.
\newblock {\em IEEE Transactions on Information Theory}, 67(1):549--558, 2021.

\bibitem{haje2009entropy}
F.~Haje~Hussein and Y.~Golubev.
\newblock On entropy estimation by m-spacing method.
\newblock {\em Journal of Mathematical Sciences}, 163(3), 2009.

\bibitem{Yanjun-2020-renyi}
Y.~Han, J.~Jiao, and T.~Weissman.
\newblock Minimax estimation of divergences between discrete distributions.
\newblock {\em IEEE Journal on Selected Areas in Information Theory}, 1(3):814--823, 2020.

\bibitem{Yanjun-2020}
Y.~Han, J.~Jiao, T.~Weissman, and Y.~Wu.
\newblock {Optimal rates of entropy estimation over Lipschitz balls}.
\newblock {\em The Annals of Statistics}, 48(6):3228--3250, Dec. 2020.

\bibitem{Mina-2024}
M.~M. Hossain, A.~Wisler, and K.~R. Moon.
\newblock Nonparametric estimation of non-smooth divergences.
\newblock In {\em Proceedings of the 33rd ACM International Conference on Information and Knowledge Management}, CIKM '24, page 3787–3791, New York, NY, USA, 2024. Association for Computing Machinery.

\bibitem{jagielski2020auditing}
M.~Jagielski, J.~Ullman, and A.~Oprea.
\newblock Auditing differentially private machine learning: How private is private {SGD}?
\newblock In {\em Proceedings of Advances in Neural Information Processing Systems}, volume~33, pages 22205--22216, 2020.

\bibitem{Jiantao-2015}
J.~Jiao, K.~Venkat, Y.~Han, and T.~Weissman.
\newblock Minimax estimation of functionals of discrete distributions.
\newblock {\em IEEE Transactions on Information Theory}, 61(5):2835--2885, 2015.

\bibitem{jin-2019-subgaussnorm}
C.~Jin, P.~Netrapalli, R.~Ge, S.~M. Kakade, and M.~I. Jordan.
\newblock A short note on concentration inequalities for random vectors with sub{G}aussian norm.
\newblock {\em arXiv:1902.03736}, Feb. 2019.

\bibitem{kandasamy2015nonparametric}
K.~Kandasamy, A.~Krishnamurthy, B.~Poczos, L.~Wasserman, and J.~M. Robins.
\newblock Nonparametric {von Mises} estimators for entropies, divergences and mutual informations.
\newblock In {\em Proceedings of Advances in Neural Information Processing Systems}, volume~28, pages 397--405, Montr{\'e}al, Canada, Dec. 2015.

\bibitem{Kerkyacharian-96}
G.~Kerkyacharian and D.~Picard.
\newblock {Estimating nonquadratic functionals of a density using Haar wavelets}.
\newblock {\em The Annals of Statistics}, 24(2):485 -- 507, 1996.

\bibitem{Klusowski-2018}
J.~M. Klusowski and A.~R. Barron.
\newblock Approximation by combinations of {R}e{LU} and squared {R}e{LU} ridge functions with $\ell^1$ and $\ell^0$ controls.
\newblock {\em IEEE Transactions on Information Theory}, 64(12):7649--7656, Oct. 2018.

\bibitem{Ioannis-Skoularidou-2016}
I.~Kontoyiannis and M.~Skoularidou.
\newblock Estimating the directed information and testing for causality.
\newblock {\em IEEE Transactions on Information Theory}, 62(11):6053--6067, 2016.

\bibitem{Kozachenko-Leonenko}
L.~F. Kozachenko and N.~N. Leonenko.
\newblock Sample estimate of the entropy of a random vector.
\newblock {\em Problems of Information Transmission}, 23(2):9--16, 1987.

\bibitem{Kraskov-2004}
A.~Kraskov, H.~St\"ogbauer, and P.~Grassberger.
\newblock Estimating mutual information.
\newblock {\em Physical Review E}, 69, Jun 2004.

\bibitem{Krishnamurthy-2014}
A.~Krishnamurthy, K.~Kandasamy, B.~P\'{o}czos, and L.~Wasserman.
\newblock Nonparametric estimation of {R\'{e}nyi} divergence and friends.
\newblock In {\em Proceedings of the 31st International Conference on International Conference on Machine Learning - Volume 32}, ICML'14, page II–919–II–927, 2014.

\bibitem{Laurent-96}
B.~Laurent.
\newblock {Efficient estimation of integral functionals of a density}.
\newblock {\em The Annals of Statistics}, 24(2):659 -- 681, 1996.

\bibitem{LT-1991}
M.~Ledoux and M.~Talagrand.
\newblock {\em Probability in Banach Spaces}.
\newblock Springer-Verlag Berlin Heidelberg, 1991.

\bibitem{Leonenko-2008}
N.~Leonenko, L.~Pronzato, and V.~Savani.
\newblock {A class of Rényi information estimators for multidimensional densities}.
\newblock {\em The Annals of Statistics}, 36(5):2153 -- 2182, 2008.

\bibitem{GWC-1978}
S.~Leung-Yan-Cheong and M.~Hellman.
\newblock The {Gaussian} wire-tap channel.
\newblock {\em IEEE Transactions on Information Theory}, 24(4):451--456, 1978.

\bibitem{Liu-2012}
H.~Liu, L.~Wasserman, and J.~Lafferty.
\newblock Exponential concentration for mutual information estimation with application to forests.
\newblock In {\em Advances in Neural Information Processing Systems}, volume~25. Curran Associates, Inc., 2012.

\bibitem{Jingbo-2017}
J.~Liu, P.~Cuff, and S.~Verdú.
\newblock ${E}_{ {\gamma }}$-resolvability.
\newblock {\em IEEE Transactions on Information Theory}, 63(5):2629--2658, 2017.

\bibitem{Mardia-2019}
J.~Mardia, J.~Jiao, E.~Tánczos, R.~D. Nowak, and T.~Weissman.
\newblock Concentration inequalities for the empirical distribution of discrete distributions: beyond the method of types.
\newblock {\em Information and Inference: A Journal of the IMA}, 9(4):813--850, 11 2019.

\bibitem{Mironov-2017}
I.~Mironov.
\newblock Rényi differential privacy.
\newblock In {\em 2017 IEEE 30th Computer Security Foundations Symposium (CSF)}, pages 263--275, 2017.

\bibitem{Moon-2018}
K.~R. Moon, K.~Sricharan, K.~Greenewald, and A.~O. Hero.
\newblock Ensemble estimation of information divergence.
\newblock {\em Entropy}, 20(8), Aug. 2018.

\bibitem{Nguyen-2010}
X.~{Nguyen}, M.~J. {Wainwright}, and M.~I. {Jordan}.
\newblock Estimating divergence functionals and the likelihood ratio by convex risk minimization.
\newblock {\em IEEE Transactions on Information Theory}, 56(11):5847--5861, Oct. 2010.

\bibitem{Noshad-2017}
M.~Noshad, K.~R. Moon, S.~Y. Sekeh, and A.~O. Hero.
\newblock Direct estimation of information divergence using nearest neighbor ratios.
\newblock In {\em Proceedings of the 2017 IEEE International Symposium on Information Theory}, pages 903--907, Jun. 2017.

\bibitem{Pal-2010}
D.~P\'{a}l, B.~P\'{o}czos, and C.~Szepesv\'{a}ri.
\newblock Estimation of {R\'{e}nyi} entropy and mutual information based on generalized nearest-neighbor graphs.
\newblock In {\em Proceedings of the 24th International Conference on Neural Information Processing Systems - Volume 2}, NIPS'10, page 1849–1857, Red Hook, NY, USA, 2010. Curran Associates Inc.

\bibitem{Paninski-2003}
L.~Paninski.
\newblock Estimation of entropy and mutual information.
\newblock {\em Neural Computation}, 15(6):1191–1253, June 2003.

\bibitem{Perez-2008}
F.~{Perez-Cruz}.
\newblock Kullback-{L}eibler divergence estimation of continuous distributions.
\newblock In {\em Proceedings of the 2008 IEEE International Symposium on Information Theory}, pages 1666--1670, Toronto, ON, Canada, Jul. 2008.

\bibitem{Poczos-Schneider-2011}
B.~Poczos and J.~Schneider.
\newblock On the estimation of $\alpha$-divergences.
\newblock In {\em Proceedings of the Fourteenth International Conference on Artificial Intelligence and Statistics}, volume~15 of {\em Proceedings of Machine Learning Research}, pages 609--617, Fort Lauderdale, FL, USA, 11--13 Apr 2011. PMLR.

\bibitem{Renyi-60}
A.~R{\'e}nyi.
\newblock On measures of entropy and information.
\newblock In {\em Proceedings of the Fourth Berkeley Symposium on Mathematical Statistics and Probability}, volume~1, pages 547--561, Berkeley, 1961. University of California Press.

\bibitem{robbins1955}
H.~Robbins.
\newblock A remark on {Stirling's} formula.
\newblock {\em The American Mathematical Monthly}, 62(1):26--29, 1955.

\bibitem{Ryu-2022}
J.~J. Ryu, S.~Ganguly, Y.-H. Kim, Y.-K. Noh, and D.~D. Lee.
\newblock Nearest neighbor density functional estimation from inverse {Laplace} transform.
\newblock {\em IEEE Transactions on Information Theory}, 68(6):3511--3551, 2022.

\bibitem{Schmidt-Hieber-2020}
J.~Schmidt-Hieber.
\newblock {Nonparametric regression using deep neural networks with ReLU activation function}.
\newblock {\em The Annals of Statistics}, 48(4):1875 -- 1897, 2020.

\bibitem{Singh-2003}
H.~Singh, N.~Misra, V.~Hnizdo, A.~Fedorowicz, and E.~Demchuk.
\newblock Nearest neighbor estimates of entropy.
\newblock {\em American Journal of Mathematical and Management Sciences}, 23(3-4):301--321, 2003.

\bibitem{Singh-Poczos-2014}
S.~Singh and B.~P\'{o}czos.
\newblock Exponential concentration of a density functional estimator.
\newblock In {\em Proceedings of Advances in Neural Information Processing Systems}, volume~27, page 3032–3040, Montreal, Canada, Dec. 2014.

\bibitem{pmlr-v32-singh14}
S.~Singh and B.~Poczos.
\newblock Generalized exponential concentration inequality for {R\'enyi} divergence estimation.
\newblock In {\em Proceedings of the 31st International Conference on Machine Learning}, volume 32 (1) of {\em Proceedings of Machine Learning Research}, pages 333--341, Bejing, China, 22--24 Jun 2014. PMLR.

\bibitem{Singh-Poczos-2016}
S.~Singh and B.~P\'{o}czos.
\newblock Finite-sample analysis of fixed-k nearest neighbor density functional estimators.
\newblock In {\em Proceedings of Advances in Neural Information Processing Systems}, volume~29, pages 1225--1233, Barcelona, Spain, Dec. 2016.

\bibitem{Sreekumar-IT-2025}
S.~Sreekumar and M.~Berta.
\newblock Limit distribution theory for quantum divergences.
\newblock {\em IEEE Transactions on Information Theory}, 71(1):459--484, 2025.

\bibitem{SBGPS-2021}
S.~Sreekumar, A.~Bunin, Z.~Goldfeld, H.~H. Permuter, and S.~Shamai.
\newblock The secrecy capacity of cost-constrained wiretap channels.
\newblock {\em IEEE Transactions on Information Theory}, 67(3):1433--1445, 2021.

\bibitem{SG-2021-SC}
S.~Sreekumar and Z.~Goldfeld.
\newblock Soft-covering via constant-composition superposition codes.
\newblock In {\em 2021 IEEE International Symposium on Information Theory (ISIT)}, pages 2876--2881, 2021.

\bibitem{sreekumar2021neural}
S.~Sreekumar and Z.~Goldfeld.
\newblock Neural estimation of statistical divergences.
\newblock {\em Journal of Machine Learning Research}, 23(126):1--75, 2022.

\bibitem{SGK-IT-2023}
S.~Sreekumar, Z.~Goldfeld, and K.~Kato.
\newblock Limit distribution theory for f-divergences.
\newblock {\em IEEE Transactions on Information Theory}, 70(2):1233--1267, 2024.

\bibitem{Sricharan-2012}
K.~Sricharan, R.~Raich, and A.~O. Hero.
\newblock Estimation of nonlinear functionals of densities with confidence.
\newblock {\em IEEE Transactions on Information Theory}, 58(7):4135--4159, 2012.

\bibitem{Tsur-2023}
D.~Tsur, Z.~Aharoni, Z.~Goldfeld, and H.~Permuter.
\newblock Neural estimation and optimization of directed information over continuous spaces.
\newblock {\em IEEE Transactions on Information Theory}, 69(8):4777--4798, 2023.

\bibitem{Tsur-2024}
D.~Tsur, Z.~Aharoni, Z.~Goldfeld, and H.~Permuter.
\newblock Data-driven optimization of directed information over discrete alphabets.
\newblock {\em IEEE Transactions on Information Theory}, 70(3):1652--1670, 2024.

\bibitem{Tsybakov-1996}
A.~B. Tsybakov and E.~C. van~der Meulen.
\newblock Root-n consistent estimators of entropy for densities with unbounded support.
\newblock {\em Scandinavian Journal of Statistics}, 23(1):75--83, 1996.

\bibitem{Valiant-2011}
G.~Valiant and P.~Valiant.
\newblock Estimating the unseen: an n/log(n)-sample estimator for entropy and support size, shown optimal via new {CLTs}.
\newblock In {\em Proceedings of the Forty-Third Annual ACM Symposium on Theory of Computing}, STOC '11, page 685–694, New York, NY, USA, 2011. Association for Computing Machinery.

\bibitem{Valiant-2017}
G.~Valiant and P.~Valiant.
\newblock Estimating the unseen: Improved estimators for entropy and other properties.
\newblock {\em J. ACM}, 64(6), Oct. 2017.

\bibitem{AVDV-book}
A.~W. van~der Vaart and J.~A. Wellner.
\newblock {\em Weak Convergence and Empirical Processes}.
\newblock Springer, New York, 1996.

\bibitem{vanEvren_Reyni_Div2014}
T.~van Erven and P.~Harremo{\"e}s.
\newblock {R}{\'e}nyi divergence and {Kullback-Leibler} divergence.
\newblock {\em IEEE Transactions on Information Theory}, 60(7):3797--3820, Jul. 2014.

\bibitem{Wang-2005}
Q.~Wang, S.~R. {Kulkarni}, and S.~{Verdu}.
\newblock Divergence estimation of continuous distributions based on data-dependent partitions.
\newblock {\em IEEE Transactions on Information Theory}, 51(9):3064--3074, Sep. 2005.

\bibitem{Wisler-2018}
A.~Wisler, K.~Moon, and V.~Berisha.
\newblock Direct ensemble estimation of density functionals.
\newblock In {\em Proceedings of the 2018 IEEE International Conference on Acoustics, Speech and Signal Processing}, pages 2866--2870, Apr. 2018.

\bibitem{Yihong-Yang-2016}
Y.~Wu and P.~Yang.
\newblock Minimax rates of entropy estimation on large alphabets via best polynomial approximation.
\newblock {\em IEEE Transactions on Information Theory}, 62(6):3702--3720, 2016.

\bibitem{YAROTSKY-2017}
D.~Yarotsky.
\newblock Error bounds for approximations with deep {R}e{LU} networks.
\newblock {\em Neural Networks}, 94:103--114, Oct. 2017.

\bibitem{Zhao-2020}
P.~Zhao and L.~Lai.
\newblock Minimax optimal estimation of {KL} divergence for continuous distributions.
\newblock {\em IEEE Transactions on Information Theory}, 66(12):7787--7811, 2020.

\end{thebibliography}

\end{document}